\documentclass[a4paper, reqno, 11pt]{amsart}
\usepackage{color}
\usepackage{amsmath}
\usepackage{amsthm}
\usepackage{amssymb}
\usepackage{amsmath}
\usepackage{mathcomp}
\usepackage{dsfont}
\usepackage[toc,page]{appendix}
\usepackage{graphicx} 
\usepackage{subfigure}
\usepackage{enumerate}
\usepackage{amsfonts}
\usepackage{amsthm}
\usepackage{amsmath}
\usepackage{times}
\usepackage{amscd}
\usepackage[utf8]{inputenc}
\usepackage{t1enc}
\usepackage{enumitem}
\usepackage[mathscr]{eucal}
\usepackage{indentfirst}
\usepackage{graphicx}
\usepackage{graphics}
\usepackage{pict2e}
\usepackage{epic}
\usepackage{esint}
\usepackage{epstopdf} 
\usepackage{verbatim}
\usepackage{cancel}
\usepackage{amssymb}
\usepackage{xcolor}
\usepackage{dsfont}
\usepackage{hyperref}

\allowdisplaybreaks

\usepackage[margin=30mm]{geometry}
    
\newtheorem{theorem}{Theorem}[section]
\newtheorem{lemma}[theorem]{Lemma}
\newtheorem{proposition}[theorem]{Proposition}

\newtheorem{corollary}[theorem]{Corollary}

\theoremstyle{definition}
 
\newtheorem{definition}{Definition}[section]
\newtheorem{remark}[definition]{Remark}

\numberwithin{equation}{section}

\def\cN{\mathcal{N}}

\def\cF{\mathcal{F}}

\def\cG{\mathcal{G}}

\def\wN{\widetilde{\mathcal{N}}_+}

\def\b1{\mathds{1}}
\def\1{\mathds{1}}
\def\Tr{\mathrm{Tr}}
\def\Re{\mathrm{Re}}

\def\a0{\mathfrak{a}_0}

\newcommand{\vertiii}[1]{{\left\vert\kern-0.25ex\left\vert\kern-0.25ex\left\vert #1 
    \right\vert\kern-0.25ex\right\vert\kern-0.25ex\right\vert}}

\begin{document}
\title[Exponential bounds of BEC for dilute Bose gases]{Exponential bounds of the condensation for dilute Bose gases}

\author[P.T. Nam]{Phan Th\`anh Nam}
\address{Department of Mathematics, LMU Munich, Theresienstrasse 39, 80333 Munich, Germany} 
\email{nam@math.lmu.de}

\author[S. Rademacher]{Simone Rademacher}
\address{Department of Mathematics, LMU Munich, Theresienstrasse 39, 80333 Munich, Germany} 
\email{simone.rademacher@math.lmu.de}

\date{\today}

\begin{abstract}
We consider $N$ bosons on the unit torus $\Lambda =[0,1]^3$ in the Gross-Pitaevski regime where the interaction potential scales as $N^2 V(N(x-y))$. We prove that the low-lying eigenfunctions and the Gibbs state at low temperatures exhibit the Bose-Einstein condensation in a strong sense, namely the probability of having $n$ particles outside of the condensation decays exponentially in $n$. 
\end{abstract}

\maketitle

\section{Introduction}


Bose--Einstein condensation (BEC) is a special phenomenon of the thermal equilibrium of Bose gases at low temperatures where a macroscopic fraction of particles occupy a common one-body quantum state. This  was predicted in 1924 by Bose \cite{Bose-24} and Einstein \cite{Einstein-24}  and has been observed experimentally in alkali gases since 1995 \cite{WC-95,K-95}, but the rigorous understanding of the BEC from first principles of quantum mechanics remains a major challenge in mathematical physics. In fact, the works \cite{Bose-24,Einstein-24} cover only the ideal gas, while in reality interactions between particles correspond to many important quantum effects such as  superfluidity and quantized vortices. 

On the mathematical side, the justification of the BEC for the ground state of interacting Bose gases 
in the thermodynamic limit remains open. However, recently the Lee--Huang--Yang formula \cite{LeeHuaYan-57} on the ground state energy has been established rigorously; see \cite{FouSol-20,FouSol-22} for lower bounds and \cite{YauYin-08,BasCenSch-21} for upper bounds (the corresponding formula for the free energy at low temperature has been also verified in \cite{HHNST,HHSST}). These works use localization methods that reduce the problem to the justification of BEC in smaller domains, where relevant error estimates essentially rely on results in the Gross--Pitaevskii regime. The aim of the present paper is to give a justification of the BEC for a class of interacting Bose gases in the Gross--Pitaevskii regime where the number of particles outside of the condensate is controlled in a rather strong sense, which in particular improves the results in  \cite{LS-02,BCCS_optimal}.

\subsection{Main results} We consider $N$ bosons on the torus $\Lambda =[0,1]^3$ in the Gross-Pitaevski regime where the system is  described by the Hamiltonian 
\begin{align}
\label{def:HN-intro}
H_N = \sum_{j=1}^N  (-\Delta_j  )+  \sum_{1 \leq i < j \leq N} N^2 v(N (x_i - x_j)) 
\end{align}
on $L_s^2 ( \Lambda^N)$, the symmetric subspace of $L^2 ( \Lambda^N)$ given by 
\begin{align}
L^2_s (\Lambda^N) := \lbrace \psi \in L^2 ( \Lambda^N) \;\; \vert \; \;  \psi( x_1, \dots, x_N) = \psi ( x_{\pi (1)}, \dots, x_{\pi (N)} ) , \; \; \forall \pi \in S_N \rbrace
\end{align}
where $S_N$ denotes the set of permutations. Here, we fix a non-negative compactly supported potential  $v$, thus ensuring that the scattering length of the interaction potential $N^2 v(Nx)$ is proportional to $N^{-1}$. This models dilute gases in the typical setting of experiments \cite{WC-95,K-95}.

In 2002, Lieb and Seiringer \cite{LS-02} proved that the ground state $\Psi_N$ of $H_N$ exhibits the complete Bose--Einstein condensation on the condensate wave function $u_0=1$, namely 
\begin{align} \label{eq:BEC-1}
\lim_{N\to \infty} \frac{1}{N} \langle \Psi_N, \cN_+ \Psi_N\rangle =0,\quad \cN_+= \sum_{i=1}^N Q_i, \quad Q=1-|u_0\rangle \langle u_0|
\end{align}
where we introduced the notation 
\begin{align}
\label{def:Qi}
Q_i = 1 \otimes \cdots \otimes 1 \otimes Q \otimes 1 \cdots \otimes 1 
\end{align}
for the operator that acts as identity on all the particles but the $i$-th, on which it acts as $Q$.

Recently, Boccato, Brennecke, Cenatiempo and Schlein \cite{BCCS_cond,BCCS_optimal} proved the  improved bound 
\begin{align} \label{eq:BEC-2}
\langle \Psi_N, \cN_+ \Psi_N\rangle \le \mathcal{O}(1),
\end{align}
which served as an important input in their proof of the validity of Bogoliubov's excitation spectrum \cite{BCCS}. For the generalization concerning inhomogeneous trapped Bose gases in $\mathbb{R}^3$, we refer to \cite{LS,NRS} for results similar to \eqref{eq:BEC-1},  \cite{NNRT,BSS_optimal_2} for results similar to \eqref{eq:BEC-2}, and  \cite{NT,BSS_bogo} for the justification of Bogoliubov's excitation spectrum.

Our main result is the following improvement of \eqref{eq:BEC-2}. 

\begin{theorem}[Exponential bound for low-lying eigenfunctions] 
\label{thm:main}Let $v\in L^3(\Lambda)$ be non-negative, compactly supported and spherically symmetric. Then there exists a constant $\kappa>0$ depending only on $v$ such that if $\psi_N$ is an eigenfunction of $H_N$ defined in \eqref{def:HN-intro}  with energy 
\begin{align}
\label{ass:ev}
\langle \psi_N, H_N \psi_N\rangle \le E_N + \mathcal{O}(1), \quad  E_N=\inf \sigma(H_N), 
\end{align}
then it holds that 
\begin{align} \label{eq:thm-mf}
\langle \psi_N, e^{ \kappa \mathcal{N}_+} \psi_N \rangle \leq  \mathcal{O}(1). 
\end{align}
\end{theorem}

Here are two quick remarks on our theorem concerning the existing literature.


\begin{remark}[Moment vs. exponential bounds] A moment bound of the form $\langle \psi_N,  \cN_+^k  \psi_N\rangle \le C_k$ was  obtained in \cite[Proposition 4.1]{BCCS} using an induction argument in $k$. The exponential bound \eqref{eq:thm-mf} would follow if one could show that $C_k\le k! C^k$ for all $k$. However, it seems to us that this conclusion does not readily follow from \cite{BCCS}. It is interesting that the approach in \cite{BCCS} works for every $\psi$ in the spectral subspace $\1_{[E_N,E_N+ \mathcal{O}(1)]}(H_N)$, while our method focuses only on eigenfunctions. 
%
%
\end{remark}
%

\begin{remark}[Probabilistic interpretation and extensions to related models] As a consequence of Theorem \ref{thm:main} and Markov's inequality, we obtain  
\begin{align}
\label{eq:expdecay}
\langle \psi_N, \1_{\{ \mathcal{N}_+ \geq n \}} \psi_N \rangle \le \langle \psi_N,  e^{\kappa (\cN_+-n)}\psi_N \rangle   \leq C e^{- \kappa n} . 
\end{align}
for all $0 \leq n  \leq N$. Consequently, the probability of finding $n$ particles outside the condensate decays exponentially in $n$. Our bounds \eqref{eq:thm-mf}-\eqref{eq:expdecay} extend easily to the less singular regimes where the interaction potential $N^2V(Nx)$ is replaced by 
\begin{align}
\label{eq:beta}
v_{N,\beta}(x)=N^{3\beta-1} v(N^\beta x), \quad \text{with parameter} \quad \beta \in [0,1). 
\end{align}
In the mean-field regime $\beta=0$, the bound \eqref{eq:expdecay} was settled by Mitrouskas \cite{M} (very recently, this result was extended by Mitrouskas-Pickl \cite{MP-2023} to include trapped bosons and also include the repulsive Coulomb potential). In the discussion below, we will illustrate our method by giving a short proof in the mean-field regime. In principle, the difficulty increases when $\beta$ becomes larger, and the Gross-Pitaevski regime $\beta=1$ is the most challenging case where strong correlations at short distances lead to a leading order correction in the ground state energy and the excitation spectrum. In another direction, a related exponential decay of excitations was derived in \cite[Proposition 4.2]{BS-2023} to investigate the ground state energy of the Fr\"ohlich Polaron model.
\end{remark} 


Next, from the exponential bound in Theorem \ref{thm:main} and the ground state  structure established in  \cite{BCCS}, we obtain a result concerning the large deviations for $\cN_+$. Motivated by the probabilistic approach to Bose-Einstein condensates introduced in \cite{BKS}, given the ground state $\psi_N$ of $H_N$ and a self-adjoint operator $A$  on $L^2_s (\Lambda^N)$, we denote
$$
\mathbb{E}(A) = \langle \psi_N, A \psi_N\rangle, \quad \mathbb{P}(A>x) = \langle \psi_N, \1_{(x,\infty)}(A) \psi_N\rangle
$$
where $ \1_{(x,\infty)}(A) $ is defined by the spectral theorem. 

\begin{corollary}[Large deviations for $\mathcal{N}_+$]
\label{thm:LDE} Let $v$ be as in Theorem \ref{thm:main} and let $\psi_N$ be the ground state of $H_N$.  Then for  $\lambda>0$ small,  we have 
\begin{align}\label{eq:exp-N+-comp}
\limsup_{N \rightarrow \infty} \left| \log  \mathbb{E} ( e^{\lambda ( \mathcal{N}_+ - \mu)} ) - \frac{\lambda^2}{2}  \sigma^2 \right| \le  \mathcal{O}( \lambda^3)
\end{align}
where
\begin{align}\label{def:sigma0}
\mu = \sum_{p \in \Lambda_+^*} \sinh^2( \nu_p),\quad \sigma^2 = \sum_{p \in \Lambda_+^*} \sinh^2( \nu_p) \cosh^2( \nu_p) , \quad \nu_p = \frac{1}{4} \log \bigg( \frac{p^2}{p^2 + 16 \pi \a0} \bigg). 
\end{align}
where the scattering length of the potential $\a0$ is defined in \eqref{def:a}. 
Consequently, there exists a constant $\lambda_0>0$ such that for all $x>0$ we have
%
\begin{align}\label{eq:fund-cor-2}  
\limsup_{N \rightarrow \infty}  \mathbb{P}  (  \mathcal{N}_+ - \mu > x ) \le  \inf_{\lambda > \lambda_0} \big[ - \lambda x + \frac{\lambda^2}{2 \sigma^2} + O ( \lambda^3) \big] \; .
\end{align}
In particular, for $0<x<\lambda_0 /\sigma^2$, the above estimate reduces to 
\begin{align}\label{eq:fund-cor-3}
\limsup_{N \rightarrow \infty} \log \mathbb{P}  (  \mathcal{N}_+ - \mu > x )  \leq  -\frac{x^2}{2 \sigma^2} + O( x^3) \; . 
\end{align}
\end{corollary}

It was proved in  \cite[Appendix A]{BCCS} that
$$\lim_{N\to \infty}  \langle \psi_N, \cN_+ \psi_N\rangle = \sum_{p \in \Lambda_+^*} \sinh^2( \nu_p)=\mu.$$ Thus Corollary \ref{thm:LDE} refines the decay property \eqref{eq:expdecay} and gives a detailed description for deviations of $\cN_+$ from its mean value.

\begin{remark}[Extensions of large deviations] As carried out in Section \ref{sec:proof-LDE}, the proof of Theorem \ref{thm:LDE} can be generalized easily with $\cN_+=\sum_{i=1}^N Q_i$ replaced by ${\rm d}\Gamma(O)= \sum_{i=1}^N O_i$ for any self-adjoint bounded operator $O$ on $L^2( \Lambda)$ such that $O = QOQ$. To be precise, we will show that there exists $\lambda_0 >0$ such that for all $0<x< \lambda_0/ \widetilde{\sigma}^2$, 
\begin{align}
\label{eq:LDE-O}
\limsup_{N \rightarrow \infty} \log \mathbb{P} \big[ \sum_{i=1}^N O_i - \widetilde{\mu} > x\big] \leq -\frac{ x }{2 \widetilde{\sigma}^2 }+ O ( x^3) 
\end{align}
 with 
 $$\widetilde{\mu} = \sum_{p \in \Lambda_+^*} \sinh( \nu_p) \widehat{O}_{p,p}, \quad \widetilde{\sigma} = \sum_{p,q} \vert  \widehat{O}_{p,q} \vert^2 \cosh^2( \nu_q) \sinh^2( \nu_p)$$ 
 where $\widehat{O}$ denotes the Fourier transform of the kernel of the operator $O$. 

Our result can be interpreted as a first step towards a more general asymptotic formula concerning large deviations in the Gross-Pitaevskii regime. In fact, it is natural to expect that a  similar bound also holds even if $O\not= QOQ$. Of course, in this case the contribution from the condensate becomes very large, and hence a suitable scaling limit has to be introduce in order to put the relevant large deviations in a rigorous form. We {\em conjecture} that for every or every self-adjoint bounded operator $O$ on $L^2( \Lambda)$, there exists a constant $\sigma_0>0$ depending on $O$ such that for $x>0$ small, 
\begin{align}
 \limsup_{N \rightarrow \infty} \frac{1}{N} \log \mathbb{P} \bigg[\frac{1}{N} \sum_{i=1}^N \big[ O_i - \langle \varphi, O \varphi \rangle \big] > x  \bigg] \leq - \frac{x^2}{2\sigma_0^2} + \mathcal{O}(x^{5/2}) .\label{eq:lde1}
\end{align}

In the mean-field regime ($\beta =0$), large deviations of the form \eqref{eq:lde1}, with $O\not= QOQ$, have been characterized \cite{KRS,R,RSe}
where the variance $\sigma_0$ is computed explicitly using Bogoliubov approximation. In this case, the matching lower bound also holds (see  \cite{R,RSe}, thus the bound refines earlier results on central limit theorems \cite{RS}. 

For more singular potentials, results on the law of large numbers and central limit theorems around Bose-Einstein condensates have been proven for correlated random variables  \cite{Rsing,COS} and for empirical measures \cite{PRV}. However, the  asymptotic formula \eqref{eq:lde1} in the Gross-Pitaevskii regime remains a very interesting {\em open problem}. 

\end{remark}

Finally, let us discuss another extension of Theorem \ref{thm:main} concerning the Gibbs state at positive temperatures. In Theorem \ref{thm:main}, we consider each eigenfunction of $H_N$ separately. It is also possible to consider all eigenfunctions at the same time, namely we turn to the thermal equilibrium of the system given by the Gibbs state 
\begin{align}
\Gamma_\beta := \frac{e^{- \beta H_N}}{ Z(\beta)}, \quad \text{where} \quad Z( \beta) = \Tr e^{-\beta H_N} \label{def:Gibbs}
\end{align}
at a positive temperature $T = 1/\beta >0$. This is the unique miminizer of the free energy functional 
\begin{align}
\mathcal{F} ( \Gamma) = \Tr \left[ H_N \Gamma \right] - \frac{1}{\beta} S( \Gamma ), \quad \text{with} \quad S( \Gamma )= - \Tr \left[ \Gamma \ln (\Gamma ) \right]
\end{align}
over the set of all mixed states on $L_s^2(\Lambda^N)$   (the set of all non-negative operators on $L_s^2(\Lambda^N)$ with trace $1$). Our bound in Theorem \ref{thm:main} extends to the Gibbs state at low temperatures.

\begin{theorem}[Exponential bound for the Gibbs state at low temperature] 
\label{thm:posT} 
Let $v\in L^3(\Lambda)$ be non-negative, compactly supported and spherically symmetric. Then for every fixed temperature $T=\beta^{-1}>0$ and for a sufficiently small $\kappa >0$, the Gibbs state $\Gamma_\beta$ given by \eqref{def:Gibbs} satisfies 
\begin{align}
\Tr \left[ e^{\kappa \cN_+} \Gamma_\beta  \right] \leq \mathcal{O}(1).  \label{eq:N-posT}
\end{align}
\end{theorem}

%

\begin{remark}[Low vs. high temperatures] For low temperatures, $T\sim 1$, the second order of the free energy can be deduced from the analysis of the excitation spectrum \cite{BCCS} (see also \cite{HST,Brooks-2023} for simplified proofs, and \cite{HHNST} for corresponding results in thermodynamic limit). However, properties of Gibbs state are  less understood; in particular \eqref{eq:N-posT} is new.
%
%
%
%
For higher temperatures, we do not expect that \eqref{eq:N-posT} holds. In particular, when $T \sim N^{2/3}$, namely $T$ is comparable to the critical temperature of the BEC phase transition, we do not expect the  complete BEC \eqref{eq:BEC-1} since the number of excited particles is also proportional to $N$ (see \cite{DS,BDS-23,CD-23} for rigorous results).
\end{remark}

\subsection{Ideas of the proof}

Now let us explain our proof strategy.  To make the ideas transparent, we will first illustrate our method by giving a short proof of \eqref{eq:thm-mf} in the mean-field regime, and then explain additional arguments needed for the Gross--Pitaevskii regime.

\bigskip
\noindent
{\em Mean-field regime:} Let us start by proving \eqref{eq:thm-mf} in the mean-field regime, where the potential $N^2v(Nx)$ is replaced by $(N-1)^{-1}v$ with a periodic potential $v$ satisfying $0\le \widehat w \in \ell^1(2\pi \mathbb{Z}^3)$. In this case, our result is comparable to \cite[Theorem 3.1]{M}, but our proof below is different. Our argument goes back to the moment estimates obtained in \cite[Lemma 3]{Nam-18} and \cite[Lemma 3]{NamNap-21}, but now we aim at exponential estimates. 

We consider the mean-field Hamiltonian, which can be written in the momentum space as 
\begin{align}
H_N^{\rm mf}= \sum_{p \in 2\pi \mathbb{Z}^3 } p^2 a_p^*a_p + \frac{1}{2(N-1)}\sum_{p,q,\ell \in 2\pi \mathbb{Z}^3} \widehat{v} (\ell) \; a^*_{p - \ell} a^*_{ q + \ell} a_p a_q    \label{def:Ham-mf}
\end{align}
where $a_p^*,a_p$ are the standard creation and annihilation operators  on the bosonic Fock space $\mathcal{F} = \bigoplus_{n \geq 0} L_s^2 ( \Lambda^n)$. They  satisfy the canonical commutation relations 
\begin{align}
\label{eq:comm}
\left[ a_p^*, a_q \right] = \delta_{p,q} ,  \quad \left[ a_p^*, a_q^* \right] = \left[ a_p, a_q \right] =0,\quad \forall p,q\in  \Lambda^*=2 \pi \mathbb{Z}^3. 
\end{align}
In particular, the condensate  is described by the constant function $u_0=1$, corresponding to the zero momentum. The number of particles outside the condensate, often called, the number of excitations, can be written as 
\begin{align}
\cN_+ =  \sum_{p \in \Lambda_+^*} a_p^*a_p, \quad \text{with} \quad \Lambda_+^* = 2 \pi \mathbb{Z}^3 \setminus \lbrace 0 \rbrace  \; . 
\end{align}
Let us prove \eqref{eq:thm-mf} for the ground state $\psi_N$ of $H_N^{\rm mf}$. We define, for $s \in [0,1]$ and $\kappa>0$ small enough,  
\begin{align}
\label{eq:xi-s-beginning}
\xi_{N} (s) := e^{s \kappa \mathcal{N}_+} \psi_N \in L_s^2 ( \Lambda^N).
\end{align}
Since $\|\xi_N(0)\|=1$, to bound $\| \xi_N (1) \|^2$  it thus suffices to control 
\begin{align}
\label{eq:deriv-0-mf}
\partial_s \| \xi_N (s) \|^2 =  2 \kappa \langle \xi_N (s), \mathcal{N}_+  \xi_N (s) \rangle. 
\end{align}
In the mean-field regime, by Onsager's inequality (that was used first in \cite{Onsager} and later in a concise form in \cite[Eq. 8]{Seiringer-11}) we have immediately  the lower bound 
\begin{align}
\label{eq:Onsager}
H_N^{\rm mf} - E_N^{\rm mf} \geq C^{-1} \mathcal{N}_+ - c
\end{align}
with the ground state energy $E_N^{\rm mf}$ of $H_N^{\rm mf}$ and constants $C,c>0$. Combining with the ground state equation $(H_N^{\rm mf}- E_N) \psi_N = 0$, we can estimate the right-hand side of \eqref{eq:deriv-0-mf} as
\begin{align}
\label{eq:bound-N-mf}
C^{-1}\langle \xi_N (s), \mathcal{N}_+  \xi_N (s) \rangle  &\leq   \langle \xi_N (s), \left( H_N^{\rm mf} -E_N^{\rm mf} \right)  \xi_N (s) \rangle \nonumber\\
&= -  \frac{1}{2} \langle \psi_N,  \;  \left[ e^{ s\kappa \mathcal{N}_+} , \; \; \left[  e^{s \kappa \mathcal{N}_+} ,  \;  H_N^{\rm mf}\right] \right]  \psi_N \rangle\; . 
\end{align}
The right-hand side of \eqref{eq:bound-N-mf} can be computed explicitly 
\begin{align}
& \left[ e^{s\kappa \mathcal{N}_+ }, \; \left[ e^{s\kappa \mathcal{N}_+}, \; H_N^{\rm mf} \right] \right] \notag \\
&\quad =  \frac{2}{N-1}  \sinh^2( s\kappa)  \; e^{s\kappa \mathcal{N}_+} \sum_{\substack{\ell \in \Lambda_+^*}} \widehat{v} ( \ell ) \;  \left[  a^*_{-\ell} a^*_\ell a_0 a_0  \;  -  a^*_0 a_0^* a_\ell a_{- \ell} \right]  e^{s\kappa \mathcal{N}_+}   \notag \\
&\quad\quad +  \frac{1}{N-1}  \sinh^2( s\kappa /2)   e^{s\kappa \mathcal{N}_+} \sum_{\substack{p, \ell \in \Lambda_+^* \\ p \not= \ell}} \widehat{v} ( \ell ) \left[  a^*_{p-\ell} a^*_0 a_p a_{\ell} + a^*_{p- \ell} a_{- \ell}^* a_p a_{0} \right] \; e^{s\kappa \mathcal{N}_+}\notag\\
&\quad\quad +  \frac{1}{N-1}    \sinh^2 ( s\kappa /2) e^{s\kappa \mathcal{N}_+} \sum_{\substack{\ell,q \in \Lambda_+^* \\ q\not= - \ell }} \widehat{v} ( \ell ) \left[   a^*_{0} a^*_{q + \ell} a_{\ell} a_q +  a^*_{-\ell} a_{q+\ell}^* a_0 a_{q} \right]  \;  e^{s\kappa \mathcal{N}_+}  \; .  \label{eq:second-comm}
\end{align}
Here we used $\mathcal{N}_+ a_0 =a_0 \mathcal{N}_+$ and $\mathcal{N}_+ a_p = a_p (\mathcal{N}_+- 1)$ for $p \in \Lambda_+^*$. We can estimate the three summands of the right hand side of \eqref{eq:second-comm} separately. For this we recall the bounds for $a^*(h) = \sum_{p \in \Lambda_+^*} h_p a^*_p$ for any $h \in \ell^2( \Lambda_+^*)$ and any Fock space vector $\xi \in \mathcal{F}$ 
\begin{align}
\label{eq:bounds-a1}
\|  a (h) \xi \| \leq \| h \|_{\ell^2} \| \mathcal{N}_+^{1/2} \xi \|, \quad \|  a^* (h) \xi \| \leq \| h \|_{\ell^2} \| (\mathcal{N}_+ + 1)^{1/2} \xi \| ,
\end{align}
and 
\begin{align}
\label{eq:bounds-a2}
\vert  \sum_{p \in \Lambda_+^*} h_p \langle \xi_1, \; a_p^*a_{-p}^* \xi_2 \rangle \vert \leq& \| h \|_{\ell^2} \| ( \mathcal{N}_+ + 1 )^{1/2} \xi_2 \| \; \| \mathcal{N}_+^{1/2} \xi_1 \| .
\end{align}
Furthermore for any operator $H$ on $\ell^2( \Lambda^*_+)$ with kernel $H_{p,q}$, we have 
\begin{align}
\label{eq:bounds-a3}
 \| \sum_{p, q \in \Lambda_+^*} H_{p,q} a_p^*a_q \xi \|   \leq \| H \|_{{\rm op}} \| \mathcal{N}_+ \xi \| \; . 
\end{align}
On the one hand, we observe that $\mathcal{N}_+$ commutes with $a_0$ and $\mathcal{N}_+ a_p = a_p ( \mathcal{N}_+ -1)$ for $p \not=0$. Therefore, for the first term of the r.h.s. of \eqref{eq:second-comm}, we find that for any vector $\psi \in L_s^2 ( \Lambda^N)$,
\begin{align}
\vert \langle & \psi, e^{s\kappa \mathcal{N}_+} \sum_{0\ne \ell \in \mathbb{Z}^d} \widehat{v} ( \ell ) \;  \left[  a^*_{-\ell} a^*_\ell a_0 a_0  \;  +  a^*_0 a_0^* a_\ell a_{- \ell} \right]  e^{s\kappa \mathcal{N}_+}  \psi  \rangle \vert \notag \\
&\leq 2 \| \widehat{v} \|_{\ell^2 ( \mathbb{Z}^d)} \|  ( \mathcal{N}_+ +1 )^{1/2} a_0 a_0 e^{s\kappa \mathcal{N}_+}\psi\| \| \mathcal{N}_+ ^{1/2}e^{s\kappa \mathcal{N}_+}  \psi \|^2 \nonumber\\
&\le C N \|( \mathcal{N}_+ +1 )^{1/2} e^{s\kappa \mathcal{N}_+} \psi\|^2
\end{align}
where we used \eqref{eq:bounds-a2} and  $a_0^* a_0 \leq N$ on $L_s^2 ( \Lambda^N)$.  
Similarly, for the second term of the r.h.s. of \eqref{eq:second-comm}, we have
\begin{align}
\vert \langle \psi, & \; e^{s\kappa\mathcal{N}_+}\sum_{\substack{\ell,p \in \Lambda_+^* \\  p \not= \ell}} \widehat{v} ( \ell ) \left[  a^*_{p-\ell} a^*_0 a_p a_{\ell} + a^*_{p- \ell} a_{- \ell}^* a_p a_{0} \right] \; e^{s\kappa \mathcal{N}_+} \; \psi \rangle \vert \notag \\
\leq& \left( \sum_{\substack{\ell,p \in \Lambda_+^* \\  p \not= \ell}} \vert \widehat{v} ( \ell ) \vert^2 \; \| a_0 a_{p - \ell} e^{s\kappa \mathcal{N}_+ } \psi \|^2 \right)^{1/2} \left( \sum_{\substack{\ell,p \in \Lambda_+^* \\  p \not= \ell}}  \; \| a_p a_{ \ell} e^{s\kappa \mathcal{N}_+ } \psi \|^2 \right)^{1/2} \notag \\
&+ \left( \sum_{\substack{\ell, p \in \Lambda_+^* \\  p \not= \ell}} \vert \widehat{v} ( \ell ) \vert^2 \; \| a_0 a_{p } e^{s\kappa \mathcal{N}_+ } \psi \|^2 \right)^{1/2} \left( \sum_{\substack{\ell,p \in \Lambda_+^* \\  p \not= \ell}} \; \| a_{p-\ell} a_{ -\ell} e^{s\kappa \mathcal{N}_+ } \psi \|^2 \right)^{1/2} \nonumber\\
&\le CN  \|( \mathcal{N}_+ +1 )^{1/2} e^{s\kappa \mathcal{N}_+} \psi \|^2 
\end{align}
where we also used $\mathcal{N}_+ \leq N$. 
For the last term of the r.h.s. of \eqref{eq:second-comm} we proceed similarly and thus arrive at 
\begin{align*}
\vert \langle  \psi, \; \left[ e^{s\kappa \mathcal{N}_+ }, \; \left[ e^{s\kappa \mathcal{N}_+}, \; H_N^{\rm mf} \right] \right]  \psi \rangle \vert 
\leq C \sinh^2( s\kappa/2)\langle \psi, e^{\kappa \mathcal{N}_+} \left( \mathcal{N}_+ + 1 \right) e^{\kappa \mathcal{N}_+} \psi \rangle . 
\end{align*}
For small $\kappa>0$ and $s\in [0,1]$ we have $\sinh^2( s\kappa/2) \leq C \kappa^2 $  for a universal positive constant $C>0$. We thus arrive at 
\begin{align}
\vert \langle  \psi_N, \; \left[ e^{s\kappa \mathcal{N}_+ }, \; \left[ e^{s\kappa \mathcal{N}_+}, \; H_N^{\rm mf} \right] \right]  \psi_N \rangle \vert 
\leq C  \kappa^2  \langle \xi_N (s),  \left( \mathcal{N}_+ + 1 \right) \xi_N (s) \rangle \label{eq:second-comm-end} \; . 
\end{align}
We recall \eqref{eq:bound-N-mf} and find that 
\begin{align}
C^{-1} \langle\xi_N(s), \;  \mathcal{N}_+ \xi_N (s) \rangle \leq& \frac{1}{2} \vert \langle  \psi_N, \; \left[ e^{s\kappa \mathcal{N}_+ }, \; \left[ e^{s\kappa \mathcal{N}_+}, \; H_N^{\rm mf} \right] \right]  \psi_N \rangle \vert \notag \\
&\leq C \kappa^2 \langle \xi_N (s),  \left( \mathcal{N}_+ + 1 \right) \xi_N (s) \rangle
\end{align}
and thus that for sufficiently small $\kappa>0$ we have 
\begin{align}
\langle \xi_N(s), \;  \mathcal{N}_+ \xi_N (s) \rangle \leq C \kappa^2 \| \xi_N (s) \|^2. 
\end{align}
Combining the latter bound with \eqref{eq:deriv-0-mf} we arrive at 
\begin{align}
\| \xi_N (1) \|^2 = \| \xi_N (0) \|^2 + \int_0^t  \partial_s \| \xi_N (s) \|^2 ds  \leq& 1+ C \kappa^2 \int_0^t \| \xi_N (s) \|^2 ds 
\end{align}
that yields with Gronwall's inequality the desired estimate 
\begin{align}
\label{eq:xi-s-end}
 \langle \psi_N, \; e^{2 \kappa \mathcal{N}_+} \psi_N \rangle = \| \xi_N (1) \|^2  \leq C e^{C \kappa^2} \; . 
\end{align}
Summarizing, our proof in the mean-field regime relies on two crucial bounds for the Hamiltonian $H^{\rm mf}_N$, namely the lower bound \eqref{eq:Onsager} and the estimate on the double commutator \eqref{eq:second-comm-end}. In scaling regimes with singular interactions of the particles (i.e. $\beta >0$ in \eqref{eq:beta}) similar estimates hold true only after regularizing the Hamiltonian with appropriate unitary transformations that extract the particles' strong correlations. We explain our strategy in the following in more detail:

 \bigskip

\noindent
{\em Gross--Pitaevskii regime:} In the  Gross--Pitaevskii regime, we need to extract strong correlations at short distances before applying the above strategy. To do this, we first use a unitary transformation introduced in \cite{LNSS} to factor out the contribution of the condensate, and then use a generalized Bogoliubov transformation developed in  \cite{BS,BCCS_cond,BCCS_optimal,BCCS} to capture the correlation structure. 

Let us write the Hamiltonian $H_N$ in \eqref{def:HN-intro} as 
\begin{align}
H_N= \sum_{p \in \mathbb{Z}^d } p^2 a_p^*a_p + \frac{1}{2N}\sum_{p,q,\ell \in \mathbb{Z}^d} \widehat{v} (\ell/N) \; a^*_{p - \ell} a^*_{ q + \ell} a_p a_q   \; . \label{def:Ham}
\end{align}
Controlling $\cN_+$ in the ground state of $H_N$, or more generally excited states with low energy, is our main goal. To this end we first factor out the condensate's contribution using the unitary $\mathcal{U}_N$ 
\begin{align}
\label{def:UN}
\mathcal{U}_N : L^2_s( \Lambda) \rightarrow \mathcal{F}^{\leq N}_{\perp u_0}=\bigoplus_{n=0}^N L^2_{\perp u_0} ( \Lambda)^{\otimes_s n}
\end{align}
introduced in  \cite{LNSS}, which maps any $N$-particle wave function  
\begin{align}
\psi_N = \eta_0 u_0^{\otimes_s N} + \eta_1 \otimes_s u_0^{\otimes_s (N-1)} + \cdots + \eta_N , \quad \text{with} \quad \eta_j \in L^2_{\perp u_0}( \Lambda )^{\otimes_s^j}
\end{align}
onto its excitation vector $(\eta_0, \cdots, \eta_N)$. Here $L^2_{\perp u_0} ( \Lambda)$ denotes the orthogonal complement of $u_0$ in $L^2( \Lambda)$. In the  following, we will focus on the excitation Hamiltonian $\mathcal{U}_N H_N \mathcal{U}_N^*$ on $\mathcal{F}^{\leq N}_{\perp u_0}$. 



In the Gross-Pitaevski regime the particles experience rare but strong interactions, and hence the correlations of the particles play a crucial role. To capture the correlation structure of particles, we use the solution $f$ of the scattering equation
\begin{align}
\label{eq:Neumann} 
\left( - \Delta + \frac{1}{2} v  \right) f  = 0 
\end{align}
with boundary condition $f(x) \rightarrow 1$ as $\vert x \vert \rightarrow \infty $. Recall that the scattering length $\mathfrak{a}_0$ of the potential $v$ is given by 
\begin{align}
\label{def:a}
\mathfrak{a}_0 = \int dx \; v(x) f(x).
\end{align}
By scaling, the scattering solution of $ N^2 v (  N \cdot )$ is $f_N(x)=f(Nx)$, and the corresponding scattering length is $a_0/N$. In the following we denote $v_N (x)  =N^3 v (Nx)$. By technical reason, in the following we will replace $f_N$ by $f_{N,\ell}$ with $0<\ell<1/2$ (independent of $N$) the solution to the Neumann boundary problem 
\begin{align}
\label{eq:Neumann} 
\left( - \Delta + \frac{1}{2N} v_N (x) \right) f_{N, \ell} (x) = \lambda_{N, \ell} f_{N, \ell} (x) 
\end{align}
on the ball $B_\ell := \lbrace x \in \mathbb{R}^3: \vert x \vert \leq \ell \rbrace $ with the normalization condition that $f_{N, \ell} (x) = 1$ for $\vert x \vert \ge \ell$. Then following the ideas in \cite{BS,BCCS_cond,BCCS_optimal,BCCS} we implement the particles' correlation structure through a Bogoliubov transformation given by 
\begin{align}
\label{def:bogo}
e^{B_\eta } \quad \text{with} \quad B_\eta := \exp \big(\frac{1}{2} \sum_{p \in \Lambda_+^*} \left( \eta_p  b_p^*b_{-p}^* -  \overline{\eta}_p b_p b_{-p} \right) \big),\quad b_p = \sqrt{1-\mathcal{N}_+/N} a_p.   
\end{align}
Here, the sequence $\eta \in \ell^2( \Lambda_+^*)$ is chosen as
\begin{align}
\label{def:eta} \eta_p = - N \widehat{\omega}_{N, \ell} (p) \quad \text{for all} \quad p \in \Lambda^*_+ . 
\end{align}
where 
$$\omega_{N, \ell} (x) = 1- f_{N, \ell} (x),\quad \widehat{\omega}_{N, \ell} (p) = \int_{\Lambda} \omega_{N, \ell} (x) e^{-ip \cdot x} dx  \quad \text{for all} \quad p \in \Lambda^*.
$$ 

Then we define the new excitation Hamiltonian with correlation structure as 
\begin{align}
\label{def:G}
\mathcal{G}_{N}:= e^{B ( \eta )}\mathcal{U}_N H_N \mathcal{U}_N^* e^{-B( \eta )} \; . 
\end{align}
We will show that $\mathcal{G}_{N}$ is bounded from below by a positive multiple of $\mathcal{H}_N = \mathcal{K} + \mathcal{V}_N$ with 
\begin{align}
\label{def:HN}
\mathcal{K} = \sum_{p \in \Lambda_+^*} p^2 a_p^*a_p, \quad \text{and} \quad \mathcal{V}_N = \sum_{\substack{p,q,r \in \Lambda_+^* \\  r\not= -p,-q}} \widehat{v}(  p/N )a_{p+r}^*a_q^*a_pa_{q+r} \; . 
\end{align} 
In particular, the proof of Theorem \ref{thm:main} is based on the following properties of $\mathcal{G}_{N}$.
\begin{proposition}
\label{prop:G}
Under the same assumptions as in Theorem \ref{thm:main}, we have 
\begin{align}
\label{eq:G-lb}
 \mathcal{G}_{N}  -E_N   \geq \frac{1}{2} \mathcal{H}_N - C   \; . 
\end{align}
where $E_N$ denotes the ground state of $E_N$. Furthermore, for sufficiently small $\kappa >0$ we have for any Fock space vector $\psi \in \mathcal{F}_{\perp u_0}^{\leq N}$
\begin{align}
\label{eq:G-dc}
\vert \langle \psi, \, \left[ e^{\kappa \cN_+}, \; \left[ e^{\kappa \cN_+}, \mathcal{G}_N \right] \right] \psi \rangle \vert \leq C  \kappa^2 \langle \psi, e^{\kappa \cN_+}\left( \mathcal{H}_N  + (\cN_+ + 1) \right) e^{\kappa \cN_+} \psi \rangle \; . 
\end{align}
Here $C=C_v>0$ depends only on the potential $v$. 
\end{proposition}

These bounds enable us to use the previous strategy in the mean-field regime, with $H_N^{\rm mf}$ replaced by  $\mathcal{G}_{N}$.  While the first bound \eqref{eq:G-lb} essentially follows from the analysis in \cite{BCCS,BCCS_optimal}, the new bound \eqref{eq:G-dc}  is important for us, and it requires several refined estimates. 

Before ending the introduction, let us make a technical remark concerning the generalized Bogoliubov transformation in \eqref{def:bogo}. The idea of using a transformation which is quadratic in $N^{-1/2} a_0^* a_p$ to diagonalize the interacting Hamiltonian goes back to the work of Seiringer \cite{Seiringer-11} on the excitation spectrum in the mean-field regime (see also \cite{GreSei-13} for the extension to trapped systems). After removing the condensate by $\mathcal{U}_N$ in \eqref{def:UN}, we find that $N^{-1/2} a_0^* a_p \mapsto b_p$ given in \eqref{def:bogo}. The idea of using the generalized Bogoliubov transformation $e^{B_\eta }$ where the kernel $\eta$ captures only the high-momentum part via the scattering solution in \eqref{eq:Neumann} goes back to the work of Brennecke--Schlein \cite{BS} in the dynamical problem, and extended further in \cite{BCCS_cond,BCCS_optimal} in the stationary problem. This gives an efficient way to renormalize the interacting Hamiltonian, leaving out only contributions of order $1$ which were further computed in \cite{BCCS} to obtain the excitation spectrum. As explained in  \cite{NT}, actually the analysis of the excitation spectrum can be done using only the standard Bogoliubov transformation with $b_p$ replaced by $a_p$.  However, we are not able to use this simplification to achieve the exponential bounds in the present paper (although we can do this for the moment bound $\langle \cN^k\rangle \le \mathcal{O}(1)$). In particular, we will benefit greatly from the precise asymptotic behavior of the generalized Bogoliubov transformation $e^{B_\eta }$  established in \cite{BS} where the error to the standard actions of the Bogoliubov transformation is estimated carefully. We hope that although our detailed analysis is inevitably complicated, the general idea is transparent from the above discussion.

\medskip

\noindent
{\bf Structure of the paper.} In Section \ref{sec:prel} we collect useful properties of the excitation Hamiltonian $\mathcal{G}_N$ and of the second nested commutator with the exponential of the number of excitations. Then we prove Proposition \ref{prop:G} in Section \ref{sec:proof-thm}. Finally, we conclude Theorems \ref{thm:main} and  \ref{thm:posT} in Section \ref{sec:proof-thm}. 

\medskip

\noindent
{\bf Acknowledgements.}  We would like to thank Lea Bo\ss mann, Christian Brennecke, Morris Brooks, and David Mitrouskas for helpful remarks. This work was partially funded by the Deutsche Forschungsgemeinschaft  (DFG project Nr. 426365943) and by the European Research Council via the ERC Consolidator Grant  RAMBAS (Project Nr. 10104424).

\section{Preliminaries} 
\label{sec:prel}

In this Section we collect preliminary results necessary for the proof of Theorem \ref{thm:main} and Proposition \ref{prop:G}. First, in Section \ref{subsec:Fock-exc}, we compute the excitation Hamiltonian $\mathcal{G}_N$ defined in \eqref{def:G}. Second, in Section \ref{subsec:prel-est}, we discuss preliminary estimates that we need to study the properties of $\mathcal{G}_N$ in Section \ref{sec:proof-prop-G}. 

%
%
%
%
%

\subsection{Excitation Hamiltonian} 
\label{subsec:Fock-exc}

To study the excitations of the condensate wave function, we consider the excitation Hamiltonian, i.e. the Hamiltonian $H_N$ mapped through the unitary $\mathcal{U}_N$ defined in \eqref{def:UN} onto Fock space of excitations $\mathcal{F}^{\leq N}$ with respect to the on which the excitation Hamiltonian 
\begin{align}
\mathcal{L}_N := \mathcal{U}_N H_N \mathcal{U}_N^* 
\end{align} 
and is given by the sum $\mathcal{L}_N =  \mathcal{L}_N^{(0)}  + \mathcal{L}_N^{(2)} + \mathcal{L}_N^{(3)} + \mathcal{L}_N^{(4)} $ of the terms 
\begin{align}
\mathcal{L}_N^{(0)} =& \frac{N-1}{2N} \widehat{v} (0) (N- \mathcal{N}_+ ) + \frac{\widehat{v}(0)}{2N} \mathcal{N}_+ ( N- \mathcal{N}_+ ) ,\notag \\
\mathcal{L}_N^{(2)} =& \sum_{p \in \Lambda_+^*} p^2 a_p^* a_p + \sum_{p \in \Lambda_+^*} \widehat{v} (p/N ) \left[ b_p^*b_p - \frac{1}{N} a_p^*a_p \right] + \frac{1}{2} \sum_{p \in \Lambda_+^*} \widehat{v}(p/N ) \left[ b_p^*b_{-p}^* + b_p b_{-p} \right] ,\notag \\
\mathcal{L}_N^{(3)} =& \frac{1}{\sqrt{N}} \sum_{\substack{p,q, \in \Lambda_+^* \\ p + q \not=0}} \widehat{v} (  p/N) \left[ b^*_{p+q} a_{-p}^*a_q + a_q^*a_{-p} b_{p+q} \right] ,\notag \\
\mathcal{L}_N^{(4)} =& \frac{1}{2N} \sum_{\substack{p,q \in \Lambda_+^*, r \in \Lambda^* \\ r \not= -p,-q }} \widehat{v} ( r/N ) a_{p+r}^*a_q^*a_p a_{q+r} \; . \label{def:Lj}
\end{align}
Here we introduced the modified creation and annihilation operators 
\begin{align}
b_p^* = a_p^* \sqrt{1- \mathcal{N}_+/N}, \quad \text{and} \quad b_p = \sqrt{1- \mathcal{N}_+/N} a_p.
\end{align}
They effectively behave as standard creation and annihilation operators in the limit of $N \rightarrow \infty$, but they help us to stick with the truncated Fock space where $\cN_+ \le N$. Their commutation relations
\begin{align}
\label{eq:comm-b}
[b_p^*, b_q^*] = [b_p,b_q] =0, \quad [b_p,b_q^*] = \delta_{p,q} ( 1- \cN_+ /N) - a_q^*a_p 
\end{align}
agree with the CCR \eqref{eq:comm} up to a contribution that is of order $N^{-1}$. Similarly to the estimates \eqref{eq:bounds-a1}-\eqref{eq:bounds-a3} for the standard creation and annihilation operators, the modified creation and annihilation operators satisfy 
\begin{align}
\label{eq:bounds-b1}
\|  b (h) \xi \| \leq \| h \|_{\ell^2} \| \mathcal{N}_+^{1/2} \xi \|, \quad \|  b^* (h) \xi \| \leq \| h \|_{\ell^2} \| (\mathcal{N}_+ + 1)^{1/2} \xi \| ,
\end{align}
where, similarly as for the standard creation and annihilation operators, we write $b^*(h) = \sum_{p \in \Lambda_+^*} h_p b_p^*$ and 
\begin{align}
\label{eq:bounds-b2}
\vert  \sum_{p \in \Lambda_+^*} h_p \langle \xi_1, \; b_p^*b_p^* \xi_2 \rangle \vert \leq& \| h \|_{\ell^2} \| ( \mathcal{N}_+ + 1 )^{1/2} \xi_2 \| \; \| \mathcal{N}_+^{1/2} \xi_1 \| .
\end{align}
Furthermore for any operator $H$ on $\ell^2(\Lambda_+^*)$ with kernel $H_{p,q}$ 
\begin{align}
\label{eq:bounds-b3}
 \| \sum_{p, q \in \Lambda_+^*} H_{p,q} b_p^*b_q \xi \| \leq \|H\|_{\rm op}
 \| \mathcal{N}_+ \xi \| \; . 
\end{align}

In the Gross-Pitaevski regime the particles' correlation structure plays a crucial role that we shall implement through the Bogoliubov transformation given by \eqref{def:bogo} with respect to the function $\eta \in \ell^2( \Lambda_+^*)$ defined in \eqref{def:eta} in terms of $\widehat{\omega}_{N, \ell} $ with $\omega_{N, \ell} (x) = 1 - f_{N, \ell} (x)$. The following Lemma collects properties of the scattering solution $f_{N,\ell}$ and $\omega_{N, \ell}$. 

\begin{lemma}[Lemma 3.1 \cite{BCCS}]\label{lemma:scattering} Let $v \in L^3( \Lambda)$ be non-negative, compactly supported and spherically symmetric. Fix $0 < \ell < \frac12$ and let $f_{N, \ell}$ denote the ground state of the solution of the Neumann problem \eqref{eq:Neumann}. 
\begin{enumerate}
\item[(i)] We have $\lambda_{N, \ell} = \frac{3 \widehat{v}(0)}{8 \pi N \ell^3} (1 + O( N^{-1}))$ and $0 \leq f_{N, \ell}, \omega_{N, \ell} \leq 1$.
\item[(ii)] There exists $C>0$ such that $\widehat{\omega}_{N, \ell} (p) \leq \frac{C}{N p^2}$ for all $p \in \Lambda_+^*$. 
\end{enumerate}
\label{lemma:scattering} 
\end{lemma}

We recall that from \eqref{def:eta} we have 
\begin{align}
 \eta_p = - N \widehat{\omega}_{N, \ell} (p) \quad \text{for all} \quad p \in \Lambda^*_+ 
\end{align}
and thus it follows from Lemma \ref{lemma:scattering} that 
\begin{align}
\vert \eta_p \vert \leq C p^{-2}, \quad \text{thus} \quad \eta \in \ell^2 ( \Lambda_+^*)
\end{align}
Note that by an appropriate choice of $\ell$, the norm $\| \eta \|_{\ell^2}$ can be choosen arbitrary small that will be important later. We remark that in the following we neglet the dependence of $\ell$ in the notation of. The scattering equation \ref{eq:Neumann} shows that
\begin{align}
\label{eq:id-eta}
p^2 \eta_p + \frac{1}{2N} \widehat{v}( p/N) + \frac{1}{2 N} \sum_{q \in \Lambda^*} \widehat{v}((p-q)/N) \eta_q = N \lambda_{N, \ell} \widehat{\chi}_\ell (p) + \lambda_{N, \ell} \sum_{q \in \Lambda^*} \widehat{\chi}_\ell (p-q) \eta_q 
\end{align}
where $\chi_\ell$ denotes the characteristic function on the ball $B_\ell$ with radius $\ell$. In the following we will study the excitation Hamiltonian $\mathcal{G}_N$ defined in \eqref{def:G}. We introduce the splitting 
\begin{align}
\label{def:G-split} 
\mathcal{G}_N := \mathcal{G}_N^{(0)} + \mathcal{G}_N^{(2)} +\mathcal{G}_N^{(3)} +\mathcal{G}_N^{(4)}
\end{align}
where the single contributions $\mathcal{G}_N^{(j)}$ are given by 
\begin{align}
\mathcal{G}_N^{(j)} := e^{-B ( \eta) } \mathcal{L}_N^{(j)} e^{B ( \eta)} 
\end{align}
with $\mathcal{L}_N^{(j)}$ given by \eqref{def:Lj}. We can explicitly compute the terms $\mathcal{G}_N^{(j)}$ using that the Bogoliubov transform's action on creation and annihilation operators is explicitly known and given by 
\begin{align}
\label{eq:action-bogo}
e^{-B( \eta)} b_p e^{B (\eta)} = \gamma_p b_p + \sigma_p b^*_{-p} + d_p, \quad \text{and} \quad  e^{-B( \eta)} b_p^* e^{B (\eta)} = \gamma_p b^*_p + \sigma_p b_{-p} + d_p^* 
\end{align}
where we introduced the shorthand notation 
\begin{align}
\label{eq:sh-sigma}
\sigma_p := \sinh( \eta_p) , \quad \gamma_p = \cosh( \eta_p) \quad \text{with $\eta_p$ given by \eqref{def:eta} } \; . 
\end{align}
Note that Lemma \ref{lemma:scattering} implies that with the splitting 
\begin{align}
\label{eq:sigma-splitting}
\sigma_p = \eta_p + \beta_p, \quad \gamma_p = 1 + \alpha_p
\end{align}
we have 
\begin{align}
\label{eq:bounds-sigma}
\| \sigma_p \|_{\ell^2}, \| \alpha_p \|_{\ell^2}, \| \beta_p \|_{\ell^2} \leq C, \quad \text{and} \quad \| \gamma_p \|_{\ell^\infty} \leq C \; . 
\end{align}
The remainders $d_p, d_p^*$ satisfy (following from \cite[Lemma 2.3]{BCCS}) for any $k \in \mathbb{Z}$  and all $p \in \Lambda_+^*$
\begin{align}
\label{eq:estimates-dp}
\| ( \mathcal{N}_+  + 1)^{k/2} d_p \psi \| \leq C_k N^{-1} \left( \| b_p ( \mathcal{N}_+ 1)^{(k+2)/2} \psi \| + \vert \mu_p \vert \; \| ( \mathcal{N}_+  + 1)^{3/2} \psi \| \right) 
\end{align}
and 
\begin{align}
\label{eq:estimates-dp*}
\| ( \mathcal{N}_+  + 1)^{k/2} d_p^* \psi \| \leq C_k N^{-1} \| ( \mathcal{N}_+  + 1)^{3/2} \psi \| \; . 
\end{align}
In the proof it will turn out to be useful to estimate some of the terms in position space. For this we define the remainders $\check{d}_x, \check{d}_x^*$ in position space by 
\begin{align}
\label{def:dx}
e^{-B (\eta)} \check{b}_x e^{B ( \eta)} = b( \check{\gamma}_x ) + b^*( \check{\sigma}_x ) + \check{d}_x, \quad e^{-B (\eta)} \check{b}_x^* e^{B ( \eta)} = b^*( \check{\gamma}_x ) + b( \check{\sigma}_x ) + \check{d}_x^*
\end{align}
with $\check{\gamma}_x(y) = \sum_{q \in \Lambda^*} \cosh( \eta_q) e^{-iq \cdot (x-y)}$ and $\check{\sigma}_x(y) = \sum_{q \in \Lambda^*} \sinh( \eta_q) e^{-iq \cdot (x-y)}$. It follows (see for example \cite[Eq. (3.20)-(3.21)]{BCCS}) that with the splitting 
\begin{align}
\check{\gamma}_x = \delta_x + \check{\alpha}_x, \quad \check{\sigma}_x = \check{\eta}_x + \check{\beta}_x
\end{align}
we have 
\begin{align}
\label{eq:sigma-x-bounds}
\| \alpha \|_{L^2 ( \Lambda \times \Lambda)}, \| \sigma \|_{L^2 ( \Lambda \times \Lambda)}, \| \beta \|_{L^2 ( \Lambda \times \Lambda)}  \leq C, \quad \text{and} \quad  \| \gamma \|_{L^\infty ( \Lambda \times \Lambda )} \leq C \; . 
 \end{align}

From \cite[Lemma 3.1]{BCCS_optimal} we have 
\begin{align}
\label{eq:estimates-dd}
\| ( \mathcal{N}_+ + 1)^{k/2}  \check{d}_x \check{d}_y \psi \| \leq  CN^{-2} \Big[ & \| \eta \|^2 \| ( \mathcal{N}_+ +1 )^{(k+6)/2} \psi \| + \| \eta \| \vert \check{\eta}(x-y) \vert \; \| ( \mathcal{N}_+ + 1)^{(k+6)/2} \psi \| \notag \\
&+ \| \eta \|^2 \| a_x ( \mathcal{N}_+ + 1)^{(k+5)/2} \psi \| + \| \eta \|^2 \| a_y ( \mathcal{N}_+ + 1)^{(k+5)/2} \psi \|  \notag \\
&+ \| \eta\|^2 \| a_xa_y ( \mathcal{N}_+ + 1)^{(k+4)/2} \psi \| \Big] 
\end{align}
and 
\begin{align}
\label{eq:estimates-db}
\| ( \mathcal{N}_+ + 1)^{k/2}  \check{b}_x \check{d}_y \psi \| \leq  CN^{-1} \Big[ & \| \eta \|^2 \| ( \mathcal{N}_+ +1 )^{(k+4)/2} \psi \| + \| \eta \| \vert \check{\eta}(y-x)) \vert \; \| ( \mathcal{N}_+ + 1)^{(k+4)/2} \psi \| \notag \\
&+   \| \eta \| \vert  \| a_x ( \mathcal{N}_+ + 1)^{(k+3)/2} \psi \| + \| \eta\|^2 \| a_xa_y ( \mathcal{N}_+ + 1)^{(k+2)/2} \psi \| \Big] \; . 
\end{align}

In particular, it follows from \cite[Corollary 3.5]{BCCS_optimal}), that these estimates \eqref{eq:estimates-dp}, \eqref{eq:estimates-dp*}, \eqref{eq:estimates-dd} remain true when replacing $d_p, d_p^*$ resp. $d_p^{\sharp_1}d_{\alpha p}^{\sharp_2}$ with their (double commutator) with $\cN_+$:
\begin{align}
\label{eq:bound-comm}
\| ( \mathcal{N}_+  + 1)^{k/2} \left[ \mathcal{N}_+,  d_p \right]  \psi \| \leq C_k  N^{-1} \left( \| b_p ( \mathcal{N}_+ + 1)^{(k+2)/2} \psi \| + \vert \mu_p \vert \; \| ( \mathcal{N}_+  + 1)^{3/2} \psi \| \right) 
\end{align} 
resp. 
\begin{align}
\label{eq:bound-double-comm}
\| ( \mathcal{N}_+  + 1)^{k/2} \left[ \mathcal{N}_+, \left[ \mathcal{N}_+,  d_p \right] \right] \psi \| \leq C_k N^{-1} \left( \| b_p ( \mathcal{N}_+ + 1)^{(k+2)/2} \psi \| + \vert \mu_p \vert \; \| ( \mathcal{N}_+  + 1)^{3/2} \psi \| \right) 
\end{align}
and similarly for the other operators. For our proof we need refined estimates for the remainder terms. More precisely we need to control single and double commutators with $e^{\kappa \cN_+}$. In the next subsection we show how to control these (double) commutators.

\subsection{Preliminary estimates} 
\label{subsec:prel-est}

We collect some preliminary results on commutators with the exponential of the number of excitations that we need to prove Proposition \ref{prop:G}.  For this we first introduce some more notation. For $k \in \mathbb{N}$ and $p_{i} \in \Lambda_+^*$ with $i \in \lbrace 1, \dots, k \rbrace $, let $B_{p_1, \dots, p_k}$ denote an operator of the form 
\begin{align}
\label{def:Ak}
B_{p_1, \dots, p_k} = b^{\sharp_1}_{p_1} \dots b^{\sharp_k}_{p_k}
\end{align}
where $\sharp_i \in \lbrace \cdot, * \rbrace$. Then we define $\natural^* (B_{p_1, \dots, p_k})$ (resp. $\natural^\cdot(B_{p_1, \dots, p_k})$) by the number of creation (resp. annihilation) operators of $B_{p_1, \dots, p_k}$, and by 
\begin{align}
\sharp( B_k)  := \natural^* (B_{p_1, \dots, p_k}) - \natural^\cdot (B_{p_1, \dots, p_k}) \; . 
\end{align}
their difference. For the proof of Proposition \ref{prop:G} we will need to control the second nested commutator with respect to $e^{\kappa \cN}$. The next Lemma provides a formula to control such commutators w.r.t. to operators of the form $B_{p_1, \dots, p_k}$.

\begin{lemma}
\label{lemma:double}
For $k \in \mathbb{N}$ let $B_{p_1, \dots, p_k}$ be defined as in \eqref{def:Ak}. Then for $\sharp \in \lbrace \cdot, * \rbrace$ we have 
\begin{align}
\left[ e^{\kappa \mathcal{N}_+}, \; B_{p_1, \dots, p_k} \right] =& 2 e^{ - \natural( B_{p_1, \dots, p_k}) \kappa/2}  \sinh ( \natural( B_{p_1, \dots, p_k}) \kappa /2 )  e^{\kappa \mathcal{N}_+} B_{p_1, \dots, p_k},  \notag \\ 
\left[ e^{\kappa \mathcal{N}_+}, \; B_{p_1, \dots, p_k} \right] =& 2 e^{  \natural( B_{p_1, \dots, p_k}) \kappa/2}  \sinh ( \natural( B_{p_1, \dots, p_k}) \kappa /2 )  B_{p_1, \dots, p_k} e^{\kappa \mathcal{N}_+}, \label{eq:comm-Ak-eins}
\end{align}
and furthermore
\begin{align}
\left[ e^{\kappa \mathcal{N}_+}, \; \left[ e^{\kappa \mathcal{N}_+}, \; B_{p_1, \dots, p_k} \right]\right] 
&= 4 \sinh^2 ( \natural( B_{p_1, \dots, p_k}) \kappa /2 ) \; e^{\kappa \mathcal{N}_+} B_{p_1, \dots, p_k}  \; e^{\kappa \mathcal{N}_+} \; . \label{eq:comm-Ak-zwei}
\end{align}
\end{lemma}

\begin{proof} The Lemma is an immediate consequence of the commutation relations \eqref{eq:comm} that show  
\begin{align}
\left[ e^{\kappa \mathcal{N}_+}, \; B_{p_1, \dots, p_k} \right] =& \left( 1- e^{ - \natural( B_k) \kappa} \right) e^{\kappa \mathcal{N}_+} B_{p_1, \dots, p_k},  \notag \\
 \left[ e^{\kappa \mathcal{N}_+}, \; B_{p_1, \dots, p_k} \right] =& \left( e^{  \natural( B_k) \kappa}- 1 \right)  B_{p_1, \dots, p_k} e^{\kappa \mathcal{N}_+} \; 
\end{align}
yielding the desired identities \eqref{eq:comm-Ak-eins}. Furthermore we have 
\begin{align}
\left[ e^{\kappa \mathcal{N}_+},  \left[ e^{\kappa \mathcal{N}_+}, \; B_{p_1, \dots, p_k} \right]\right] =& \left( 1- e^{ - \sharp( B_{p_1, \dots, p_k}) \kappa} \right) \left( e^{ \natural( B_{p_1, \dots, p_k}) \kappa} -1 \right)e^{\kappa \mathcal{N}_+} B_{p_1, \dots, p_k} \; e^{\kappa \mathcal{N}_+} \notag \\
&= 4 \sinh^2 ( \natural( B_{p_1, \dots, p_k}) \kappa /2 ) \; e^{\kappa \mathcal{N}_+} B_{p_1, \dots, p_k} \; e^{\kappa \mathcal{N}_+} \; 
\end{align}
and thus identity \eqref{eq:comm-Ak-zwei} follows. 
\end{proof}

In particular it follows from Lemma \ref{lemma:double} that 
\begin{align*}
\| \left[ e^{\kappa \mathcal{N}_+}, \; B_{p_1, \dots, p_k} \right] \psi \| &\leq C \kappa \| B_k e^{\kappa \cN_+}\psi \| \\
\| e^{- \kappa \cN_+}\left[ e^{\kappa \mathcal{N}_+},   \left[ e^{\kappa \mathcal{N}_+}, \; B_{p_1, \dots, p_k} \right] \right] \psi \| &\leq C \kappa^2 \| B_k ^{\kappa \cN_+}\psi \|  \; . 
\end{align*}

Next we prove some similar properties for the remainders $d_p^*,d_p$ of the generalized Bogoliubov transform defined in \eqref{eq:action-bogo}. More precisely, we consider commutators of the form 
 \begin{align}
[ e^{ \kappa \mathcal{N}_+},  [ e^{\lambda \kappa \mathcal{N}_+}, d_p^{\sharp} ]] 
 \end{align}
with $\sharp \in \lbrace \cdot, * \rbrace$ and $\kappa \in \mathbb{R}$. For this, we use properties of $d_p,d_p^*$ proven in \cite{BCCS} that are based on the expansion  
\begin{align}
e^{-B(\eta)} b_p e^{B( \eta)} =& \sum_{n=1}^{m-1} (-1)^n \frac{{\rm ad}_{B(\eta)}^{(n)} (b_p)}{n!} \notag \\
 &+ \int_0^1 ds_1 \int_0^{s_1}  ds_2\dots \int_0^{s_{m-1}} ds_m e^{-s_m B( \eta)}{\rm ad}_{B( \eta)}^{(m)} (b_p) e^{s_m B( \eta)} \; . 
\end{align}
The nested commutators are defined recursively through  
\begin{align}
{\rm ad}_{B ( \eta)}^{(0)} (A) = A \quad \text{and} \quad {\rm ad}_{B( \eta)}^{(n)} = \left[ B( \eta), {\rm ad}_{B (\eta)}^{(n-1)} (A) \right] \; . 
\end{align}
It follows from \cite[Lemma 3.2]{BCCS_optimal} that the nested commutators of $b_p,b_p^*$ are given in terms of the following operators: For $f_1, \dots , f_n \in \ell^2( \Lambda_+^*)$, $\sharp = ( \sharp_1, \dots, \sharp_n), \flat = (\flat_0, \dots, \flat_{n-1}) \in \lbrace \cdot, * \rbrace^n $ we define the $\Pi^{(2)}$-operator of order $n $ by 
\begin{align}
\label{def:Pi2}
\Pi_{\sharp, \flat}^{(2)} (f_1, \dots, f_n ) = \sum_{p_1, \dots, p_n \in \Lambda_+^*} b_{\alpha_0 p_1}^{\flat_0} a_{\beta_1 p_1}^{\sharp_1} a_{\alpha_1 p_2}^{\flat_1} a_{\beta_2 p_2}^{\sharp_2} a_{\alpha_2 p_3}^{\flat_2} \dots a_{\beta_{n-1} p_{n-1}}^{\sharp_{n-1}} a_{\alpha_{n-1} p_n}^{\flat_{n-1}} b_{\beta_n p_n}^{\sharp_n} \prod_{\ell=1}^n f_\ell (p_\ell ) 
\end{align}
were for $\ell = 0,1, \dots, n $ we define $\alpha_\ell =1$ if $\flat_\ell =*$., $\alpha_\ell = -1$ if $\flat_\ell = \cdot$, $\beta_\ell =1$ if $\sharp_\ell =\cdot$ and $\beta_\ell = -1$ of $\sharp_\ell = *$. Moreover, we require that for every $j=1, \dots, n-1$ we have either $\sharp_j = \cdot$ and $\flat_j = *$ or $\sharp_j = *$ and $\flat_j = \cdot$ (so that the product $a_{\beta_\ell p_\ell}^{\sharp_\ell} a_{\alpha_\ell p_{\ell + 1}}^{\flat_\ell}$ preserves the number of particles for all $\ell =1, \dots, n-1$). Then, the operator $\Pi_{\sharp, \flat}^{(2)} (f_1, \dots, f_n ) $ leaves the truncated Fock space invariant. Moreover if for some $\ell =1, \dots, n$, $\flat_{\ell -1} = \cdot$ and $\sharp_\ell =*$, we furthermore require that $f_\ell \in \ell^1 ( \Lambda_+^*)$ (so that we can normal order the operators).  
For $g, f_1, \dots, f_n \in \ell^2( \Lambda_+^*)$, $\sharp = ( \sharp_1, \dots, \sharp_n)\in \lbrace \cdot, *\rbrace^n, \flat = ( \flat_0, \dots, \flat_n) \in \lbrace \cdot, * \rbrace^{n+1}$ we define a $\Pi^{(1)}$-operator of order $n$ by 
\begin{align}
&\Pi^{(1)}_{\sharp,\flat} (f_1, \dots, f_n; g) \notag \\
&\quad = \sum_{p_1, \dots, p_n \in \Lambda_+^*} b_{\alpha_0, p_1}^{\flat_0} a_{\beta_1 p_1}^{\sharp_1} a_{\alpha_1 p_2}^{\flat_1} a_{\beta_2p_2}^{\sharp_2}a_{\alpha_2 p_3}^{\flat_2} \dots a_{\beta_{n-1}p_{n-1}}^{\sharp_{n-1}} a_{\alpha_{n-1} p_n}^{\flat_{n-1}} a_{\beta_n p_n}^{\sharp_n} a^{\flat_n} (g) \prod_{\ell = 1}^n f_\ell (p_\ell)
\end{align}
where $\alpha_\ell$ and $\beta_\ell$ are defined as before. Also here, we require that for all $\ell= 1, \dots, n$ either $\sharp_\ell = \cdot$ and $\flat_\ell =*$ or $\sharp = *$ and $\flat_\ell = \cdot$. Note that the $\Pi^{(1)}$ leaves the truncated Fock space invariant. We require that $f_\ell \in \ell^1 ( \Lambda_+^*)$ if $\flat_{\ell -1} = \cdot$ and $\sharp_\ell = *$ for some $\ell=1, \dots, n$. It follows from \cite{BS} that nested commutators ${\rm ad}_{B(\eta)} (b_p)$ can be expressed in the following form. 

\begin{lemma}[Lemma 3.2 \cite{BS}]
\label{lemma:dp}
Let $\eta \in \ell^2( \Lambda_+^*)$ be such that $\eta_p = \eta_{-p}$ for all $p \in \ell^2( \Lambda_+^*)$. To simplify the notation, assume also $\eta$ to be real valued. Let $B( \eta)$ be defined as in \eqref{def:bogo}, $n \in \mathbb{N}$ and $p \in \Lambda_+^*$. Then the nested commutator ${\rm ad}_{B(\eta)}^{(n)} (b_p)$ can be written as the sum of exactly $2^n n!$ terms with the following properties. 

\begin{enumerate}
\item[(i)] Possibly up to a sign, each term has the form 
\begin{align}
\label{eq:lemma(i)}
\Lambda_1 \Lambda_2 \dots \Lambda_i N^{-k} \Pi_{\sharp, \flat}^{(1)} ( \eta^{j_1}. \dots, \eta^{j_k}; \eta_p^s \varphi_{\alpha p} )
\end{align}
for some $i,k,s \in \mathbb{N}$, $j_1, \dots, j_k \in \mathbb{N} \setminus \lbrace 0 \rbrace$, $\sharp \in \lbrace \cdot, * \rbrace^k, \flat \in \lbrace \cdot, * \rbrace^{k+1}$ and $\alpha \in \lbrace \pm \rbrace$ chosen so that $\alpha =1$ if $\flat_k = \cdot$ and $\alpha =-1$ of $\flat_k = *$ (recall that $\varphi_p (x) = e^{-ip \cdot x}$). In \eqref{eq:lemma(i)} each operator $\Lambda_w : \mathcal{F}^{\leq N} \rightarrow \mathcal{F}^{\leq N}$, $w=1, \dots, i$ is either a factor of $(N-\mathcal{N}_+)/N$, a factor $( N - ( \mathcal{N}_+ - 1))/N$ or an operator of the form 
\begin{align}
\label{eq:lemma(ii)}
N^{-h}\Pi_{\sharp',\flat'}^{(2)} ( \eta^{z_1}, \eta^{z_2}, \dots, \eta^{z_h} )
\end{align}
for some $h,z_1, \dots, z_h \in \mathbb{N} \setminus \lbrace 0 \rbrace, \sharp, \beta \in \lbrace \cdot, * \rbrace^h$. 
\item[(ii)]  If a term of the form \eqref{eq:lemma(i)} cantains $m \in \mathbb{N}$ factors $(N-\mathcal{N}_+)/N$ or $( N- ( \mathcal{N}_+ +1))/N$ and $j \in \mathbb{N}$ factors of te form \eqref{eq:lemma(i)} with $\Pi^{(2)}$ operators pf order $h_1, \dots, h_j \in \mathbb{N} \setminus \lbrace 0 \rbrace$, then we have 
\begin{align}
m+ (h_1 + 1) + \dots + (h_j +1) + (k+1) = n+1 
\end{align}
\item[(iii)]  If a term of the form \eqref{eq:lemma(i)} contains (considering all $\Lambda$-operators and the $\Pi^{(1)}$-operator) the arguments $\eta^{i_1}, \dots, \eta^{i_m}$ and the factor $\eta_p^s$ for some $m,s \in \mathbb{N}$ and $i_1, \dots, i_m \in \mathbb{N} \setminus \lbrace 0 \rbrace$, then 
\begin{align}
i_1 + \dots + i_m + s =n \; . 
\end{align}
\item[(iv)]  There is exactly one term having the form \eqref{eq:lemma(i)} with $k=0$ and such that all $\Lambda$-operators are factors of $(N-\mathcal{N}_+)/N$ or of $( N+1-\mathcal{N})/N$. It is given by 
\begin{align}
\left( \frac{N-\mathcal{N}_+}{N}\right)^{n/2} \left( \frac{N+1-\mathcal{N}_+}{N} \right)^{n/2} \eta_p^n b_p 
\end{align}
if $n$ is even, and by 
\begin{align}
- \left( \frac{N-\mathcal{N}_+}{N}\right)^{(n+1)/2} \left( \frac{N+1-\mathcal{N}_+}{N} \right)^{(n-1)/2} \eta_p^n b_{-p}^* 
\end{align}
if $n$ is odd. 
\item[(v)]  If the $\Pi^{(1)}$-operator in \eqref{eq:lemma(i)} is of order $k \in \mathbb{N}\setminus \lbrace 0 \rbrace$, it has either the form 
\begin{align}
\sum_{p_1, \dots, p_k} b_{\alpha_0 p_1}^{\flat_0} \prod_{i=1}^{k-1} a_{\beta_ip_i}^{\sharp_i} a_{\alpha_i p_{i+1}}^{\flat_i} a_{-p_k}^* \eta_{p}^{2r} a_p \prod_{i=1}^k \eta_{p_i}^{j_i} 
\end{align}
or the form 
\begin{align}
\sum_{p_1, \dots, p_k} b_{\alpha_0p_1}^{\flat_0} \prod_{i=1}^{k-1} a^{\sharp_i}_{\beta_i p_i} a_{\alpha_i p_{i+1}}^{\flat_i} a_{p_k}\eta_p^{2r+1} a_p^* \prod_{i=1}^k
 \eta_{p_i}^{j_i}
 \end{align}
for some $r \in \mathbb{N}$, $j_1, \dots, j_k \in \mathbb{N} \setminus \lbrace 0 \rbrace$. If it is of order $k=0$, then it is either given by $\eta_p^{2r}b_p$ or by $\eta_p^{2r+1} b_{-p}^*$ for some $r\in \mathbb{N}$. 

\item[(vi)] For every non-normally ordered term of the form 
\begin{align}
\sum_{q \in \Lambda^*} \eta_q^i a_q a_q^* , \quad \sum_{q \in \Lambda^*} \eta_q^i b_q a_q^*, \quad \sum_{q \in \Lambda^*} \eta_q^i a_q b_q^* \quad \text{or} \quad \sum_{q \in \Lambda^*} \eta_q^i b_q b_q^* 
\end{align}
appearing either in the $\Lambda$-operators or in the $\Pi^{(1)}$-operator in \eqref{eq:lemma(i)}, we have $i\geq 2$.
\end{enumerate}
\end{lemma}

Lemma \ref{lemma:dp} in particular shows that for small enough $\| \eta \| $ the series 
\begin{align}
e^{-B( \eta)} b_p e^{B (\eta)} = \sum_{n=0}^\infty \frac{(-1)^n}{n!} {\rm ad}_{B ( \eta)}^{(n)} ( b_p ), \quad e^{-B( \eta)} b_p^* e^{B (\eta)} = \sum_{n=0}^\infty \frac{(-1)^n}{n!} {\rm ad}_{B ( \eta)}^{(n)} ( b_p^* ) 
\end{align}
converge absolutely (see \cite[Lemma 3.3]{BCCS_optimal}) and we get an explicitly definition of the remainders by 
\begin{align}
\label{def:d-2}
d_p = \sum_{m \geq 0} \frac{1}{m!} \left[ {\rm ad}_{-B(\eta)}^{(m)} (b_p) - \eta_p^m b_{\alpha_m p}^{\sharp_m} \right], \quad d_p^*= \sum_{m \geq 0} \frac{1}{m!} \left[ {\rm ad}_{-B(\eta)}^{(m)} (b_p^*) - \eta_p^m b_{\alpha_m p}^{\sharp_{m+1}} \right]
\end{align}
where $p \in \Lambda_+^*$, $( \sharp_m, \alpha_m) = (\cdot, +1)$ if $m$ is even and $(\sharp_m, \alpha_m)= (*,-1)$ if $m$ is odd. Moreover we use this represenation to prove the following Lemma.

\begin{lemma}
\label{lemma:eN-dp}
Under the same assumptions and notations of Lemma \ref{lemma:dp}, we have for $0<\lambda<1$ and sufficiently small $\| \eta \| $ and $k \in \mathbb{Z}$
\begin{align}
\label{eq:eN-dp}
\|    ( \mathcal{N}_+  & +1)^{k/2} \left[ e^{\lambda \cN_+}, d_p \right] \psi \| \notag \\
& \leq C \lambda N^{-1} \left(  \| b_p ( \mathcal{N}_+ +1)^{(k+2)/2} e^{\lambda \cN_+}\psi \| + \vert\eta_p \vert \| ( \mathcal{N}_+ +1)^{(3+k)/2} ) e^{\lambda \cN_+} \psi \| \right), \notag \\
\|   ( \mathcal{N}_+  & +1)^{k/2}\left[ e^{\lambda \cN_+}, d_p^* \right] \psi \| 
 \leq C \lambda  N^{-1} \|  ( \mathcal{N}_+ +1)^{(k+3)/2} e^{\lambda \cN_+} \psi \|, 
\end{align}
and 
\begin{align}
\label{eq:eN-dp-double-comm}
\|  ( \mathcal{N}_+  & +1)^{k} e^{- \lambda \cN_+} \left[ e^{\lambda\cN_+}, \left[ e^{\lambda \cN_+}, d_p \right] \right] \psi \|\notag \\
 &\leq C \lambda^2 N^{-1} \left(  \| b_p ( \mathcal{N}_+ +1)^{(k+2)/2} e^{\lambda \cN_+}\psi \| + \vert\eta_p \vert \| ( \mathcal{N}_+ +1)^{(k+3)/2} ) e^{\lambda \cN_+} \psi \| \right), \notag \\
\| ( \mathcal{N}_+  &+1)^{k}e^{- \lambda \cN_+}  \left[ e^{\lambda\cN_+}, \left[ e^{\lambda \cN_+}, d_p^* \right] \right] \psi \| \leq C \lambda^2  N^{-1} \|  ( \mathcal{N}_+ +1)^{(k+3)/2} e^{\lambda \cN_+} \psi \| \; . 
\end{align}
Furthermore, the operators $\check{d}_x, \check{d}^*_x$ defined by \eqref{def:dx} satisfy 
\begin{align}
\label{eq:eN-dx-single}
\|   ( \mathcal{N}_+  & +1)^{k/2} [e^{\lambda \cN_+}, \check{d}_x \check{d}_y] \psi \|  \notag \\
&\leq C \lambda N^{-2} \Big[  \| \eta \|^2 \| ( \mathcal{N}_+ +1 )^{(k+6)/2} e^{\lambda \cN_+} \psi \| + \| \eta \| \vert \check{\eta}(x-y) \vert \; \| ( \mathcal{N}_+ + 1)^{(k+4)/2} e^{\lambda \cN_+}\psi \| \notag \\
&\hspace{3cm} + \| \eta \|^2 \| a_x ( \mathcal{N}_+ + 1)^{(k+5)/2} e^{\lambda \cN_+}\psi \| + \| \eta \|^2 \| a_y ( \mathcal{N}_+ + 1)^{(k+4)/2} e^{\lambda \cN_+}\psi \|  \notag \\
&\hspace{3cm} + \| \eta\|^2 \| a_xa_y ( \mathcal{N}_+ + 1)^{(k+4)/2} e^{\lambda \cN_+} \psi \| \Big] \; 
\end{align}
and 
\begin{align}
\label{eq:eN-dx-double}
\| & ( \mathcal{N}_+ +1)^{k/2} e^{-\lambda \cN_+}\Big[ e^{\lambda \cN_+} \Big[ e^{\lambda \cN_+}, \check{d}_x \check{d}_y\Big]\Big]  \psi \| \notag \\
&\leq C \lambda^2 N^{-2} \Big[  \| \eta \|^2 \| ( \mathcal{N}_+ +1 )^{(k+6)/2} e^{\lambda \cN_+} \psi \| + \| \eta \| \vert \check{\eta}(x-y) \vert \; \| ( \mathcal{N}_+ + 1)^{(k+4)/2} e^{\lambda \cN_+}\psi \| \notag \\
&\hspace{3cm} + \| \eta \|^2 \| a_x ( \mathcal{N}_+ + 1)^{(k+5)/2} e^{\lambda \cN_+}\psi \| + \| \eta \|^2 \| a_y ( \mathcal{N}_+ + 1)^{(k+4)/2} e^{\lambda \cN_+}\psi \|  \notag \\
&\hspace{3cm} + \| \eta\|^2 \| a_xa_y ( \mathcal{N}_+ + 1)^{(k+4)/2} e^{\lambda \cN_+} \psi \| \Big] \; 
\end{align}
Moreover,
\begin{align}
\label{eq:eN-db-single}
&\| ( \mathcal{N}_+  + 1)^{k/2} \Big[ e^{\lambda \cN_+},  \check{b}_x \check{d}_y \Big]  \psi \| \notag \\
& \leq  C \lambda N^{-1}  \Big[  \| \eta \|^2 \| ( \mathcal{N}_+ +1 )^{(k+4)/2}  e^{\lambda \mathcal{N}_+}\psi \| + \| \eta \| \vert \check{\eta}(y-x)) \vert \; \| ( \mathcal{N}_+ + 1)^{(k+4)/2} e^{\lambda \mathcal{N}_+}\psi \| \notag \\
&\quad  + \| \eta \| \vert  \| a_x ( \mathcal{N}_+ + 1)^{(k+3)/2}e^{\lambda \mathcal{N}_+} \psi \| + \| \eta\|^2 \| a_xa_y ( \mathcal{N}_+ + 1)^{(k+2)/2} e^{\lambda \mathcal{N}_+}\psi \| \Big] 
\end{align}
and 
\begin{align}
\label{eq:eN-db-double}
&\| ( \mathcal{N}_+  + 1)^{k/2} e^{- \lambda \cN_+} \Big [ e^{\lambda \cN_+}, \Big[ e^{\lambda \cN_+},  \check{b}_x \check{d}_y \Big] \Big] \psi \| \notag \\
&\leq  C \lambda N^{-1}  \Big[  \| \eta \|^2 \| ( \mathcal{N}_+ +1 )^{(k+4)/2} e^{\lambda \mathcal{N}_+}\psi \| + \| \eta \| \vert \check{\eta}(y-x)) \vert \; \| ( \mathcal{N}_+ + 1)^{(k+4)/2} e^{\lambda \mathcal{N}_+}\psi \| \notag \\
&\quad + \| \eta \| \vert  \| a_x ( \mathcal{N}_+ + 1)^{(k+3)/2}e^{\lambda \mathcal{N}_+} \psi \| + \| \eta\|^2 \| a_xa_y ( \mathcal{N}_+ + 1)^{(k+2)/2} e^{\lambda \mathcal{N}_+}\psi \| \Big] . 
\end{align}
\end{lemma}

\begin{proof}
We start with proving \eqref{eq:eN-dp}. 
We find from \eqref{def:d-2} that 
\begin{align} \label{eq:term10}
\| & \left[ e^{\lambda \cN_+}, d_p \right] \psi \| = \| \left( e^{\lambda \cN_+}d_p e^{-\lambda \cN_+} - d_p \right) e^{ \lambda \cN_+}\psi \| \\
&\leq  \sum_{m \geq 0} \frac{1}{m!}  \Big\|\left( e^{\lambda \cN_+}\left[ {\rm ad}_{-B(\eta)}^{(m)} (b_p) - \eta_p^m b_{\alpha_m p}^{\sharp_m} \right] e^{-\lambda \cN_+}- \left[ {\rm ad}_{-B(\eta)}^{(m)} (b_p) - \eta_p^m b_{\alpha_m p}^{\sharp_m} \right] \right) e^{ \lambda \cN_+} \psi \Big\|  .\notag
\end{align}
Moreover, by Lemma \ref{lemma:dp} the difference 
\begin{align}
e^{\lambda \cN_+}\left[ {\rm ad}_{-B(\eta)}^{(m)} (b_p) - \eta_p^m b_{\alpha_m p}^{\sharp_m} \right] e^{-\lambda \cN_+}- \left[ {\rm ad}_{-B(\eta)}^{(m)} (b_p) - \eta_p^m b_{\alpha_m p}^{\sharp_m} \right]
\end{align}
is the sum of one term of the form 
\begin{align}
A_p =  e^{\lambda \cN_+} & \left( \frac{N-\cN_+}{N}\right)^{\frac{m+(1-\alpha_m)/2}{2}} \left( \frac{N+1-\cN_+}{N}\right)^{\frac{m+(1+\alpha_m)/2}{2}} \eta_p b_{\alpha_m p}^{\sharp_m}  e^{-\lambda \cN_+} \notag \\
 &- \left( \frac{N-\cN_+}{N}\right)^{\frac{m+(1-\alpha_m)/2}{2}} \left( \frac{N+1-\cN_+}{N}\right)^{\frac{m+(1+\alpha_m)/2}{2}} \eta_p b_{\alpha_m p}^{\sharp_m} \label{eq:term11}
\end{align}
and $2^m m! - 1$ terms of the form 
 \begin{align}
B_p =  e^{\kappa \lambda \cN_+}  & \Lambda_1 \dots \Lambda_{i_1} N^{-k} \Pi_{\sharp, \flat}^{(1)} ( \eta^{j_1}, \dots, \eta^{j_{k_1}}; \eta_p^{\ell_1} \varphi_{\alpha_{\ell_1} p} g_p ) e^{- \lambda \kappa\cN_+ } \notag \\
 &- \Lambda_1 \dots \Lambda_{i_1} N^{-k} \Pi_{\sharp, \flat}^{(1)} ( \eta^{j_1}, \dots, \eta^{j_{k_1}}; \eta_p^{\ell_1} \varphi_{\alpha_{\ell_1} p} ) \label{eq:term12}
\end{align} 
where $i_1, k_1, \ell_1 \in \mathbb{N}, j_1, \dots, j_k \in \mathbb{N} \setminus \lbrace 0 \rbrace$ and where each operator $\Lambda_r$ is either a factor $( N- \cN_+)/N$, a factor $(N+1-\cN_+)/N$  or a $\Pi^{(2)}$ operator of the form 
\begin{align}
\label{eq:term13}
N^{-h}\Pi^{(2)}_{\sharp,\flat} ( \eta^{z_1}, \dots, \eta^{z_h} ) 
\end{align}
with $h, z_1, \dots, z_h \in \mathbb{N} \setminus \lbrace 0 \rbrace$. We consider \eqref{eq:term11} and \eqref{eq:term12} separately, thus each term that is of the form \eqref{eq:term11} either has $k_1 >0$ or contains at least one operator of the form \eqref{eq:term13}. We start with estimating \eqref{eq:term11} first that vanishes for $m=0$. Thus we have 
\begin{align}
 & \|  A_p e^{ \lambda \cN_+}\psi \| \notag \\
 &=  \Big\| \left( \frac{N-\cN_+}{N}\right)^{\frac{m+(1-\alpha_m)/2}{2}} \left( \frac{N+1-\cN_+}{N}\right)^{\frac{m+(1+\alpha_m)/2}{2}} \eta_p^m \left(   e^{\lambda \cN_+}  b_{\alpha_m p}^{\sharp_m} e^{-\lambda \cN_+} -  b_{\alpha_m p}^{\sharp_m}\right) e^{ \lambda \cN_+}\psi \| \notag \\ 
&\leq  \kappa \lambda C^m \vert \eta_p \vert^m N^{-1}\|( \mathcal{N}_+ + 1)^{3/2}e^{ \lambda \cN_+} \psi \| \; . \label{eq:term14}
\end{align}
For \eqref{eq:term12} we find 
\begin{align}
& B_p = \sum_{u=1}^i \left(  \prod_{t=1}^{u-1} e^{\lambda \cN_+}  \Lambda_t e^{-\lambda \cN_+} \right) \left( e^{\lambda \cN_+}  \Lambda_u e^{-\lambda \cN_+} - \Lambda_u \right) \times \nonumber \\
&\qquad \times \prod_{t=u+1}^{i}   \Lambda_t  N^{-k} \Pi_{\sharp, \flat}^{(1)} ( \eta^{j_1}, \dots, \eta^{j_{k_1}}; \eta_p^{\ell_1} \varphi_{\alpha_{\ell_1} p} )  \\
+& \left( \prod_{t=1}^{i}   \Lambda_t \right)   N^{-k}  \left(  e^{\lambda \cN_+}\Pi_{\sharp, \flat}^{(1)} ( \eta^{j_1}, \dots, \eta^{j_{k_1}}; \eta_p^{\ell_1} \varphi_{\alpha_{\ell_1} p} ) e^{-\lambda \cN_+} - \Pi_{\sharp, \flat}^{(1)} ( \eta^{j_1}, \dots, \eta^{j_{k_1}}; \eta_p^{\ell_1} \varphi_{\alpha_{\ell_1} p} )\right) \;. \nonumber
\end{align}
In case $\Lambda_u$ is of the form $( N- \mathcal{N}_+)/N $ or $( N +1- \mathcal{N}_+)/N$ then $e^{\lambda \cN_+}  \Lambda_u e^{-\lambda \cN_+} - \Lambda_u $ vanishes. Otherwise, if $\Lambda_u$ is an operator of the form $\Pi^{(2)}$ it creates resp. annihilates two particles, thus, we have $e^{\lambda \cN_+}  \Lambda_u e^{-\lambda \cN_+} - \Lambda_u =  ( e^{\lambda \kappa_u} -1)\Lambda_u$ with $\kappa_u =2$ or $\kappa_u = -2$. Similarly, as the operator $\Pi^{(1)}$ creates or annihilates one particle, we have 
\begin{align}
\Pi_{\sharp, \flat}^{(1)}  &( \eta^{j_1}, \dots, \eta^{j_1}, \dots, \eta^{j_{k_1}}; \eta_p^{\ell_1} \varphi_{\alpha_{\ell_1} p} ) e^{-\lambda \cN_+} - \Pi_{\sharp, \flat}^{(1)} ( \eta^{j_1}, \dots, \eta^{j_{k_1}}; \eta_p^{\ell_1} \varphi_{\alpha_{\ell_1} p} ) \notag \\
=&(e^{\lambda \kappa} - 1) \Pi_{\sharp, \flat}^{(1)} ( \eta^{j_1}, \dots,   \eta^{j_1}, \dots, \eta^{j_{k_1}}; \eta_p^{\ell_1} \varphi_{\alpha_{\ell_1} p} ) 
\end{align}
with $\kappa = 1$ or $\kappa=-1$. Therefore we find 
\begin{align}
\Big\| B_p e^{ \lambda \cN_+} \psi \Big\| \leq& \left(\sum_{u=1}^i (e^{\kappa_u} -1)+ (e^{\kappa}-1) \right) \|  \prod_{t=1}^{i}   \Lambda_t N^{-k} \Pi_{\sharp, \flat}^{(1)} ( \eta^{j_1}, \dots, \eta^{j_{k_1}}; \eta_p^{\ell_1} \varphi_{\alpha_{\ell_1} p} ) \psi \| \; .
\end{align}
We consider the case $\ell_1 =0$ and $\ell_1 >0$ separately (see for example \cite[Lemma 3.4]{BCCS_optimal} resp. \cite[Section 5]{BS}) and arrive with $\vert \eta_p \vert \leq \| \eta \|$ at 
\begin{align}
\Big\| B_pe^{ \lambda \cN_+}  \psi \Big\| \leq & \lambda C^m N^{-1} \left( \| \eta \|^{m-\ell_1} \vert \eta_p \vert^{\ell_1} \delta_{\ell_1 >0} \| ( \mathcal{N}_+ +1)^{3/2} \psi \| + \| \eta \|^m \| b_p ( \mathcal{N}_+ +1) e^{ \lambda \cN_+} \psi \| \right) \notag \\
\leq& \lambda C^m N^{-1} \| \eta\|^{m-1}\left(\vert \eta_p \vert \delta_{m>0}  \| ( \mathcal{N}_+ +1)^{3/2}e^{ \lambda \cN_+} \psi \| + \| \eta \| \| b_p ( \mathcal{N}_+ +1) e^{ \lambda \cN_+}\psi \| \right) . \label{eq:term15}
\end{align}
We plug \eqref{eq:term14} and \eqref{eq:term15} into \eqref{eq:term10} and conclude for sufficiently small $\|\eta \|$ at \eqref{eq:eN-dp}. The second bound follows similarly using that in the case $\ell_1 =0$ we only have $\| b_p^* ( \cN_+ + 1) e^{ \lambda \cN_+}\psi \| \leq \| ( \cN_+ + 1)^{3/2} e^{ \lambda \cN_+} \psi \|$. 

The bound on the double commutator follows similarly. We write 
\begin{align}
e^{- \lambda \cN_+} \left[ e^{\lambda\cN_+}, \left[ e^{\lambda \cN_+}, d_p \right] \right] e^{- \lambda \cN_+} =  e^{\lambda \cN_+}  d_p  e^{- \lambda \cN_+} -  e^{-\lambda \cN_+}  d_p  e^{ \lambda \cN_+} \; ,
\end{align}
and thus find 
\begin{align}
& \| e^{- \lambda \cN_+} \left[ e^{\lambda\cN_+}, \left[ e^{\lambda \cN_+}, d_p \right] \right] e^{- \lambda \cN_+}  \psi \| \leq  \sum_{m \geq 0} \frac{1}{m!}  \times \notag \\
& \times \Big\|\left( e^{\lambda \cN_+}\left[ {\rm ad}_{-B(\eta)}^{(m)} (b_p) - \eta_p^m b_{\alpha_m p}^{\sharp_m} \right] e^{-\lambda \cN_+}-  e^{-\lambda \cN_+}\left[ {\rm ad}_{-B(\eta)}^{(m)} (b_p) - \eta_p^m b_{\alpha_m p}^{\sharp_m} \right] e^{\lambda \cN_+}\right)  \psi \Big\| .
\end{align}
By Lemma \ref{lemma:dp} the difference 
\begin{align}
 e^{\lambda \cN_+}\left[ {\rm ad}_{-B(\eta)}^{(m)} (b_p) - \eta_p^m b_{\alpha_m p}^{\sharp_m} \right] e^{-\lambda \cN_+}-  e^{-\lambda \cN_+}\left[ {\rm ad}_{-B(\eta)}^{(m)} (b_p) - \eta_p^m b_{\alpha_m p}^{\sharp_m} \right] e^{\lambda \cN_+}
\end{align}
is the sum of one term of the form 
\begin{align}
A'_p =  e^{\lambda \cN_+} & \left( \frac{N-\cN_+}{N}\right)^{\frac{m+(1-\alpha_m)/2}{2}} \left( \frac{N+1-\cN_+}{N}\right)^{\frac{m+(1+\alpha_m)/2}{2}} \eta_p b_{\alpha_m p}^{\sharp_m}  e^{-\lambda \cN_+} \label{eq:term111}  \\
 &- e^{-\lambda \cN_+}  \left( \frac{N-\cN_+}{N}\right)^{\frac{m+(1-\alpha_m)/2}{2}} \left( \frac{N+1-\cN_+}{N}\right)^{\frac{m+(1+\alpha_m)/2}{2}} \eta_p b_{\alpha_m p}^{\sharp_m}  e^{\lambda \cN_+}  \notag
\end{align}
and $2^m m! - 1$ terms are of the form 
 \begin{align}
B_p =  e^{\kappa \lambda \cN_+}  & \Lambda_1 \dots \Lambda_{i_1} N^{-k} \Pi_{\sharp, \flat}^{(1)} ( \eta^{j_1}, \dots, \eta^{j_{k_1}}; \eta_p^{\ell_1} \varphi_{\alpha_{\ell_1} p} g_p ) e^{- \lambda \kappa\cN_+ } \notag \\
 &-  e^{-\kappa \lambda \cN_+}   \Lambda_1 \dots \Lambda_{i_1} N^{-k} \Pi_{\sharp, \flat}^{(1)} ( \eta^{j_1}, \dots, \eta^{j_{k_1}}; \eta_p^{\ell_1} \varphi_{\alpha_{\ell_1} p} g_p ) e^{ \lambda \kappa\cN_+ } \label{eq:term122}
\end{align} 
where $i_1, k_1, \ell_1 \in \mathbb{N}, j_1, \dots, j_k \in \mathbb{N} \setminus \lbrace 0 \rbrace$ and where each operator $\Lambda_r$ is either a factor $( N- \cN_+)/N$, a factor $(N+1-\cN_+)/N$  or a $\Pi^{(2)}$ operator of the form 
\begin{align}
\label{eq:term133}
N^{-h}\Pi^{(2)}_{\sharp,\flat} ( \eta^{z_1}, \dots, \eta^{z_h} ) 
\end{align}
with $h, z_1, \dots, z_h \in \mathbb{N} \setminus \lbrace 0 \rbrace$. We consider \eqref{eq:term111} and \eqref{eq:term122} separately, thus each term that is of the form \eqref{eq:term111} either has $k_1 >0$ or contains at least one operator of the form \eqref{eq:term133}. We estimate \eqref{eq:term111} first that vanishes for $m=0$. Thus we have 
\begin{align}
 & \|  A_p e^{ \lambda \cN_+}\psi \| = \Big\| \left( \frac{N-\cN_+}{N}\right)^{\frac{m+(1-\alpha_m)/2}{2}} \left( \frac{N+1-\cN_+}{N}\right)^{\frac{m+(1+\alpha_m)/2}{2}}  \times \notag \\
 &\qquad \qquad \qquad  \qquad \qquad \times \eta_p^m \left(   e^{\lambda \cN_+}  b_{\alpha_m p}^{\sharp_m} e^{-\lambda \cN_+} -   e^{-\lambda \cN_+}  b_{\alpha_m p}^{\sharp_m} e^{\lambda \cN_+}\right) \psi \| \notag \\ 
&\leq  \kappa^2 \lambda C^m \vert \eta_p \vert^m N^{-1}\|( \mathcal{N}_+ + 1)^{3/2}e^{ \lambda \cN_+} \psi \| \; . \label{eq:term14}
\end{align}
For \eqref{eq:term12} we find 
\begin{align}
B_p =& \sum_{u=1}^i \left(  \prod_{t=1}^{u-1} e^{\lambda \cN_+}  \Lambda_t e^{-\lambda \cN_+} \right) \left( e^{\lambda \cN_+}  \Lambda_u e^{-\lambda \cN_+} -e^{-\lambda \cN_+}  \Lambda_u e^{\lambda \cN_+}  \right) \prod_{t=u+1}^{i}  e^{\lambda \cN_+}  e^{-\lambda \cN_+} \Lambda_t e^{\lambda \cN_+}  \notag \\
& \quad \times N^{-k} \Pi_{\sharp, \flat}^{(1)} ( \eta^{j_1}, \dots, \eta^{j_{k_1}}; \eta_p^{\ell_1} \varphi_{\alpha_{\ell_1} p} ) \notag \\
+& \left( \prod_{t=1}^{i}   e^{\lambda \cN_+}  \Lambda_t e^{-\lambda \cN_+}  \right)   N^{-k} \notag \\
& \quad \times \left(  e^{\lambda \cN_+}\Pi_{\sharp, \flat}^{(1)} ( \eta^{j_1}, \dots, \eta^{j_{k_1}}; \eta_p^{\ell_1} \varphi_{\alpha_{\ell_1} p} ) e^{-\lambda \cN_+} - \Pi_{\sharp, \flat}^{(1)} ( \eta^{j_1}, \dots, \eta^{j_{k_1}}; \eta_p^{\ell_1} \varphi_{\alpha_{\ell_1} p} )\right) \;.
\end{align}
In case $\Lambda_u$ is of the form $( N- \mathcal{N}_+)/N $ or $( N +1- \mathcal{N}_+)/N$ then $e^{\lambda \cN_+}  \Lambda_u e^{-\lambda \cN_+} - e^{-\lambda \cN_+}  \Lambda_u e^{\lambda \cN_+}  $ vanishes. Otherwise, if $\Lambda_u$ is an operator of the form $\Pi^{(2)}$ it creates resp. annihilates two particles, thus, we have $e^{\lambda \cN_+}  \Lambda_u e^{-\lambda \cN_+} - e^{-\lambda \cN_+}  \Lambda_u e^{\lambda \cN_+}  =  ( e^{\lambda \kappa_u} -e^{-\lambda \kappa_u})\Lambda_u$ with $\kappa_u =2$ or $\kappa_u = -2$. Similarly, as the operator $\Pi^{(1)}$ creates or annihilates one particle, we have 
\begin{align}
&e^{\lambda \cN_+ }\Pi_{\sharp, \flat}^{(1)}  ( \eta^{j_1}, \dots, \eta^{j_1}, \dots, \eta^{j_{k_1}}; \eta_p^{\ell_1} \varphi_{\alpha_{\ell_1} p} ) e^{-\lambda \cN_+} \nonumber\\
&\qquad \qquad \qquad  - e^{-\lambda \cN_+ }\Pi_{\sharp, \flat}^{(1)}  ( \eta^{j_1}, \dots, \eta^{j_1}, \dots, \eta^{j_{k_1}}; \eta_p^{\ell_1} \varphi_{\alpha_{\ell_1} p} ) e^{\lambda \cN_+} \notag \\
=&(e^{\lambda \tilde\kappa} - e^{-\lambda \tilde\kappa}) \Pi_{\sharp, \flat}^{(1)} ( \eta^{j_1}, \dots,   \eta^{j_1}, \dots, \eta^{j_{k_1}}; \eta_p^{\ell_1} \varphi_{\alpha_{\ell_1} p} ) 
\end{align}
with $\tilde\kappa = 1$ or $\tilde\kappa=-1$. Therefore we find 
\begin{align}
\Big\| B_p e^{ \lambda \cN_+} \psi \Big\| \leq& \left(\sum_{u=1}^i (e^{\kappa_u} -e^{\lambda \kappa_u})+ (e^{\kappa}-e^{\lambda \kappa}) \right) \|  \prod_{t=1}^{i}   \Lambda_t N^{-k} \Pi_{\sharp, \flat}^{(1)} ( \eta^{j_1}, \dots, \eta^{j_{k_1}}; \eta_p^{\ell_1} \varphi_{\alpha_{\ell_1} p} ) \psi \| \; .
\end{align}
We consider the case $\ell_1 =0$ and $\ell_1 >0$ separately (see for example \cite[Lemma 3.4]{BCCS_optimal} resp. \cite[Section 5]{BS}) and arrive with $\vert \eta_p \vert \leq \| \eta \|$ at 
\begin{align}
\Big\| B_pe^{ \lambda \cN_+}  \psi \Big\| \leq & \lambda^2 C^m N^{-1} \left( \| \eta \|^{m-\ell_1} \vert \eta_p \vert^{\ell_1} \delta_{\ell_1 >0} \| ( \mathcal{N}_+ +1)^{3/2} \psi \| + \| \eta \|^m \| b_p ( \mathcal{N}_+ +1) e^{ \lambda \cN_+} \psi \| \right) \notag \\
\leq& \lambda^2 C^m N^{-1} \| \eta\|^{m-1}\left(\vert \eta_p \vert \delta_{m>0}  \| ( \mathcal{N}_+ +1)^{3/2}e^{ \lambda \cN_+} \psi \| + \| \eta \| \| b_p ( \mathcal{N}_+ +1) e^{ \lambda \cN_+}\psi \| \right) . \label{eq:term15}
\end{align}
We plug \eqref{eq:term14} and \eqref{eq:term15} into \eqref{eq:term10} and conclude for sufficiently small $\|\eta \|$ at \eqref{eq:eN-dp} for $k=0$. Since $\cN_+$ can be easily commuted through any operators of the form $\Pi^{(1)}, \Pi^{(2)}$ and $\Lambda_i$, the case $k \in \mathbb{Z}$ follows. The second bound follows similarly using that in the case $\ell_1 =0$ we only have $\| b_p^* ( \cN_+ + 1) e^{ \lambda \cN_+}\psi \| \leq \| ( \cN_+ + 1)^{3/2} e^{ \lambda \cN_+} \psi \|$. 

For the remaining estimates \eqref{eq:eN-dx-single}, \eqref{eq:eN-dx-double} and \eqref{eq:eN-db-single}, \eqref{eq:eN-db-double} we observe 
\begin{align}
\Big[ e^{\lambda \cN_+}, \check{d}_x \check{d}_y \Big] &= (  e^{\lambda\cN_+ } \check{d}_x \check{d}_y e^{-\lambda\cN_+ } -\check{d}_x \check{d}_y) e^{\lambda\cN_+ } \notag \\
&=   (e^{\lambda\cN_+ } \check{d}_x e^{-\lambda\cN_+ } -\check{d}_x ) e^{\lambda\cN_+ } \check{d}_y +  \check{d}_x ( e^{\lambda\cN_+ } \check{d}_y e^{-\lambda\cN_+ } - \check{d}_y) e^{\lambda\cN_+ }
\end{align}
resp. 
\begin{align}
e^{-\lambda \cN_+} &  \Big[e^{\lambda \cN_+},  \Big[ e^{\lambda \cN_+}, \check{d}_x \check{d}_y \Big] \notag \\
=& \left( e^{-\lambda \cN_+} \check{d}_x \check{d}_y e^{\lambda \cN_+}-  2 \check{d}_x \check{d}_y +  e^{\lambda \cN_+}\check{d}_x \check{d}_y e^{-\lambda \cN_+} \right) e^{\lambda \cN_+} \notag \\
=&   (e^{-\lambda\cN_+ } \check{d}_x e^{\lambda\cN_+ } -\check{d}_x ) e^{-\lambda\cN_+ } \check{d}_y e^{2\lambda\cN_+ }+  \check{d}_x ( e^{-\lambda\cN_+ } \check{d}_y e^{\lambda\cN_+ } - \check{d}_y) e^{\lambda \cN_+} \notag \\
&-  (e^{\lambda\cN_+ } \check{d}_x e^{-\lambda\cN_+ } -\check{d}_x ) e^{\lambda\cN_+ } \check{d}_y -  \check{d}_x ( e^{\lambda\cN_+ } \check{d}_y e^{-\lambda\cN_+ } - \check{d}_y) e^{\lambda\cN_+ }
\end{align}
and similarly for products of the form $\check{d}_x \check{b}_y$. Then we use the bounds for 
\begin{align}
e^{\lambda\cN_+ } \check{d}_x e^{-\lambda\cN_+ } -\check{d}_x , \quad \text{resp.} \quad e^{- \lambda \cN_+} \check{d}_y e^{ \lambda \cN_+} -  e^{\lambda \cN_+} \check{d}_y e^{ \lambda \cN_+}  \quad \text{and} \quad  e^{-\lambda \cN_+} \check{d}_x  e^{\lambda \cN_+} 
\end{align}
obtained before and then \eqref{eq:eN-dx-single}, \eqref{eq:eN-dx-double} follow by  controlling the commutator of $a_x$ through operators of the form $\Pi^{(1)}, \Pi^{(2)}$ and $\Lambda_i$. Since 
\begin{align}
[\check{a}_x, \int_{\Lambda^2} dydz \check{a}_{y}^* \check{a}_z \eta^{(j)} (y;z)] = \check{a}( \eta^{(j)}_x ) , \quad \text{and} \quad [\check{a}_x, \check{a}^*( \eta_y)] = \eta(x-y)
\end{align}
we then arrive at \eqref{eq:eN-dx-single}, \eqref{eq:eN-dx-double} (see also \cite[Lemma 3.4]{BCCS_optimal}). 

The estimates \eqref{eq:eN-db-single}, \eqref{eq:eN-db-double} follow in the same way using that $e^{\lambda \cN_+} \check{b}_x = \check{b}_x e^{\cN_+ -1}$ from the commutation relations \eqref{eq:comm-b}.
\end{proof}

From the previous Lemma \ref{lemma:eN-dp}, we get estimates on 
\begin{align}
\label{def:wN}
\wN :=& e^{B( \eta)} \mathcal{N}_+ e^{-B( \eta)} 
\end{align}
resp. single and double commutators with $e^{\kappa \cN_+}$. To derive those estimates, we use that 
\begin{align}
\wN =& \cN_+ + \int_0^1 ds \; e^{ s B( \eta)} \sum_{p \in \Lambda_+^*} \eta_p [B( \eta), a_p^*a_p] e^{-s( \eta)} \notag \\
&= \cN_+ + \int_0^1 ds \; e^{ s B( \eta)} \sum_{p \in \Lambda_+^*} \eta_p [b_p^*b_{-p}^* + b_pb_{-p}] e^{-sB( \eta)} 
\end{align}
that we write with \eqref{eq:action-bogo} as 
\begin{align}
\label{eq:wN-1}
\wN =& \cN_+  + \sum_{p \in \Lambda_+^*} \left(  (\gamma_p^2 + \sigma_p^2 -1) b_p^* b_p +  \gamma_p \sigma_p b_p^*b_{-p} + \sigma_p^2 [b_p^*,b_p] \right) \\
&+ \sum_{p \in \Lambda_+^*}  \eta_p \int_0^1 ds \; \left( (\gamma_p^{(s)} b_p + \sigma_p^{(s)} b_{-p}^* ) d_p^{(s)} + {\rm h.c.} \right) + \sum_{p \in \Lambda_+^*} \eta_p \int_0^1 ds \; ( d_p^{(s)} d_p^{(s)} +{\rm h.c.} ) \; 
\end{align}
where we introduced the notation $\gamma_p^{(s)} = \cosh( s \eta_p), \sigma_p^{(s)} = \sinh ( s \eta_p)$ and $d_p^{(s)}$ for the remainder terms defined by \eqref{def:d-2} for the kernel $s \eta_p$.

\begin{lemma}
\label{lemma:bounds-wN}
Let $\wN$ be defined in \eqref{def:wN}. Let $\xi_1, \xi_2 \in \mathcal{F}_{\perp u_0}^{\leq N}$  and $j \in \mathbb{N}_0$. Then, there exists $C>0$ such that 
\begin{align}
\label{eq:bounds-wN}
 \vert \langle \xi_1, \;  \wN  \xi_2 \rangle \vert \leq&  C  \| ( \mathcal{N} +1)^{(1 -j)/2} \xi_1 \| \; \| ( \mathcal{N} +1)^{(j+1)/2} \xi_2 \| \; . 
\end{align}
Furthermore, for $\kappa>0$ we have 
\begin{align}
\label{eq:bound-tildeN-V}
\| e^{\kappa \cN_+} \wN \xi \| \leq&  C e^{ 2 \kappa} \|  \wN e^{\kappa \cN_+} \xi \| 
\end{align}
and 
\begin{align}
\label{eq:single-wN}
\vert \langle \xi_1, \;  \left[ e^{\kappa \mathcal{N}_+}, \; \wN \right] \xi_2 \rangle \vert \leq& C \kappa \| ( \mathcal{N} +1)^{(1 -j)/2} \xi_1 \| \; \| ( \mathcal{N} +1)^{(j+1)/2} e^{\kappa \cN_+} \xi_2 \| \notag \\
\vert \langle \xi_1, \;  \left[ e^{\kappa \mathcal{N}_+}, \; \wN \right] \xi_2 \rangle \vert \leq& C \kappa\| ( \mathcal{N} +1)^{(1 -j)/2}e^{\kappa \cN_+} \xi_1 \| \; \| ( \mathcal{N} +1)^{(j+1)/2}  \xi_2 \| \; 
\end{align}
and 
\begin{align}
\label{eq:double-wN}
\vert \langle \xi_1, \; \left[ e^{\kappa \mathcal{N}_+}, \;  \left[ e^{\kappa \mathcal{N}_+}, \; \wN \right]\right] \xi_2 \rangle \vert \leq C \kappa^2 \| ( \mathcal{N} +1)^{(1 -j)/2} \xi_1 \| \; \| ( \mathcal{N} +1)^{(j+1)/2} \xi_2 \| \; . 
\end{align}
\end{lemma} 

\begin{remark}
\label{remark:wN-norm}
Note that since $\eta \in \ell^2( \Lambda_+^*)$ the estimates  \eqref{eq:bounds-b1}-\eqref{eq:bounds-b3} imply for any $\xi \in \mathcal{F}_{\perp u_0}^{\leq N}$ 
\begin{align}
\| \wN \xi \| \leq C \; \| ( \cN_+ + 1 ) \xi \|
\end{align}
and moreover by Lemma \ref{lemma:double} 
\begin{align}
\| \left[ e^{\kappa \mathcal{N}_+}, \; \wN \right] \xi \| \leq C e^{\kappa} \sinh( \kappa ) \; \| ( \mathcal{N}_+ + 1 ) e^{\kappa \cN_+} \xi \| \; . 
\end{align}
\end{remark}

\begin{proof} From \eqref{eq:wN-1} and \eqref{eq:comm-b} we get 
\begin{align}
\label{eq:wN-11}
\langle & \xi_1, \; \wN \xi_2 \rangle \notag \\
=&  \sum_{p \in \Lambda_+^*}   \left( \gamma_p^2 + \sigma_p^2 \right)  \langle \xi_1 ,  \; b_p^*b_p  \xi_2 \rangle + \sum_{p \in \Lambda_+^*}   \sigma_p \gamma_p  \langle \xi_1, \; \left( b_p^*b^*_{-p} + b_p b_{-p} \right) \xi_2  \rangle  + \| \sigma \|_{\ell^2}^2  \langle\xi_1, \xi_2 \rangle  \notag \\
&+\sum_{p \in \Lambda_+^*}  \eta_p \int_0^1 ds \; \langle \xi_1, \left( (\gamma_p^{(s)} b_p + \sigma_p^{(s)} b_{-p}^* ) d_p^{(s)} + {\rm h.c.} \right) \xi_2 \rangle \notag \\
&+ \sum_{p \in \Lambda_+^*} \eta_p \int_0^1 ds \; \langle \xi_1, ( d_p^{(s)} d_p^{(s)} +{\rm h.c.} ) \xi_2 \rangle \; . 
\end{align}
Inserting $( \mathcal{N}_++1)^{-j}( \mathcal{N}_+ + 1)^j $ with $j \in \mathbb{N}_0$, we furthermore find with the commutation relations \eqref{eq:comm-b} 
\begin{align}
\langle \xi_1, &\; \wN \xi_2 \rangle =  \| \sigma \|_{\ell^2}^2  \langle\xi_1, \xi_2 \rangle +   \sum_{p \in \Lambda_+^*}   \left( \gamma_p^2 + \sigma_p^2 \right)  \langle \xi_1 , ( \mathcal{N}_++1)^{-j} \; b_p^*b_p  ( \mathcal{N}_++1)^{j} \xi_2 \rangle \notag \\
&\quad + \sum_{p \in \Lambda_+^*}   \sigma_p \gamma_p  \langle \xi_1, \; \left( ( \mathcal{N}_++1)^{-j}b_p^*b^*_{-p} ( \mathcal{N}_++3)^{j}+  ( \mathcal{N}_++1)^{-j}b_p b_{-p}  ( \mathcal{N}_+ - 1)^{j}\right) \xi_2 \rangle  \notag \\
&\quad +\sum_{p \in \Lambda_+^*}  \eta_p \int_0^1 ds \; \langle \xi_1, \; \left( (\gamma_p^{(s)} b_p + \sigma_p^{(s)} b_{-p}^* ) (\cN_+ + 1)^{-j+j}d_p^{(s)} + {\rm h.c.} \right) \xi_2 \rangle \notag \\
&\quad+ \sum_{p \in \Lambda_+^*} \eta_p \int_0^1 ds \; \langle \xi_1, ( d_p^{(s)} (\cN_+ + 1)^{-j+j} d_p^{(s)} +{\rm h.c.} )\xi_2 \rangle  \; .
\end{align}
Now we estimate the terms of the r.h.s. With \eqref{eq:bounds-b1}-\eqref{eq:bounds-b3}, \eqref{eq:bounds-sigma} and \eqref{eq:estimates-dp}-\eqref{eq:estimates-dd}, we find 
\begin{align}
\vert \langle \xi_1, \; \wN \xi_2 \rangle \vert \leq \| ( \mathcal{N} +1)^{(1 -j)/2} \xi_1 \| \; \| ( \mathcal{N} +1)^{(j+1)/2} \xi_2 \| \; . 
\end{align}
and moreover with 
\begin{align}
e^{\kappa \cN_+}  & \wN e^{-\kappa \cN_+} \notag \\
=&  \sum_{p \in \Lambda_+^*} \left[ \left( \gamma_p^2 + \sigma_p^2 \right) b_p^*b_p + e^{2\kappa} \sigma_p \gamma_p \left( b_p^*b_{-p}^* +e^{-2 \kappa} b_p b_{-p} \right) \right] + \| \sigma \|_{\ell^2}^2 \notag \\
&+\sum_{p \in \Lambda_+^*}  \eta_p \int_0^1 ds \;  e^{\kappa \cN_+}\left( (\gamma_p^{(s)} b_p + \sigma_p^{(s)} b_{-p}^* )  d_p^{(s)} + {\rm h.c.} \right)e^{-\kappa \cN_+} \notag \\
&+ \sum_{p \in \Lambda_+^*} \eta_p \int_0^1 ds \;  e^{-\kappa \cN_+}  ( d_p^{(s)} d_p^{(s)} +{\rm h.c.} ) e^{-\kappa \cN_+}  \; 
\end{align}
and Lemma \ref{lemma:eN-dp} the second bound from \eqref{eq:bounds-wN}.

For the remaining estimates \eqref{eq:single-wN}, \eqref{eq:double-wN} we first observe with Lemma \ref{lemma:double} that 
\begin{align}
\langle \xi_1, & \; \left[ e^{\kappa \mathcal{N}_+}, \; \wN \right]  \xi_2 \rangle \notag \\
=&  2\sinh ( \kappa ) \sum_{p \in \Lambda_+^*}   \sigma_p \gamma_p  \langle \xi_1, \; \left( e^{\kappa}b_p^*b^*_{-p} + e^{- \kappa}b_p b_{-p} \right) e^{\kappa \cN_+} \xi_2 \rangle\notag \\
&+\sum_{p \in \Lambda_+^*}  \eta_p \int_0^1 ds \;  \langle \xi_1, [ e^{\kappa \cN_+}, \left( (\gamma_p^{(s)} b_p + \sigma_p^{(s)} b_{-p}^* )  d_p^{(s)} + {\rm h.c.} \right]  \xi_2 \rangle \notag \\
&+ \sum_{p \in \Lambda_+^*} \eta_p \int_0^1 ds \; \langle \xi_1, [ e^{\kappa \cN_+}, ( d_p^{(s)} d_p^{(s)} +{\rm h.c.} ) ] \xi_2 \rangle \; 
\end{align}
and for the last two lines 
\begin{align}
\left[ e^{\kappa \cN_+}, b_p^{\sharp_{1}} d_{\alpha p }^{\sharp_2} \right] =& \left[ e^{\kappa \cN_+}, b_p^{\sharp_{1}} \right]d_{\alpha p }^{\sharp_2} + b_p^{\sharp_{1}} \left[ e^{\kappa \cN_+},  d_{\alpha p }^{\sharp_2} \right]\notag \\
=& (2 \sinh (\kappa/2) e^{ \beta \kappa/2} + 1)  b_p^{\sharp_{1}} \left[ e^{\kappa \cN_+},  d_{\alpha p }^{\sharp_2} \right] + 2 \sinh (\kappa/2) e^{ \beta \kappa/2}   b_p^{\sharp_{1}} d_{\alpha p }^{\sharp_2}   e^{\kappa \cN_+} 
\end{align}
with $\sharp_1, \sharp_2 \in \lbrace \cdot, *\rbrace$ and either $\sharp_1 = *, \sharp_2 = \cdot $ and $\alpha = 1, \beta =1$ or $\sharp_1 = \sharp_2$ and $\alpha = -1$ and $\beta =1$ if $\sharp_1 = *$ and $\beta =-1$ otherwise. Similarly
\begin{align}
\langle \xi_1,  &\; \left[ e^{\kappa \mathcal{N}_+}, \; \wN \right]  \xi_2 \rangle \notag \\
=&  2 \sinh ( \kappa ) \sum_{p \in \Lambda_+^*}   \sigma_p \gamma_p  \langle \xi_1, e^{\kappa \cN_+} \; \left( e^{-\kappa}b_p^*b^*_{-p} + e^{ \kappa}b_p b_{-p} \right)  \xi_2 \rangle\notag \\
&+\sum_{p \in \Lambda_+^*}  \eta_p \int_0^1 ds \;  \langle \xi_1, \left[ e^{\kappa \cN_+}, \big( (\gamma_p^{(s)} b_p + \sigma_p^{(s)} b_{-p}^* )  d_p^{(s)} + {\rm h.c.}\big)  \right]  \xi_2 \rangle \notag \\
&+ \sum_{p \in \Lambda_+^*} \eta_p \int_0^1 ds \; \langle \xi_1, [ e^{\kappa \cN_+}, ( d_p^{(s)} d_p^{(s)} +{\rm h.c.} ) ] \xi_2 \rangle \; 
\end{align}
and for the last line 
\begin{align}
\left[ e^{\kappa \cN_+}, b_p^{\sharp_{1}} d_{\alpha p }^{\sharp_2} \right] =& \left[ e^{\kappa \cN_+}, b_p^{\sharp_{1}} \right]d_{\alpha p }^{\sharp_2} + b_p^{\sharp_{1}} \left[ e^{\kappa \cN_+},  d_{\alpha p }^{\sharp_2} \right]\notag \\
=& 2 \sinh (\kappa/2) e^{ -\beta \kappa/2} e^{\kappa \cN_+} b_p^{\sharp_{1}} d_{\alpha p }^{\sharp_2}  +  e^{ -\beta \kappa} e^{\kappa \cN_+}b_p^{\sharp_{1}} e^{-\kappa \cN_+} \left[ e^{\kappa \cN_+} , d_{\alpha p }^{\sharp_2}   \right]
\end{align}
with $\sharp_1, \sharp_2 \in \lbrace \cdot, *\rbrace$ and either $\sharp_1 = *, \sharp_2 = \cdot $ and $\alpha = 1, \beta =1$ or $\sharp_1 = \sharp_2$ and $\alpha = -1$ and $\beta =1$ if $\sharp_1 = *$ and $\beta =-1$ otherwise.
Moreover, 
\begin{align}
\langle \xi_1,  &\; \left[ e^{\kappa \mathcal{N}_+}, \; \left[ e^{\kappa \mathcal{N}_+}, \; \wN \right]\right]  \xi_2 \rangle\notag \\
 =& 4 \sinh^2 ( \kappa ) \sum_{p \in \Lambda_+^*}   \sigma_p \gamma_p  \langle \xi_1, \;e^{\kappa \cN_+} \left( b_p^*b^*_{-p} + b_p b_{-p} \right) e^{\kappa \cN_+} \xi_2 \rangle \notag \\
&+\sum_{p \in \Lambda_+^*}  \eta_p \int_0^1 ds \;  \langle \xi_1, [ e^{\kappa \cN_+}, \left[ e^{\kappa \cN_+}, \big( (\gamma_p^{(s)} b_p + \sigma_p^{(s)} b_{-p}^* )  d_p^{(s)} + {\rm h.c.} \big) ]\right]  \xi_2 \rangle \notag \\
&+ \sum_{p \in \Lambda_+^*} \eta_p \int_0^1 ds \; \langle \xi_1, [ e^{\kappa \cN_+}, [e^{\kappa \cN_+},  ( d_p^{(s)} d_p^{(s)} +{\rm h.c.} ) ]] \xi_2 \rangle \; 
\end{align}
and for the last line 
\begin{align}
\label{eq:bd-dc}
 & \left[ e^{\kappa \cN_+}, \left[ e^{\kappa \cN_+}, b_p^{\sharp_{1}} d_{\alpha p }^{\sharp_2} \right] \right] \notag \\
&= \left[ e^{\kappa \cN_+}, \left[ e^{\kappa \cN_+}, b_p^{\sharp_{1}} \right] \right] d_{\alpha p }^{\sharp_2} + b_p^{\sharp_{1}}\left[ e^{\kappa \cN_+}, \left[ e^{\kappa \cN_+},  d_{\alpha p }^{\sharp_2} \right] \right]  + 2 \left[ e^{\kappa \cN_+}, b_p^{\sharp_{1}} \right] \left[ e^{\kappa \cN_+},  d_{\alpha p }^{\sharp_2} \right] \notag \\
&= \left[ e^{\kappa \cN_+}, \left[ e^{\kappa \cN_+}, b_p^{\sharp_{1}} \right] \right] e^{-\kappa \cN_+} \left( e^{\kappa \cN_+} d_{\alpha p }^{\sharp_2} e^{-\kappa \cN_+} \right)  + e^{\beta \kappa}e^{\kappa \cN_+} b_p^{\sharp_{1}} e^{- \kappa \cN_+} \left[ e^{\kappa \cN_+}, \left[ e^{\kappa \cN_+},  d_{\alpha p }^{\sharp_2} \right] \right]  \notag \\
&\quad + 2 \left[ e^{\kappa \cN_+}, b_p^{\sharp_{1}} \right] \left[ e^{\kappa \cN_+},  d_{\alpha p }^{\sharp_2} \right] 
\end{align}
with $\sharp_1, \sharp_2 \in \lbrace \cdot, *\rbrace$ and either $\sharp_1 = *, \sharp_2 = \cdot $ and $\alpha = 1, \beta =1$ or $\sharp_1 = \sharp_2$ and $\alpha = -1$ and $\beta =1$ if $\sharp_1 = *$ and $\beta =-1$ otherwise.
Thus with similar ideas as before, we conclude by \eqref{eq:bounds-b1}-\eqref{eq:bounds-b3}, \eqref{eq:bounds-sigma} and Lemma \ref{lemma:eN-dp} with \eqref{eq:single-wN} resp. \eqref{eq:double-wN}. 
\end{proof}

\section{Proof of Proposition \ref{prop:G}}
\label{sec:proof-prop-G}

In this section we will analyze properties of the single contributions $\mathcal{G}_N^{(j)}$ of the excitation Hamiltonian $\cG_N$ in \eqref{def:G-split}, and then conclude Proposition \ref{prop:G} at the end. 

%

For our analysis it will be useful to use the expression of $\mathcal{V}_N$ in \eqref{def:HN} in position space, 
\begin{align}
\label{eq:VN-position}
\mathcal{V}_N  = \frac{1}{N} \int_{\Lambda \times \Lambda} dxdy \;  v_N (x-y) a_x^*a_y^*a_xa_y, \quad v_N (x) = N^{3} v (N x). 
\end{align}

\subsection{Analysis of $\mathcal{G}_N^{(0)}$} With \eqref{eq:action-bogo} we obtain 
\begin{align}
\label{def:G0}
\mathcal{G}_N^{(0)} = \mathcal{C}_{\mathcal{G}_N^{(0)}} + \mathcal{G}_N^{(0,1)}
\end{align}
where $\mathcal{G}_N^{(0)}$ is defined by \eqref{def:G-split} and $ \mathcal{C}_{\mathcal{G}_N^{(0)}}$ is a constant term given by 
\begin{align}
\mathcal{C}_{\mathcal{G}_N^{(0)}} = \frac{( N-1 )}{2}\widehat{v} (0) 
\end{align}
and the remaining terms reads with \eqref{def:wN} 
\begin{align}
\label{def:G0} 
\mathcal{G}_N^{(0,1)} =& - \frac{( N-1 )}{ 2N} \wN  + \frac{\widehat{v}(0)}{2N}  \wN ( N- \wN ) \; . 
 \end{align}


\begin{lemma}
\label{lemma_G0}
Let $\mathcal{G}_N^{(0)}$ be given by \eqref{def:G0}. Then there exists $C>0$ independent of $N$ such that 
\begin{align}
\label{eq:bound-G0}
\mathcal{G}_N^{(0)} - \mathcal{C}_{\mathcal{G}_N^{(0)}} \geq -  C (\mathcal{N}_+ + 1) \; 
\end{align}
as operator inequality on $\mathcal{F}_{\perp u_0}^{\leq N}$.  Furthermore let $\kappa>0$ be sufficiently small, then there exists $C>0$ such that for any $\psi \in \mathcal{F}_{\perp u_0}^{\leq N}$ we have 
\begin{align}
\label{eq:bound-G0-dc}
\vert\langle \psi, \;  \left[ e^{\kappa \mathcal{N}_+}, \;  \left[ e^{\kappa \mathcal{N}_+}, \; \mathcal{G}_N^{(0)} \right]\right] \psi \rangle \vert \leq C\kappa^2 \langle \psi, \; (\mathcal{N} +1) \psi \rangle \; . 
\end{align} 
\end{lemma}


\begin{proof}
The first estimate \eqref{eq:bound-G0} immediately follows from the observation $\mathcal{N}_+\leq N$ on $\mathcal{F}_{\perp u_0}^{\leq N}$ and Lemma \ref{lemma:bounds-wN}. For the second bound \eqref{eq:bound-G0-dc}, we find from the properties of the commutator and by definition \eqref{def:G0} of $\mathcal{G}_N^{(0)}$ that 
\begin{align}
 & \left[ e^{\kappa \mathcal{N}_+}, \; \left[ e^{\kappa \mathcal{N}_+}, \; \mathcal{G}_N^{(0)} \right]\right]  \notag \\
  & \quad = \frac{(N-1)}{2N} \widehat{v}(0) \left[ e^{\kappa \mathcal{N}_+}, \;  \left[ e^{\kappa \mathcal{N}_+}, \; \wN \right]\right] + \frac{ \widehat{v}(0)}{2} \left[ e^{\kappa \mathcal{N}_+}, \; \left[ e^{\kappa \mathcal{N}_+}, \; \wN \right]\right]   \notag \\
 &\quad \quad + \frac{ \widehat{v}(0)}{2N} \left[ e^{\kappa \mathcal{N}_+}, \; \left[ e^{\kappa \mathcal{N}_+}, \; \wN \right]\right]  \wN + \frac{ \widehat{v}(0)}{2N} \wN \left[ e^{\kappa \mathcal{N}_+}, \; \left[ e^{\kappa \mathcal{N}_+}, \; \wN \right]\right]  \notag \\
 &+ \frac{ \widehat{v}(0)}{N} \left[ e^{\kappa \mathcal{N}_+}, \;  \wN \right] \left[ e^{\kappa \mathcal{N}_+}, \;  \wN \right]   \; . 
\end{align}
Lemma \ref{lemma:bounds-wN} shows that for any $\xi \in \mathcal{F}_{\perp u_0}^{\leq N}$
\begin{align}
\vert \langle &  \xi, \;  \left[ e^{\kappa \mathcal{N}_+}, \; \left[ e^{\kappa \mathcal{N}_+}, \; \mathcal{G}_N^{(0)} \right]\right]  \xi \rangle \vert \notag \\
\leq&  C \kappa^2 \| ( \mathcal{N}_+ + 1)^{1/2} \xi \|^2  + \frac{C}{N} \kappa^2 \| (\mathcal{N}_+ +1)^{3/2} e^{\kappa \cN_+} \xi \| \| (\mathcal{N}_+ +1)^{1/2} e^{\kappa \cN_+} \wN \xi \| \notag \\
&+ \frac{C}{N} \kappa^2 \| (\mathcal{N}_+ +1) e^{\kappa \cN_+} \xi \|^2 \; . 
\end{align}
\end{proof}

\subsection{Analysis of $\mathcal{G}_N^{(2)}$.} Recalling the definition \eqref{def:Lj} of $\mathcal{L}_N^{(2)}$ we compute in this section 
\begin{align}
\mathcal{G}_N^{(2)} =& e^{B ( \eta)} \mathcal{L}_N^{(2)} e^{- B( \eta)} = e^{B ( \eta)} \sum_{p \in \Lambda_+^*} p^2 a_p^* a_p e^{- B ( \eta)} \notag \\
&+ e^{B ( \eta)} \sum_{p \in \Lambda_+^*} \widehat{v} (p/N ) \left[ b_p^*b_p - \frac{1}{N} a_p^*a_p \right] e^{-B ( \eta)} \notag \\
&+ e^{B ( \eta )} \frac{1}{2} \sum_{p \in \Lambda_+^*} \widehat{v}(p/N ) \left[ b_p^*b_{-p}^* + b_p b_{-p} \right] e^{- B ( \eta)} \; . 
\end{align}
For the last line, we use the generalized Bogoliubov transform's approximate action on modified creation and annihilation operators \eqref{eq:action-bogo}, while for the terms of the first and second line formulated w.r.t. to standard creation and annihilation operators we use arguments similar as in the proof of Lemma \ref{lemma:bounds-wN} to arrive at 
\begin{align}
\label{def:G2}
\mathcal{G}_N^{(2)} - \frac{1 }{2N} \sum_{p,q \in \Lambda_+^*} \widehat{v} ( (p-q)/N)  \eta_q   \left[ b_p^*b_{-p}^* + b_p b_{-p} \right]  = \mathcal{C}_{\mathcal{G}_N^{(2)}} + \widetilde{\mathcal{G}}_N^{(2)}
\end{align}
where $\mathcal{C}_{\mathcal{G}_N^{(2)}}$ is a constant term given by 
\begin{align}
\mathcal{C}_{\mathcal{G}_N^{(2)}} := \sum_{p \in \Lambda_+^*} \left[ \left(p^2 +  \widehat{v} (  p/N ) \right) \sigma_p^2 +  \widehat{v}(p/N ) \sigma_p \gamma_p \right] 
\end{align}
and the remaining term is given by the sum $\widetilde{\mathcal{G}}_N^{(2)} =\sum_{j=1}^4 \mathcal{G}_N^{(2,j)} $ of 
\begin{align}
\mathcal{G}_N^{(2,1)} 
=& \sum_{p \in \Lambda+^* } F_p b_p^*b_p + \frac{1}{2}\sum_{p \in \Lambda_+^*} G_p \left[ b_p^*b_{-p}^* + b_p b_{-p} \right] \notag \\
&+ \frac{1 }{2N} \sum_{p,q \in \Lambda_+^*} \widehat{v} ( (p-q)/N)  \eta_q  \left[ \gamma_p^2 -1 + \sigma_p^2 \right] \left[ b_p^*b_{-p}^* + b_p b_{-p} \right] \notag \\
\mathcal{G}_N^{(2,2)} 
=&  \sum_{p \in \Lambda_+^*}\widehat{v} (  p/N ) \left[ ( \gamma_p b_p^* + \sigma_p b_{-p}) d_p + {\rm h.c.} \right] + \sum_{p \in \Lambda_+^*} \widehat{v}( p/N) d_p^* d_p   \notag \\
\mathcal{G}_N^{(2,3)} =&  \sum_{p \in \Lambda_+^*} \widehat{v} ( p/N ) \left[ ( \gamma_p b_p^* + \sigma_p b_{-p}) d_{-p}^* + d_{p}^*( \gamma_p b_{-p}^* + \sigma_p b_{p}) + d_p^*d_{-p}^* \right]+ {\rm h.c.} \notag \\
\mathcal{G}_N^{(2,4)} =& \frac{1}{N} \sum_{p \in \Lambda_+^*} \widehat{v} ( p/N) \eta_p \int_0^1 ds \; \left[   (\gamma_p^{(s)} b_p^* + \sigma_{p}^{(s)} b_{-p} ) d_p^{(s)}+ {\rm h.c.} \right]  \notag \\
&+  \frac{1}{N}\sum_{p \in \Lambda_+^*} \widehat{v} ( p/N) \eta_p \int_0^1 ds (d_p^{(s)} )^* d_p^{(s)} +\frac{1}{N} \sum_{p \in \Lambda_+^*} \widehat{v}(p/N) (b_p^*b_{p} - a_p^*a_p) \; \notag \\
\mathcal{G}_N^{(2,5)} =& \sum_{p \in \Lambda_+^*} p^2 \eta_p \int_0^1 ds \; \left[   (\gamma_p^{(s)} b_p^* + \sigma_{p}^{(s)} b_{-p} ) d_p^{(s)}+ {\rm h.c.} \right]  +   \sum_{p \in \Lambda_+^*} p^2 \eta_p \int_0^1 ds (d_p^{(s)} )^* d_p^{(s)}
\end{align}
where we introduced the notation
\begin{align}
F_p =&  \left[  p^2 + \widehat{v} (p/N) \right]  \left[ \gamma_p^2 + \sigma_p^2 \right] + 2 \gamma_p \sigma_p \widehat{v} (p/N), \notag \\
G_p =&   \left[ \gamma_p^2 + \sigma_p^2 \right] \Big( \widehat{v} ( p/N ) -   \frac{1}{2N} \sum_{q \in \Lambda_+^*} \widehat{v} ( (p-q)/N)  \eta_q \Big)  + 2 \gamma_p \sigma_p \left[ p^2 + \widehat{v} ( p/N ) \right] 
\end{align}
and $\sigma_p^{(s)} = \sinh (s \eta_p), \gamma_p^{(s)} = \cosh( s \eta_p)$, and the operator $d_p^{(s)}$ is defined by \eqref{def:d-2} where $\eta_p$ is replaced by $s \eta_p$. 

\begin{lemma}
\label{lemma:G2}
Let $\widetilde{\mathcal{G}}_N^{(2)}$ be given by \eqref{def:G2}. Then there exists $\varepsilon, C_\varepsilon >0$ independent of $N$ such that 
\begin{align}
\label{eq:G2-lb}
\widetilde{\mathcal{G}}_N^{(2)}  \geq \frac{1}{2} \mathcal{K} -  C_\varepsilon  (\mathcal{N}_+ +1 ) - \varepsilon \mathcal{V}_N 
\end{align}
as operator inequality on $\mathcal{F}_{\perp u_0}^{\leq N}$.  Furthermore let $\kappa>0$ be sufficiently small, then there exists $C>0$ such that for any $\psi \in \mathcal{F}_{\perp u_0}^{\leq N}$ we have 
\begin{align}
\label{eq:G2-dc}
\vert \langle \psi, \;  \left[ e^{\kappa \mathcal{N}_+}, \;  \left[ e^{\kappa \mathcal{N}_+}, \; \widetilde{\mathcal{G}}_N^{(2)} \right]\right] \psi \rangle \vert  \leq C \kappa^2 \langle \psi, \; \left[ (\mathcal{N} +1) + \mathcal{V}_N \right] \psi \rangle \; . 
\end{align}
\end{lemma}


\begin{proof} 
To prove \eqref{eq:G2-lb} we consider every single contribution of $\mathcal{G}_N^{(2)}$ separately and start with $\mathcal{G}_N^{(2,1)}$. Note that $G_p$ is bounded in $\ell^2 ( \Lambda_+^*)$ uniformly in $N$ as with the splitting $\sigma_p = \eta_p + \beta_p, \gamma_p = 1+ \alpha_p$ and we have 
\begin{align}
G_p =& 2 ( p^2  +  \widehat{v} (p/N )) \eta_p +  \widehat{v} ( p/N ) -   \frac{1}{2N} \sum_{q \in \Lambda_+^*} \widehat{v} ((p-q)/N)  \eta_q \notag \\
&+ 2 \left[ \sigma_p \alpha_p + \beta_p \right]  ( p^2  +  \widehat{v} ( p/N )) \eta_p \notag \\
&+  \left[ \gamma_p \alpha_p + \alpha_p  + \sigma_p^2 \right] \Big( \widehat{v} ( p/N ) -   \frac{1}{2N} \sum_{q \in \Lambda_+^*} \widehat{v} ( (p-q)/N)  \eta_q \Big) \; . 
\end{align}
For the first line of the r.h.s. of the formula above we use the identity \eqref{eq:id-eta} for the operator kernel $\eta_p$. In fact it follows from \cite[Lemma 5.1]{BCCS} that 
\begin{align}
\label{eq:Gp}
\vert G_p \vert \leq C  p^{-2}, \quad \text{and} \quad p^2/2 \leq F_p \leq C (1 + p^2)
\end{align}
for some positive constants $C>0$, in particular yielding $\| G_p \|_{\ell^2} \leq C $. Moreover $ \gamma_p^2 -1, \sigma_p \in \ell^2( \Lambda_+^*)$ and 
\begin{align}
\frac{1 }{2N} \sum_{q \in \Lambda_+^*} \widehat{v} ( (p-q)/N)\eta_q  \leq C
\end{align}
and thus with \eqref{eq:bounds-b1}-\eqref{eq:bounds-b3} 
\begin{align}
\label{eq:G21}
\mathcal{G}_N^{(2,1)} \geq \frac{1}{2} \mathcal{K} - C  ( \mathcal{N}_+ + 1 ) \; . 
\end{align}
For the second term $\mathcal{G}_{N}^{(2,2)}$ we use that from \eqref{eq:bounds-sigma} we have 
\begin{align}
\widehat{v}( p/N) ( \gamma_p^2 + \sigma_p^2) \in \ell^\infty, \quad \widehat{v}( p/N) \gamma_p \sigma_p \in \ell^2
\end{align}
with norms independent of $N$. Thus with the bounds \eqref{eq:bounds-b1}-\eqref{eq:bounds-b3} and \eqref{eq:estimates-dp} we obtain 
\begin{align}
\label{eq:G22}
\vert \langle \xi, \mathcal{G}_N^{(2,2)} \xi \rangle \vert \leq C  \| ( \mathcal{N} + 1 )^{1/2} \xi \|^2  \; . 
\end{align}
The third term $\mathcal{G}_N^{(2,3)}$ we split 
\begin{align}
\mathcal{G}_N^{(2,3)} =& \sum_{p \in \Lambda_+^*} \widehat{v} ( p/N ) \Big[ \sigma_p b_{-p} d_{-p}^* +  \sigma_p d_{p}^* b_{p}) \Big] \notag \\
&+ \sum_{p \in \Lambda_+^*}  \Big[ \gamma_p b_{p}^* d_{-p}^* +  \gamma_p d_{p}^* b_{-p}^*  + d_p^*d_{-p}^* \Big] \notag \\
&+ {\rm h.c.} \\
=& \mathcal{G}_N^{(2,3,1)} + \mathcal{G}_N^{(2,3,2)} + {\rm h.c.}\label{eq:G23}
\end{align}
and find for the first term that since $\sigma_p \in \ell^2( \Lambda_+^*)$ (with norm uniform in $N$) that 
\begin{align}
\label{eq:GN231}
\vert \langle \xi, \; \mathcal{G}_N^{(2,3,1)} \xi \rangle \vert \leq C \| ( \mathcal{N}_+ + 1)^{1/2} \xi \|^2 \; . 
\end{align}
The second contribution of \eqref{eq:G23} we estimate more carefully in terms of $\mathcal{V}_N$. For this we can write $\mathcal{V}_N$ in in position space as \eqref{eq:VN-position}, and similarly 
\begin{align}
\label{eq:G232}
\mathcal{G}_N^{(2,3,2)} = \int_{\Lambda \times \Lambda} dxdy \; v_N (x-y) \big[ \check{b}^*( \check{\gamma}_{x})\check{d}^*_y + \check{d}^*_x \check{b}^*( \check{\gamma}_{y}) + \check{d}_x \check{d}_y\big]  + {\rm h.c.}
\end{align}
where we introduced the point-wise modified creation and annihilation operators $\check{b}_x, \check{b}_y$ for $x,y \in \Lambda$. With these notations we find 
\begin{align}
\vert \langle \xi, \mathcal{G}_N^{(2,3,2)} \xi \rangle \vert 
&\leq  \int_{\Lambda \times \Lambda}dxdy \;  v_N (x-y)  \| ( \mathcal{N}_+ + 1)^{1/2} \xi \| \; \notag \\
& \hspace{1cm} \times \Big( \|( \mathcal{N}_+ + 1)^{-1/2} \check{b}( \check{\gamma}_x ) \check{d}_y \xi \|  +  \|  ( \mathcal{N}_+ + 1)^{-1/2}\check{d}_x\check{b}( \check{\gamma}_y ) \xi \|  \notag \\
& \hspace{2cm} + \| ( \mathcal{N}_+ + 1)^{-1/2}\check{d}_x\check{d}_y \xi \|  \Big) . 
\end{align}
From \eqref{eq:estimates-dp}-\eqref{eq:estimates-db} and \eqref{eq:sigma-x-bounds} we get 
\begin{align}
\label{eq:GN232}
\vert \langle \xi, \mathcal{G}_N^{(2,3,2)} \xi \rangle \vert 
&\leq  \frac{C}{\sqrt{N}}\int_{\Lambda \times \Lambda}dxdy \;  v_N (x-y)  \| ( \mathcal{N}_+ + 1)^{1/2} \xi \| \; \notag \\
& \hspace{1cm} \times \Big( \|  \check{a}_x \xi \|  +  \|  \check{a}_y \xi \|  + \| \check{a}_x \check{a}_y \xi \|  + \| ( \mathcal{N}_+ + 1)^{1/2} \xi \| \Big)  \notag \\
&\leq C \| v_N \|_{L^1( \Lambda) } (\| \mathcal{V}_N^{1/2} \xi \| +  \| (\mathcal{N} _+ + 1)^{1/2} \xi \| )  \| (\mathcal{N} _+ + 1)^{1/2} \xi \| \notag \\
&\leq \varepsilon \langle \xi, \mathcal{V}_N \xi \rangle  + C_\varepsilon \langle \xi, ( \mathcal{N} + 1) \xi \rangle \; 
\end{align}
for some $C_\varepsilon, \varepsilon >0$. Summarizing \eqref{eq:GN231}, \eqref{eq:GN232} we get 
\begin{align}
\label{eq:G23}
 \vert \langle \xi, \mathcal{G}_N^{(2,3)} \xi \rangle \vert  \leq \varepsilon \langle \xi, \mathcal{V}_N \xi \rangle + C_\varepsilon \langle \xi, ( \mathcal{N}_+ + 1) \xi \rangle \; .
\end{align}
In order to estimate the forth term of \eqref{def:G2} we proceed similarly as for the second term $\mathcal{G}_N^{(2,2)}$. We estimate 
\begin{align}
\vert \langle \xi, \mathcal{G}_N^{(2,4)} \xi \rangle \vert \leq \sum_{p \in \Lambda_+^*} \widehat{v} (p/N) \eta_p \int_0^1 ds \; \| d_p^{(s)} \xi \| \Big( \| d_p^{(s)} \xi \|  + \vert \gamma_p^{(s)} \vert \|  b_p \xi \| + \vert \sigma_p^{(s)} \vert b_p^* \xi \|\Big) + C \| \xi \|^2
\end{align}
and thus find with \eqref{eq:estimates-dp}, \eqref{eq:bounds-sigma}
\begin{align}
\label{eq:G24}
\vert \langle \xi, \mathcal{G}_N^{(2,4)} \xi \rangle \vert \leq C \| ( \mathcal{N}_+ + 1)^{1/2} \xi \|^2 \; .
\end{align}
For the fifth term we find with similar arguments as $p^2 \eta_p \in \ell^\infty ( \Lambda_+^*)$ from Lemma \ref{lemma:scattering} that 
\begin{align}
\label{eq:G25}
\vert \langle \xi, \mathcal{G}_N^{(2,5)} \xi \rangle \vert \leq C \| ( \mathcal{N}_+ + 1)^{1/2} \xi \|^2 \; .
\end{align}
Summarizing \eqref{eq:G21},\eqref{eq:G22},\eqref{eq:G23},\eqref{eq:G24} and \eqref{eq:G25} we arrive at the first bound \eqref{eq:G2-lb}. 

Next we prove \eqref{eq:G2-dc}. For this we estimate the four terms of $\mathcal{G}_N^{(j)}$ separately. With Lemma \ref{lemma:double} we observe that 
\begin{align}
\label{eq:bound-G0-1}
\left[ e^{\kappa \mathcal{N}_+}, \; \left[  e^{\kappa \cN_+}, \; \mathcal{G}_N^{(2,1)} \right] \right] = 2 \sinh^2( \kappa) \;  e^{\kappa \cN_+} \sum_{p \in \Lambda_+^*} G_p  \left[ b_p^*b_{-p}^* + b_p b_{-p} \right] e^{\kappa \cN_+} 
\end{align}
We recall from \eqref{eq:Gp} that $\| G_p \|_{\ell^2} \leq C$ and thus we arrive with \eqref{eq:bounds-b1}-\eqref{eq:bounds-b3} for sufficiently small $\kappa>0$ at 
\begin{align}
\vert \langle \xi, \; \left[ e^{\kappa \mathcal{N}_+}, \; \left[  e^{\kappa \cN_+}, \; \mathcal{G}_N^{(2,1)} \right] \right] \xi \rangle \vert \leq C \kappa^2 \| ( \mathcal{N} + 1)^{1/2} e^{\kappa \cN_+} \xi \|^2 . \label{eq:G21-dc}
\end{align}
For the second term we write 
\begin{align}
& \Big[   e^{\kappa \mathcal{N}_+},  \; \left[  e^{\kappa \cN_+}, \; \mathcal{G}_N^{(2,2)} \right] \Big] \notag \\
=&  \sinh ( \kappa/2)^2 \sum_{p \in \Lambda_+^*}  \widehat{v}(p/N) e^{\kappa \cN_+} ( \gamma_p e^{\kappa} b_p^* + \sigma_p e^{- \kappa} b_{-p}) ) d_p \notag \\
&+  \sinh ( \kappa/2) \sum_{p \in \Lambda_+^*}  \widehat{v}(p/N) e^{\kappa \cN_+} ( \gamma_p e^{\kappa} b_p^* + \sigma_p e^{- \kappa} b_{-p}) ) \left[ e^{\kappa \cN_+}, d_p \right] \notag \\
&+   \sum_{p \in \Lambda_+^*}  \widehat{v}(p/N) e^{\kappa \cN_+} ( \gamma_p e^{\kappa} b_p^* + \sigma_p e^{- \kappa} b_{-p}) ) e^{- \lambda \cN_+} \left[ e^{\kappa \cN_+}, [ e^{\kappa \cN_+}, d_p] \right] \notag \\
&+   \sum_{p \in \Lambda_+^*}  \widehat{v}(p/N) e^{\kappa \cN_+} \Big( d_p^* \left[ e^{\kappa \cN_+}, [ e^{\kappa \cN_+}, d_p] \right] + \left[ e^{\kappa \cN_+}, [ e^{\kappa \cN_+}, d_p^*]\right] d_p  + 2 [ e^{\kappa \cN_+}, d_p^*] [ e^{\kappa \cN_+}, d_p] \Big) \; . 
\end{align}
Now we can estimate all contributions similarly to \eqref{eq:G22} using instead of the bounds for $d_p,d_p^*$ in \eqref{eq:estimates-dp}, \eqref{eq:estimates-dp*} the estimates of Lemma \ref{lemma:eN-dp}. In fact notice that the bounds \eqref{eq:eN-dp} in Lemma \ref{lemma:eN-dp} differ from \eqref{eq:estimates-dp} only by a factor of $\kappa$ for the single and $\kappa^2$ for the double commutator. Thus we get 
\begin{align}
\vert \langle \xi, \Big[   e^{\kappa \mathcal{N}_+}, & \; \left[  e^{\kappa \cN_+}, \; \mathcal{G}_N^{(2,2)} \right] \Big] \xi \rangle \vert \leq C \kappa^2 \| ( \mathcal{N}_+ +1 ) e^{\kappa \cN_+} \xi \|^2 \; . 
\end{align}
For the third term $\mathcal{G}_N^{(2,3)}$ we use the same splitting as before (see \eqref{eq:G23}) and find using again \eqref{eq:eN-dp} of Lemma \ref{lemma:eN-dp} instead of \eqref{eq:estimates-dp} that 
\begin{align}
\vert \langle \xi, \Big[   e^{\kappa \mathcal{N}_+}, & \; \left[  e^{\kappa \cN_+}, \; \mathcal{G}_N^{(2,3,1)} \right] \Big] \xi \rangle \vert \leq C \kappa^2 \| ( \mathcal{N}_+ +1 ) e^{\kappa \cN_+} \xi \|^2 \; . 
\end{align}
The term $\mathcal{G}_N^{(2,3,2)}$ we estimate again in position space and find 
\begin{align}
\label{eq:GN232-dc}
\vert \langle \xi, \Big[  &  e^{\kappa \mathcal{N}_+},  \; \left[  e^{\kappa \cN_+}, \; \mathcal{G}_N^{(2,3,2)} \right] \Big]  \xi \rangle \vert \notag \\
&\leq  \int_{\Lambda \times \Lambda}dxdy \;  v_N (x-y)  \|e^{\kappa \cN_+} ( \mathcal{N}_+ + 1)^{1/2} \xi \| \; \notag \\
& \hspace{1cm} \times \Big( \|( \mathcal{N}_+ + 1)^{-1/2} \Big[   e^{\kappa \mathcal{N}_+},  \; e^{-\kappa \cN_+}\Big[  e^{\kappa \cN_+},\check{b}( \check{\gamma}_x ) \check{d}_y \Big] \Big] e^{\kappa \cN_+}\xi \|  \notag \\
&\hspace{2cm} +  \|  ( \mathcal{N}_+ + 1)^{-1/2} e^{-\kappa \cN_+}\Big[   e^{\kappa \mathcal{N}_+},  \; \Big[  e^{\kappa \cN_+},\check{d}_x\check{b}( \check{\gamma}_y ) \Big] \Big] e^{\kappa \cN_+}\xi \|  \notag \\
& \hspace{2cm} + \| ( \mathcal{N}_+ + 1)^{-1/2} e^{-\kappa \cN_+} \Big[   e^{\kappa \mathcal{N}_+},  \; \Big[  e^{\kappa \cN_+}, \check{d}_x\check{d}_y \Big] \Big] e^{\kappa \cN_+}\xi \|  \Big) . 
\end{align}
We conclude in the same way as in \eqref{eq:GN232} using instead of \eqref{eq:estimates-dd}, \eqref{eq:estimates-db} the estimates \eqref{eq:eN-dp-double-comm}, \eqref{eq:eN-dx-double} of Lemma \ref{lemma:double} (that again differ by a factor $\lambda^2$ only). Thus we get 
\begin{align}
\vert \langle \xi, \Big[  &  e^{\kappa \mathcal{N}_+},  \; \left[  e^{\kappa \cN_+}, \; \mathcal{G}_N^{(2,3,2)} \right] \Big]  \xi \rangle \vert \leq  C \kappa^2 \langle e^{\kappa \cN_+}\xi, ( \mathcal{N}_+ + 1) + \mathcal{V}_N )e^{\kappa \cN_+} \xi \rangle \; . 
\end{align}
For the remaining contributions $\mathcal{G}_N^{(4)}, \mathcal{G}_N^{(5)}$ we proceed similarly as in \eqref{eq:G24}, \eqref{eq:G25} using Lemma \ref{lemma:eN-dp} instead of \eqref{eq:estimates-dp}-\eqref{eq:estimates-db} and thus arrive at \eqref{eq:G2-dc}. 
\end{proof}

\subsection{Analysis of $\mathcal{G}_N^{(3)}$} Next we consider 
\begin{align}
\mathcal{G}_N^{(3)} = e^{-B (\eta)} \mathcal{L}_N^{(3)} e^{B( \eta)} = \frac{1}{\sqrt{N}} \sum_{\substack{p,q \in \Lambda_+^* \\ p+q \not= 0}} e^{-B( \eta)} \left[ b_{p+q}^* a_{-p}^* a_q + {\rm h.c.} \right]e^{B( \eta)} \; . 
 \end{align}
With \eqref{eq:action-bogo} we can approximately compute $e^{-B( \eta)}  b_{p+q}^* e^{B( \eta)}$ while for $e^{-B( \eta)} a_{-p}^* a_q e^{B( \eta)}$ we use a similar idea as in \eqref{eq:wN-1}. We introduce the splitting 
\begin{align}
\label{def:G3} 
\mathcal{G}_N^{(3)} = \sum_{j=1}^4 \mathcal{G}_N^{(3,j)} + {\rm h.c.}
\end{align}
where the single terms $\mathcal{G}_N^{(3,j)} $ are given by 
\begin{align}
\mathcal{G}_N^{(3,1)} =& \frac{1}{\sqrt{N}} \sum_{\substack{p,q, \in \Lambda_+^* \\ p + q \not=0}} \widehat{v} (  p/N) \left[  \gamma_{p+q} \gamma_p \gamma_q b^*_{p+q} b_{-p}^* b_{-q} +   \gamma_{p+q} \gamma_p \sigma_q b^*_{p+q} b_{-p}^* b^*_{q} \right.  \notag \\
&\hspace{4cm} +    \gamma_{p+q} \sigma_p \gamma_q b^*_{p+q} b_{p} b_{-q} +   \gamma_{p+q} \sigma_p \sigma_pq b^*_{p+q}  b^*_{q} b_{p} \notag \\
&\hspace{4cm} + \sigma_{p+q} \gamma_p \gamma_q b_{-p}^* b_{-p-q}  b_{-q} +   \sigma_{-p-q} \gamma_p \sigma_q  b_{-p}^* b^*_{q}  b_{-p-q} \notag \\
&\hspace{4cm} +  \left.  \sigma_{p+q} \sigma_p \gamma_q b_{-p-q} b_{p} b_{-q} +   \sigma_{-p-q} \sigma_p \sigma_p  b^*_{q} b_{-p-q} b_{p} \right] \label{def:G31}
 \end{align}
and 
  \begin{align}
 \label{def:G32}
\mathcal{G}_N^{(3,2)} =&  \frac{1}{\sqrt{N}} \sum_{\substack{p,q, \in \Lambda_+^* \\ p + q \not=0}} \widehat{v} (  p/N)  [ (\gamma_{p} \gamma_{q} + \sigma_p \sigma_q) d_{p+q}^* b_{-p}^*  b_{q} + \gamma_{p} \sigma_{q} d_{p+q}^* b_{-p}^* b_{-q}^* \notag \\ 
& \hspace{4cm} +   \sigma_p \gamma_q d_{p+q}^* b_{p} b_{q}   \Big] + {\rm h.c.}
 \end{align} 
and 
 \begin{align}
 \label{def:G33}
&\mathcal{G}_N^{(3,3)} =  \frac{1}{\sqrt{N}}   \sum_{\substack{p,q, \in \Lambda_+^* \\ p + q \not=0}} \widehat{v} (  p/N) ( \gamma_{p+q} b_{p+q}^* + \sigma_{p+q} b_{p+q}  ) \times  \\
& \times \int_0^1 ds \eta_q \Big( \gamma_{-p}^{(s)} \gamma_q^{(s)} b_{-p}^* b_q^*  + \gamma_{-p}^{(s)} \sigma_q^{(s)} b_{-p}^* b_q + \sigma_{-p}^{(s)} \gamma_q^{(s)} b_{-p} b_q^* + \sigma_{-p}^{(s)} \sigma_q^{(s)} b_{p} b_{-q} + {\rm h.c.}  \Big) \notag
&+ {\rm h.c.}
 \end{align}
and 
 \begin{align}
 \label{def:G34}
& \mathcal{G}_N^{(3,4)} =  \frac{1}{\sqrt{N}}   \sum_{\substack{p,q, \in \Lambda_+^* \\ p + q \not=0}} \eta_q \widehat{v} (  p/N) d_{p+q}^* \times  \\
&\times \int_0^1 ds \Big( \gamma_{-p}^{(s)} \gamma_q^{(s)} b_{-p}^* b_q^*  + \gamma_{-p}^{(s)} \sigma_q^{(s)} b_{-p}^* b_q + \sigma_{-p}^{(s)} \gamma_q^{(s)} b_{-p} b_q^* + \sigma_{-p}^{(s)} \sigma_q^{(s)} b_{p} b_{-q} + {\rm h.c.}  \Big) + {\rm h.c. } \notag
 \end{align}
 
\begin{lemma}
\label{lemma:G3}
Let $\mathcal{G}_N^{(3)}$ be given by \eqref{def:G3}. Then there exists $  \varepsilon, C_\varepsilon >0$ such that 
\begin{align}
\label{eq:G3-lb}
\mathcal{G}_N^{(3)} \geq  - \varepsilon \mathcal{V}_N - C_\varepsilon  ( \mathcal{N} + 1 ) 
\end{align}
as operator inequality on $\mathcal{F}_{\perp u_0}^{\leq N}$.  Furthermore let $\kappa>0$ be sufficiently small, then there exists $C>0$ such that for any $\psi \in \mathcal{F}_{\perp u_0}^{\leq N}$ we have  
\begin{align}
\label{eq:bound-G3-dc}
\vert  \langle \psi, \; \left[ e^{\kappa \mathcal{N}_+}, \;  \left[ e^{\kappa \mathcal{N}_+}, \; \mathcal{G}_N^{(3)} \right]\right] \psi \rangle \vert \leq C  \kappa^2 \vert \langle \psi, \left[  \mathcal{N}_+ + \mathcal{V}_N + 1 \right] \psi \rangle \vert  \; . 
\end{align}
\end{lemma}

\begin{proof} We start with the proof of the lower bounds \eqref{eq:G3-lb} and start with the first summand $\mathcal{G}_N^{(3,1)}$ given by \eqref{def:G31}. To bound the term of the r.h.s. of \eqref{def:G31} we first observe that with the splitting \eqref{eq:sigma-splitting} we have
\begin{align}
\gamma_{p+q} \gamma_{p} \gamma_q  = 1+    \alpha_{q} +\alpha_{p}\gamma_q + \alpha_{p+q}  \gamma_{p} \gamma_{q}  \; .\label{eq:G31-term1-splitting}
\end{align}
To estimate those terms it is convenient to switch to position space. We have 
\begin{align}
 N^{-1/2}\sum_{\substack{p,q \in \Lambda_+^* \\ p+q \not=0}} \widehat{v} ( p/N) \langle \psi, \; b^*_{p+q} b_{-p}^* b_{q} \psi \rangle =  N^{-1/2} \int_{\Lambda \times \Lambda} dxdy \;  v_N (x-y) \langle \psi, \check{b}_x^*\check{b}_y^*\check{b}_x \psi \rangle 
\end{align}
that we can thus estimate using \eqref{eq:bounds-b1}-\eqref{eq:bounds-b3} by 
\begin{align}
\label{eq:G311}
N^{-1/2} & \sum_{\substack{p,q \in \Lambda_+^* \\ p+q \not=0}} \widehat{v} ( p/N) \vert \langle \psi, \; b^*_{p+q} b_{-p}^* b_{q} \psi \rangle\vert \notag \\
& \leq  \left(  N^{-1} \int_{\Lambda \times \Lambda} dxdy \; \vert  v_N (x-y)\vert \;   \| \check{a}_x\check{a}_y \psi \|^2 \right)^{1/2} \left( \int_{\Lambda \times \Lambda} dxdy \; \vert v_N (x-y) \vert \| \check{a}_x \psi \|^2 \right)^{1/2} \; . 
\end{align}
Since $ \sup_x  \int_{\Lambda} dy \; \vert  v_N (x-y) \vert \leq C $ we conclude
\begin{align}
 N^{-1/2} \sum_{\substack{p,q \in \Lambda_+^* \\ p+q \not=0}}&  \widehat{v} ( p/N) \vert \langle \psi, \; b^*_{p+q} b_{-p}^* b_{q} \psi \rangle\vert \leq C  \| \mathcal{V}_N^{1/2} \psi \| \; \| \mathcal{N}^{1/2} \psi \| \; . 
\end{align}

Therefore we find with \eqref{eq:G311} similar arguments as for \eqref{eq:G311} that
\begin{align}
N^{-1/2} & \sum_{\substack{p,q \in \Lambda_+^* \\ p+q \not=0}} \widehat{v} (p/N) \vert \langle \psi, \;\gamma_{p+q} \gamma_{p} \gamma_q  b^*_{p+q} b_{-p}^* b_{q} \psi \rangle\vert \leq C  \| \mathcal{V}_N^{1/2} \psi \| \; \| \mathcal{N}^{1/2} \psi \| \; . 
\end{align}
The second term of the r.h.s. of \eqref{def:G31} we write in position space, too, and find 
\begin{align}
\sum_{\substack{p,q \in \Lambda_+^* \\ p+q \not=0}}&  \widehat{v} ( p/N)  \langle \psi, \; \gamma_{p+q} \gamma_q \sigma_p b^*_{p+q} b_{-p}^* b_{-q}^* \psi \rangle \notag \\
& =  \int_{\Lambda \times \Lambda} dxdy \; v_N (x-y) \;  \langle \psi, \check{b}^*(\check{\gamma}_x) \check{b}^*( \check{\gamma}_y) \check{b}^*( \check{\sigma}_x ) \psi \rangle  \; . 
\end{align}
With the bounds \eqref{eq:bounds-b1}-\eqref{eq:bounds-b3} we find that 
\begin{align}
N^{-1/2} \vert  \sum_{\substack{p,q \in \Lambda_+^* \\ p+q \not=0}}&  \widehat{v} ( p/N)  \langle \psi, \; \sigma_p b^*_{p+q} b_{-p}^* b_{-q}^* \psi \rangle \vert \notag \\
\leq&  \Big(N^{-1}  \int_{\Lambda \times \Lambda} dxdy \; v_N (x-y) \| \check{a}( \check{\gamma}_x) \check{a} ( \check{\gamma}_y) \psi \|^2 \Big)^{1/2} \notag \\
& \quad \times \Big( \int_{\Lambda \times \Lambda} dxdy \; v_N (x-y) \| \check{a}^*( \check{\sigma}_x ) \psi \|^2 \Big)^{1/2}  \; . 
\end{align}
We remark that we have from \cite[Eq. (3.20)-(3.21)]{BCCS}
\begin{align}
\label{eq:sigma-x-bounds-2}
\sup_{x} \| \check{\sigma}_x \|_{L^2( \Lambda)},\quad  \sup_{x} \| \check{\alpha}_x \|_{L^2( \Lambda)}, \quad \sup_{x} \| \check{\beta}_x \|_{L^2( \Lambda)} \leq C
\end{align}
and, in particular, $\| v_N * \| \check{\sigma}_x \|_{L^2( \Lambda)}^2 \|_{L^1 ( \Lambda)} \leq C \| v_N \|_{L^1 ( \Lambda)} \leq C $ so that we arrive with \eqref{eq:sigma-x-bounds} at 
\begin{align}
 N^{-1/2} \vert  \sum_{\substack{p,q \in \Lambda_+^* \\ p+q \not=0}}&  \widehat{v} (p/N)  \langle \psi, \; \gamma_{p+q} \gamma_q \sigma_p b^*_{p+q} b_{-p}^* b_{-q}^* \psi \rangle \vert \leq C  \| \mathcal{V}_N^{1/2} \psi \| \; \| ( \mathcal{N}_+ + 1)^{1/2} \psi \| \; . 
\end{align}
The forth term of the r.h.s. of \eqref{def:G31} can be bounded similarly. For the third term we have 
\begin{align}
N^{-1/2} \sum_{\substack{p,q \in \Lambda_+^* \\ p+q\not=0}} &  \widehat{v} (p/N )  \vert \langle \psi,  \gamma_{p+q} \gamma_q \sigma_p b_{p+q}^*b_p b_q \psi \rangle \vert \notag \\
\leq& \Big( \sum_{\substack{p,q \in \Lambda_+^* \\ p+q\not=0}}  \widehat{v} (p/N )  \sigma_p^2 \| a_{p+q} \psi \|^2 \Big)^{1/2} \notag \\& \hspace{3cm} \times \Big( N^{-1}\sum_{\substack{p,q \in \Lambda_+^* \\ p+q\not=0}}  \widehat{v} (p/N ) \gamma_{p+q}^2 \gamma_q^2 \| a_{p} a_q \psi \|^2 \Big)^{1/2}
\end{align}
Since $\alpha_p, \sigma_p \in \ell^2 (\Lambda_+^*)$ and $\gamma_p \in \ell^\infty ( \Lambda_+^*)$ from \eqref{eq:bounds-sigma} we find that 
\begin{align}
N^{-1/2} \sum_{\substack{p,q \in \Lambda_+^* \\ p+q\not=0}}  &  \widehat{v} (p/N )  \vert \langle \psi,  \gamma_{p+q} \gamma_q \sigma_p b_{p+q}^*b_p b_q \psi \rangle \vert\notag   \leq C  \| ( \mathcal{N} + 1)^{1/2} \psi \|\; \| \mathcal{V}_N^{1/2} \psi \| \; . \label{eq:G314}
\end{align}
The fifth term of the r.h.s. of \eqref{def:G31} follows in the same way while for the sixth term we find with \eqref{eq:bounds-b1}-\eqref{eq:bounds-b3} in position space that  
\begin{align}
N^{-1/2} \sum_{\substack{p,q \in \Lambda_+^* \\ p+q\not=0}} &  \widehat{v} (p/N )  \vert \langle \psi,  \sigma_{p+q} \gamma_p \sigma_q b_{-p}^* b^*_{q}b_{-p-q} \psi \rangle \vert \notag \\
\leq&  N^{-1/2} \Big( \int_{\Lambda \times \Lambda } dxdy v_N (x-y) \| \check{a}( \check{\sigma}_x) \psi \|^2 \Big)^{1/2} \notag \\
& \hspace{3cm} \times \Big( \int_{\Lambda \times \Lambda}dxdy  v_N (x-y)  \| \check{a}( \check{\gamma}_y) \check{a}( \check{\sigma}_x) \psi \|^2 \Big)^{1/2} \; 
\end{align}
and thus we conclude for any $\psi \in \mathcal{F}_{\perp u_0}^{\leq N}$ that 
\begin{align}
N^{-1/2} \sum_{\substack{p,q \in \Lambda_+^* \\ p+q\not=0}} &  \widehat{v} (p/N )  \vert \langle \psi,  \sigma_{p+q} \gamma_p \sigma_q b_{-p}^* b^*_{q}b_{-p-q} \psi \rangle \vert \leq C  \| ( \mathcal{N} + 1)^{1/2} \psi \|^2 \; . 
\end{align}
The remaining terms can be estimated similarly using \eqref{eq:sigma-x-bounds}, \eqref{eq:sigma-x-bounds-2}. For the hermitian conjugate of $\mathcal{G}_N^{(3,1)}$ we can proceed similarly. 

We observe that $\mathcal{G}_N^{(3,2)}$ can be estimated similarly to the first four terms of $\mathcal{G}_N^{(3,1)}$ in \eqref{eq:G311}-\eqref{eq:G314} using \eqref{eq:estimates-dp}-\eqref{eq:estimates-dd}. More precisely we switch in position space and find with \eqref{eq:estimates-dp} for the first term
\begin{align}
N^{-1/2} & \sum_{\substack{p,q \in \Lambda_+^* \\ p+q \not=0}} \gamma_p \gamma_q \vert \langle \xi, d_{p+q}^*b_{-p}^*b_q \xi \rangle \vert \notag \\
 =& N^{-1/2} \int dxdy v_N (x-y) \vert \langle \xi, \check{d}_x^* b^*(\check{\gamma}_y) b ( \check{\gamma}_y) \xi \rangle \vert \notag \\
\leq& \left( N^{-1} \int dxdy v_N (x-y) \| \check{b}( \check{\gamma}_y) \check{d}_x \xi \|^2 \right)^{1/2}  \left( \int dxdy v_N(x-y) \| \check{b}(\check{\gamma}_y) \xi \|^2 \right)^{1/2} 
\end{align}
With \eqref{eq:estimates-db} we get 
\begin{align}
\| b( \check{\gamma}_y) d_x \xi \| \leq CN^{-1} \| ( \cN_+ +1)^2 \xi \|^2 + \| \check{a}_x ( \cN_+ + 1)^{3/2} \xi \| + \| \check{a}_x\check{a}_y ( \cN_+ + 1) \xi \| 
\end{align}
and thus 
\begin{align}
N^{-1/2}  \sum_{\substack{p,q \in \Lambda_+^* \\ p+q \not=0}} \gamma_p \gamma_q \vert \langle \xi, d_{p+q}^*b_{-p}^*b_q \xi \rangle \vert \notag \leq  C \| ( \cN_+ + 1)^{1/2} \xi \| \left( \| ( \cN_+ + 1)^{1/2} \xi \| + \| \mathcal{V}_N^{1/2} \xi \| \right) \; . 
\end{align}
The remaining terms of $\mathcal{G}_N^{(3,2)}$ can be bounded similarly to \eqref{eq:G311}-\eqref{eq:G314} with \eqref{eq:estimates-dp}- \eqref{eq:estimates-dd} and we arrive at 
\begin{align}
\label{eq:G32-estimates}
\vert \langle \xi, \mathcal{G}_N^{(3,2)} \xi \rangle \vert  \leq  C \| ( \cN_+ + 1)^{1/2} \xi \| \left( \| ( \cN_+ + 1)^{1/2} \xi \| + \| \mathcal{V}_N^{1/2} \xi \| \right)  \leq \varepsilon \langle \xi, \mathcal{V}_N \xi \rangle + C_\varepsilon \| ( \mathcal{N}_+ + 1) \xi \|^2 \; . 
\end{align}
The contributions of $\mathcal{G}_N^{(3,3)}$ can be estimated with similar ideas as the second to the seventh term of $\mathcal{G}_N^{(3,1)}$ due to the additional factor $\eta_p$ in the second line of \eqref{def:G33}. In fact we find for the first term 
\begin{align}
N^{-1/2} & \int_0^1 ds \sum_{\substack{p,q \in \Lambda_+^* \\ p+q \not=0}} \eta_q \gamma_{p+q} \gamma_p^{(s)} \gamma_q^{(s)} \vert \langle \xi, b_{p+q}^*b_{-p}^*b_q \xi \rangle \vert \notag \\
 =& N^{-1/2} \int dxdy v_N (x-y) \vert \langle \xi, \check{b}^*( \check{\gamma}_x^{(s)}) \check{b}^*(\check{\gamma}_y^{(s)}) \check{b} ( (\gamma^{(s) \eta)_y}) \xi \rangle \vert \notag \\
\leq& \left( N^{-1} \int dxdy v_N (x-y) \| \check{b}( \check{\gamma}_y) \check{b}( \check{\gamma}_x) \xi \|^2 \right)^{1/2}  \left( \int dxdy v_N(x-y) \| \check{b}( (\gamma \eta)_y) \xi \|^2 \right)^{1/2} \label{eq:estimates-G33}
\end{align}
and similarly as before we get 
\begin{align}
N^{-1/2} & \int_0^1 ds \sum_{\substack{p,q \in \Lambda_+^* \\ p+q \not=0}} \eta_q \gamma_{p+q} \gamma_p^{(s)} \gamma_q^{(s)} \vert \langle \xi, b_{p+q}^*b_{-p}^*b_q \xi \rangle \vert \notag \\
&\leq C \| ( \mathcal{N}_+ + 1)^{1/2} \xi \| ( \| \mathcal{V}_N^{1/2} \xi \| + \| ( \cN_+ + 1)^{1/2} \xi \| ) \; .
\end{align}
The remaining contributions of \eqref{def:G33} can be bounded as in \eqref{eq:G311}-\eqref{eq:G314}. The last term $\mathcal{G}_N^{(3,4)}$ can be bounded as the second term $\mathcal{G}_N^{(3,2)}$ using \eqref{eq:estimates-dp}-\eqref{eq:estimates-db} instead of the bounds \eqref{eq:bounds-b1}-\eqref{eq:bounds-b3} (similarly as for $\mathcal{G}_N^{(3,2)}$). 

To prove \eqref{eq:bound-G3-dc} we again consider the two terms $\mathcal{G}_N^{(3,j)}$ separately. From Lemma \ref{lemma:double} it follows that 
\begin{align}
\vert \langle \psi, \;  \left[ e^{\kappa \mathcal{N}_+}, \;  \left[ e^{\kappa \mathcal{N}_+}, \; \mathcal{G}_N^{(3,1)} \right]\right] \psi \rangle  \leq& C  \sinh^2( \kappa/2) \vert \langle \psi, \; \; e^{\kappa \cN_+}\mathcal{G}_N^{(3,1)} e^{\kappa \cN_+} \psi \rangle \vert 
 \end{align}
and thus with similar arguments as in the first part of this proof we find that for sufficiently small $\kappa >0$
\begin{align}
\vert \langle \psi, &  \left[ e^{\kappa \mathcal{N}_+}, \;  \left[ e^{\kappa \mathcal{N}_+}, \; \mathcal{G}_N^{(3,1)} \right]\right] \psi \rangle \vert \notag \\
\leq& C \kappa^2  \| ( \cN_+ +1 )^{1/2} e^{\kappa \cN_+}\psi \| \; \left(  \| ( \mathcal{N} + 1 ) e^{\kappa \cN_+}\psi \| + \| \mathcal{V}_N^{1/2} e^{\kappa \cN_+} \psi \|  \right) 
\end{align}
For the second term $\mathcal{G}_N^{(3,2)}$ we find with Lemma \ref{lemma:double} 
\begin{align}
& \left[ e^{\kappa \mathcal{N}_+}, \;  \left[ e^{\kappa \mathcal{N}_+}, \; \mathcal{G}_N^{(3,2)} \right]\right] \notag \\
&= \frac{1}{\sqrt{N}} \sum_{\substack{p,q \in \Lambda_+^* \\ p+q\not=0}} \widehat{v}( p/N) \Big( (\gamma_p \gamma_q  + \sigma_p \sigma_q) \left[ e^{\kappa \mathcal{N}_+}, \;  \left[ e^{\kappa \mathcal{N}_+}, \; d^*_{p+q} \right]\right] b^*_{-p} b_q  \notag \\
&\hspace{2cm} + \left[ e^{\kappa \mathcal{N}_+}, \;  \left[ e^{\kappa \mathcal{N}_+}, \; d^*_{p+q} \right]\right] e^{- \kappa \cN_+}(  e^{2 \kappa}\sigma_p \gamma_q b_p b_q+  e^{-2 \kappa}\gamma_p \sigma_q b_pb_q) e^{\kappa \cN_+} \notag \\
&\hspace{2cm} + 2 \sinh( \kappa/2) \left[ e^{\kappa \mathcal{N}_+}, \; d^*_{p+q} \right] (  e^{2 \kappa}\sigma_p \gamma_q b_p b_q+  e^{-2 \kappa}\gamma_p \sigma_q b_pb_q) e^{\kappa \cN_+}  \notag \\
&\hspace{2cm} +  e^{\kappa \cN_+}( e^{-\kappa \cN_+}d^*_{p+q} e^{\kappa \cN_+} )  (  e^{2 \kappa}\sigma_p \gamma_q b_p b_q+  \gamma_p \sigma_q b_pb_q) 
\end{align}
that we can estimate in the same way as $\mathcal{G}_N^{(3,2)}$ using \eqref{eq:eN-db-double}, \eqref{eq:eN-db-single} of Lemma \ref{lemma:eN-dp} instead of \eqref{eq:estimates-dp}. Thus we get 
\begin{align}
\vert \langle \xi, \left[ e^{\kappa \mathcal{N}_+}, \;  \left[ e^{\kappa \mathcal{N}_+}, \; \mathcal{G}_N^{(3,2)} \right]\right] \xi \rangle \vert \leq  C \kappa^2 \Big( \langle e^{\kappa \cN_+} \xi, \mathcal{V}_N e^{\kappa \cN_+} \xi \rangle + \langle e^{\kappa \cN_+}\xi, ( \cN_+ + 1) e^{\kappa \cN_+}\xi \rangle \Big) \; . 
\end{align}
The remaining double commutators of $\mathcal{G}_N^{(3,3)},\mathcal{G}_N^{(3,4)}$ can be bounded with similar ideas, i.e. with Lemma \ref{lemma:double} and Lemma \ref{lemma:eN-dp} instead of \eqref{eq:estimates-dp}, we arrive with similar arguments as in \eqref{eq:estimates-G33} at \eqref{eq:bound-G3-dc}.  
\end{proof}

\subsection{Analysis of $\mathcal{G}_N^{(4)}$} Here we consider the operator 
\begin{align}
\mathcal{G}_N^{(4)} := e^{-B( \eta)} \mathcal{L}_N^{(4)} e^{B ( \eta)} = \frac{1}{2N} \sum_{ \substack{p,q \in \Lambda_+^* \\ r \not=0 -p,q}} \widehat{v}(r/N) e^{-B( \eta)} a_p^*a_q^*a_{q-r} a_{p+r} e^{B( \eta)}
\end{align}
that we compute (following the ideas from \cite[Section 7.4]{BCCS_optimal}
\begin{align}
\mathcal{G}_N^{(4)}  =& \mathcal{V}_N +  \frac{1}{2N} \sum_{ \substack{p,q \in \Lambda_+^*, r \in \Lambda^* \\ r \not= -p,q}} \widehat{v}(r/N) \int_0^1 ds \; e^{-sB( \eta)} [a_p^*a_q^*a_{q-r} a_{p+r}, B( \eta) ] e^{sB( \eta)} \notag\\
=& \mathcal{V}_N + \frac{1}{2N} \sum_{ \substack{p,q \in \Lambda_+^*, r \in \Lambda^* \\ r \not=0 -p,q}} \widehat{v}(r/N) \eta (q+r) \int_0^1 ds \; \Big( e^{-sB( \eta)}  b_q^* b_{-q}^* e^{sB( \eta)} + {\rm h.c.} \Big) \notag \\
&+ \frac{1}{2N} \sum_{ \substack{p,q \in \Lambda_+^*, r \in \Lambda^* \\ r \not=0 -p,q}} \widehat{v}(r/N) \eta (q+r) \int_0^1 ds \; \Big( e^{-sB( \eta)}  b_{p+r}^* b_{q}^* a^*_{-q-r} a_p e^{sB( \eta)} + {\rm h.c.} \Big) \; . 
\end{align}
For the third term of the r.h.s. we observe 
\begin{align}
e^{-sB ( \eta)}  & a^*_{-q-r}a_p e^{sB ( \eta)} \notag \\
=& a^*_{-q-r}a_p + \int_0^s d\tau e^{-\tau B( \eta)} [ a^*_{-q-r} a_p, B ( \eta) ]  e^{\tau B( \eta)}\notag \\
=& a^*_{-q-r}a_p + \int_0^s d\tau e^{-\tau B( \eta)} ( \eta(p) b_{-p}^*b_{-q-r}^* + \eta (q+r) b_pb_{q+r} ) e^{\tau B( \eta)} \; . 
\end{align}
With these formulas we introduce the splitting 
\begin{align}
\label{def:G4}
\mathcal{G}_N^{(4)} = \mathcal{V}_N + \sum_{j=1}^3 \mathcal{G}_N^{(4,j)} + \mathcal{C}_{\mathcal{G}_N^{(4)}} 
\end{align}
with 
\begin{align}
\mathcal{C}_{\mathcal{G}_N^{(4)}} = \frac{1}{2N} \sum_{q \in \Lambda_+^*, r \in \Lambda^*} \widehat{v} (r/N) \eta_{q+r} \eta_q 
\end{align}
and 
\begin{align}
\mathcal{G}_N^{(4,1)} = \frac{1}{2N} \sum_{ \substack{q \in \Lambda_+^*, r \in \Lambda^* }} \widehat{v}(r/N) \eta (q+r) \int_0^1 ds \; \Big( e^{-sB( \eta)}  b_q^* b_{-q}^* e^{sB( \eta)} + {\rm h.c.} \Big) \notag \\
\mathcal{G}_N^{(4,2)} = \frac{1}{2N} \sum_{ \substack{p,q \in \Lambda_+^*, r \in \Lambda^* \\ r \not=0 -p,q}} \widehat{v}(r/N) \eta (q+r) \int_0^1 ds \; \Big( e^{-sB( \eta)}  b_q^* b_{-q}^* e^{sB( \eta)} a^*_{-q-r}a_p  + {\rm h.c.} \Big)
\end{align}
and 
\begin{align}
\mathcal{G}_N^{(4,3)} =& \frac{1}{2N} \sum_{ \substack{p,q \in \Lambda_+^*, r \in \Lambda^* \\ r \not=0 -p,q}} \widehat{v}(r/N) \eta (q+r) \eta (p) \notag \\
& \hspace{2cm} \times \int_0^1 ds\int_0^s d\tau \; \Big( e^{-sB( \eta)}  b_q^* b_{-q}^* e^{sB( \eta)}  e^{-\tau B( \eta)} b_{-p}^*b_{-q-r}^* e^{\tau B( \eta)}  + {\rm h.c.} \Big) 
\end{align}
and 
\begin{align}
\mathcal{G}_N^{(4,4)} =& \frac{1}{2N} \sum_{ \substack{p,q \in \Lambda_+^*, r \in \Lambda^* \\ r \not=0 -p,q}} \widehat{v}(r/N) \eta (q+r)^2 \notag \\
& \hspace{2cm} \times \int_0^1 ds\int_0^s d\tau \; \Big( e^{-sB( \eta)}  b_q^* b_{-q}^* e^{sB( \eta)}  e^{-\tau B( \eta)}b_pb_{q+r} e^{\tau B( \eta)}  + {\rm h.c.} \Big) 
\end{align}
For 
\begin{align}
\label{def:wG4}
\widetilde{\mathcal{G}}_N^{(4)}  = \mathcal{G}_N^{(4)} - \frac{1}{2N} \sum_{\substack{ q \in \Lambda_+^*, r \in \Lambda^*}} \widehat{v}(r/N) \eta_{q+r} ( b_qb_{-q} + b_q^*b_{-q}^*) - \mathcal{C}_{\mathcal{G}_N^{(4)}} 
\end{align}
we then have the following properties. 

\begin{lemma}
\label{lemma:G4}
Let $\mathcal{G}_N^{(4)}$ be given by \eqref{def:G4}. Then there exists $\varepsilon,C_\varepsilon >0$ independent of $N$ such that 
\begin{align}
\label{eq:G4-lb}
\widetilde{\mathcal{G}}_N^{(4)} - \mathcal{V}_N \geq  \varepsilon \mathcal{V}_N -C_\varepsilon  (\mathcal{N}_+ +1) 
\end{align}
as operator inequality on $\mathcal{F}_{\perp u_0}^{\leq N}$.  Furthermore let $\kappa>0$ be sufficiently small, then there exists $C>0$ such that for any $\psi \in \mathcal{F}_{\perp u_0}{\leq N}$ we have 
\begin{align}
\label{eq:G4-dc}
\vert \langle \psi, \;  \left[ e^{\kappa \mathcal{N}_+}, \;  \left[ e^{\kappa \mathcal{N}_+}, \; \widetilde{\mathcal{G}}_N^{(4)} \right]\right] \psi \rangle \vert \leq C\kappa^2  \langle e^{\kappa \cN_+} \psi, \; (\mathcal{V}_N + \mathcal{N}_+  + 1) e^{\kappa \cN_+}\psi \rangle \; . 
 \end{align}
\end{lemma}

\begin{proof} The proof of \eqref{eq:G4-lb} follows from arguments in \cite[Section 7]{BCCS_optimal} that we are briefly recalling here. For this we estimate the single contributions $\mathcal{G}_N^{(4,j)}$ separately. We start with the first that is with \eqref{eq:action-bogo} of the form 
\begin{align}
\mathcal{G}_N^{(4,1)} =& \frac{1}{2N} \sum_{ \substack{q \in \Lambda_+^*, r \in \Lambda^*}} \widehat{v}(r/N) \eta (q+r) \notag \\
&\hspace{2cm} \times \int_0^1 ds \; \Big( \gamma_q^{(s)} b_q^*  + \sigma_q^{(s)} b_{-q} + d_q^{(s)}\Big)\Big(\gamma_q^{(s)}  b_{-q}^*  + \sigma_q^{(s)} b_q + d_{-q}^{(s)} \Big) + {\rm h.c.}  
\end{align}
where $\gamma_q^{(s)} = \cosh( s \eta_q), \sigma_q^{(s)} = \sinh ( s \eta_q)$ and $d^{(s)}_q$ defined in \eqref{def:d-2} with $\eta$ replaced by $s \eta$. We write 
\begin{align}
& \mathcal{G}_N^{(4,1)}   - \frac{1}{2N} \sum_{\substack{ q \in \Lambda_+^*, r \in \Lambda^*}} \widehat{v}(r/N) \eta_{q+r} ( b_qb_{-q} + b_q^*b_{-q}^*) \notag \\
=&   \frac{1}{2N} \sum_{ \substack{q \in \Lambda_+^*, r \in \Lambda^* }} \widehat{v}(r/N) \eta (q+r) \notag \\
&\quad \times \int_0^1 ds \; \Big( ((\gamma_q^{(s)})^2 -1) b_q^* b_{-q}^* +  {\rm h.c.} + (\sigma_q^{(s)})^2 b_{-q} b_q  + 2 \sigma_q^{(s)}\gamma_q^{(s)} b_q^*b_q + {\rm h.c.}\Big)  \notag \\
&+   \frac{1}{2N} \sum_{ \substack{q \in \Lambda_+^*, r \in \Lambda^*}} \widehat{v}(r/N) \eta (q+r) \notag \\
&\quad  \times \int_0^1 ds \; \Big( \Big( \gamma_q^{(s)} b_q^*  + \sigma_q^{(s)} b_{-q} \Big) (d_{-q}^{(s)})^* +(d_q^{(s)})^* \Big(\gamma_q^{(s)}  b_{-q}^*  + \sigma_q^{(s)} b_q  \Big) + (d_{q}^{(s)})^* (d_{-q}^{(s)})^* + {\rm h.c.}\Big)    \notag \\
&+   \frac{1}{2N} \sum_{ \substack{q \in \Lambda_+^*, r \in \Lambda^*}} \widehat{v}(r/N) \eta (q+r) \int_0^1 ds \; \Big( \sigma_q^{(s)}\gamma_q^{(s)} [b_q, b_q^*] + {\rm h.c.} \Big)  \notag \\
=& \sum_{j=1}^3 \mathcal{G}_N^{(4,1,j)}  \label{eq:G41}
\end{align}
For the first summand of \eqref{eq:G41} we use that 
\begin{align}
\sup_{q \in \Lambda_+^*} \frac{1}{N} \sum_{r \in \Lambda^*} \vert \widehat{v} (r/N) \vert \vert \eta_{q+r} \vert \leq C
\end{align}
uniformly in $N$ and $\vert (\gamma_p^{(s)})^2 - 1\vert , \vert \sigma_p^{(s)} \vert \leq C \vert \eta_p \vert $. We find 
\begin{align}
\vert \langle \xi, \mathcal{G}_N^{(4,1,1)} \xi \rangle \vert \leq C \sum_{q \in \Lambda_+^*} \Big[ \vert \eta_q \vert  \| b_q  \xi \|^2 + \| \eta_q^2 \| b_q \xi \| \| ( \cN_+ + 1)^{1/2} \xi \| \Big] \leq C \| ( \cN_+ +1 )^{1/2} \xi \|^2 \; . 
\end{align}
To estimate the second summand of \eqref{eq:G41} we switch (similarly to \eqref{eq:G232}) to position space, and arrive with \eqref{eq:estimates-dp}, 
\begin{align}
\frac{1}{N^2} \sum_{r \in \Lambda^*, q \in \Lambda_+^*} \vert \widehat{v} (r/N) \vert \vert \eta_{q+r} \vert \vert \eta_q \vert  \leq C
\end{align}
and \eqref{eq:estimates-dd} at
\begin{align}
\vert \langle \xi, \mathcal{G}_N^{(4,1,2)} \xi \rangle \vert \leq C \| ( \cN_+ +1 )^{1/2} \xi \| \Big( \| ( \cN_+ + 1)^{1/2} \xi \| + \| \mathcal{V}_N^{1/2} \xi \| \Big) \; . 
\end{align}
For more details see for example \cite[formula (7.62)-(7-64)]{BCCS_optimal}. For the third term of \eqref{eq:G41} we find with the commutation relations \eqref{eq:comm-b}
\begin{align}
\mathcal{G}_N^{(4,1,3)} - \mathcal{C}_{\mathcal{G}_N^{(4)}}=& \frac{1}{2N} \sum_{q \in \Lambda_+^*, r \in \Lambda^*} \int_0^1 ds \gamma_q^{(s)} \sigma_q^{(s)} \Big( N^{-1}\cN_+ -N^{-1} a_q^*a_q - s\eta_q)  
\end{align}
and we find with similar arguments as before 
\begin{align}
\langle \xi, \; \Big( \mathcal{G}_N^{(4,1,3)} - \mathcal{C}_{\mathcal{G}_N^{(4)}} \Big) \xi \rangle \leq C \| ( \cN_+ + 1)^{1/2} \xi \|^2 \; . 
\end{align}
Thus summarizing, we get  for 
\begin{align}
\widetilde{\mathcal{G}}_N^{(4,1)} = \mathcal{G}_N^{(4,1)}  - \frac{1}{2N} \sum_{\substack{ q \in \Lambda_+^*, r \in \Lambda^*}} \widehat{v}(r/N) \eta_{q+r} ( b_qb_{-q} + b_q^*b_{-q}^*) - \mathcal{C}_{\mathcal{G}_N^{(4)}}
\end{align}
that 
\begin{align}
\vert \langle \xi, \mathcal{G}_N^{(4,1)} \xi \rangle \vert \leq C \| ( \cN_+ +1 )^{1/2} \xi \| \Big( \| ( \cN_+ + 1)^{1/2} \xi \| + \| \mathcal{V}_N^{1/2} \xi \| \Big) 
\end{align}
To bound $\mathcal{G}_N^{(4,2)}$ we switch to position space and find 
\begin{align}
\vert \langle \xi, \mathcal{G}_N^{(4,2)} \xi \rangle \vert \leq&  \frac{1}{N}\int dxdy v_N (x-y) \int_0^1 ds \notag \\
& \hspace{1cm}\times \| ( \cN_+ + 1)^{1/2} e^{-sB ( \eta)} \check{b}_x \check{b}_y e^{s B ( \eta)} \xi \| \; \| ( \cN_++ 1)^{-1/2} a^*( \eta_x) \check{a}_y \xi \| \; .
\end{align}
On the one hand 
\begin{align}
\|  ( \cN_+ + 1)^{-1/2} a^* ( \eta_x) \check{a}_y \xi \| \leq C \| \eta \| \| \check{a}_y \xi \| \leq C \| \check{a}_y \xi \| \label{eq:bb-eB}
\end{align}
and on the other hand with \eqref{eq:sigma-x-bounds}, \eqref{eq:sigma-x-bounds-2} and \eqref{eq:estimates-dp} 
\begin{align}
\| & ( \cN_+ + 1)^{1/2} e^{-sB ( \eta)} \check{b}_x \check{b}_y e^{s B ( \eta)}  \xi \| \notag \\
& \leq C \Big( N \| ( \cN_+ + 1)^{1/2} \xi \| + N \| \check{a}_x \xi \| + N \| \check{a}_y \xi \| + N^{1/2} \| \check{a}_x \check{a}_y \xi \| \Big) 
\end{align}
so that we arrive at 
\begin{align}
\vert \langle \xi, \mathcal{G}_N^{(4,2)} \xi \rangle \vert \leq& C \| ( \cN_+ +1 )^{1/2} \xi \| \Big( \| ( \cN_+ + 1)^{1/2} \xi \| + \| \mathcal{V}_N^{1/2} \xi \| \Big) \; . 
\end{align}
For the third term we work again in position space and argue similarly as 
\begin{align}
\vert \langle \xi, \mathcal{G}_N^{(4,3)} \xi \rangle \vert &\leq \int dxdy v_N(x-y) \int_0^1 ds \int_0^s d \tau \| ( \cN_+ + 1)^{1/2} e^{-sB (\eta)} \check{b}_x \check{b}_y e^{sB( \eta)} \xi \| \notag \\
& \hspace{2cm} \times \| ( \cN_+ + 1)^{-1/2} e^{-\tau B ( \eta)} b^*( \check{\eta}_x ) b^*( \check{\eta}_y) e^{\tau B ( \eta)} \xi \| 
\end{align}
and \eqref{eq:estimates-dp*}
\begin{align}
\|  ( \cN_+ + 1)^{-1/2} e^{- \tau B (\eta)} \check{b}^* ( \eta_x) \check{b}^*( \eta_y) e^{\tau B( \eta)} \xi \| \leq C \| \eta \|^2 \| ( \cN_+ + 1)^{1/2} \xi \| 
\end{align}
and thus with \eqref{eq:bb-eB} 
\begin{align}
\vert \langle \xi, \mathcal{G}_N^{(4,3)} \xi \rangle \vert \leq& C \| ( \cN_+ +1 )^{1/2} \xi \| \Big( \| ( \cN_+ + 1)^{1/2} \xi \| + \| \mathcal{V}_N^{1/2} \xi \| \Big) \; . 
\end{align}
The forth term can be estimated in position space by 
\begin{align}
\vert \langle \xi, \mathcal{G}_N^{(4,4)} \xi \rangle \vert &\leq \int dxdy v_N(x-y) \int_0^1 ds \int_0^s d \tau \| ( \cN_+ + 1)^{1/2} e^{-sB (\eta)} \check{b}_x \check{b}_y e^{sB( \eta)} \xi \| \notag \\
& \hspace{2cm} \times \| ( \cN_+ + 1)^{-1/2} e^{-\tau B ( \eta)} b( \check{\eta}_x^2 ) \check{b}_y e^{\tau B ( \eta)} \xi \| 
\end{align}
and thus with \eqref{eq:bb-eB} and \eqref{eq:estimates-dp} 
\begin{align}
\vert \langle \xi, \mathcal{G}_N^{(4,4)} \xi \rangle \vert &\leq \int dxdy v_N(x-y) \int_0^1 ds \int_0^s d \tau \| ( \cN_+ + 1)^{1/2} e^{-sB (\eta)} \check{b}_x \check{b}_y e^{sB( \eta)} \xi \| \notag \\
& \hspace{2cm} \times \|e^{-\tau B ( \eta)} \check{b}_y e^{\tau B ( \eta)} \xi \| \notag \\
\leq& C \| ( \cN_+ +1 )^{1/2} \xi \| \Big( \| ( \cN_+ + 1)^{1/2} \xi \| + \| \mathcal{V}_N^{1/2} \xi \| \Big)  \; . 
\end{align}
We finally conclude by 
\begin{align}
\vert \langle \xi, \mathcal{G}_N^{(4,4)} \xi \rangle \vert
\leq& C \| ( \cN_+ +1 )^{1/2} \xi \| \Big( \| ( \cN_+ + 1)^{1/2} \xi \| + \| \mathcal{V}_N^{1/2} \xi \| \Big)
\end{align}

To prove the upper bound \eqref{eq:G4-dc} on the second nested commutator of $\mathcal{G}_N^{(4)}$ we first observe that since $[\mathcal{N}_+, \mathcal{V}_N ] =0$ we have 
\begin{align}
\left[ e^{\kappa \mathcal{N}_+}, \;  \left[ e^{\kappa \mathcal{N}_+}, \; \widetilde{\mathcal{G}}_N^{(4)} \right]\right] = \Big[ e^{\kappa \mathcal{N}_+}, \; \Big[ e^{\kappa \mathcal{N}_+}, \;   \sum_{j=1}^4 \mathcal{G}_N^{(4,j)} \Big]\Big] \; . 
\end{align}
Thus it suffices to control the second nested commutator of the single contributions $\mathcal{G}_N^{(4,j)}$. For this we proceed analogously as in the proof of the previous lemmas on nested commutators of $\mathcal{G}_N^{(2)}, \mathcal{G}_N^{(3)}$. That is that we the estimates before as we the only ingredient for our estimates were either bounds on $b_p^*,b_p$ by \eqref{eq:bounds-b1}-\eqref{eq:bounds-b3} or bounds on $d_p,d_p^*$ and $\check{d}_x \check{d}_y$ by \eqref{eq:estimates-dp}, \eqref{eq:estimates-dp*} or \eqref{eq:estimates-dd} respectively. However, bounds on single and double commutators of $b_p^*,b_p$, $d_p,d_p^*$ and $\check{d}_x \check{d}_y$ are given by Lemmas \eqref{lemma:double}, \ref{lemma:eN-dp} and agree with \eqref{eq:bounds-b1}-\eqref{eq:bounds-b3},  \eqref{eq:estimates-dp}, \eqref{eq:estimates-dp*} and \eqref{eq:estimates-dd} respectively modulus a factor of $\kappa$  for the single and $\kappa^2$ for the double commutator. Thus we conclude with \eqref{eq:G4-dc}. 

\end{proof}

\subsection{Conclusion of Proposition \ref{prop:G}}  

Here we proof Proposition \ref{prop:G} from Lemmas \ref{lemma_G0}-\ref{lemma:G4}. 
\label{subsec:proof-G}

\begin{proof}[Proof of Proposition \ref{prop:G}] First we remark that it follows from \cite[Section 7]{BCCS} that with the choice of $\eta$ in \eqref{def:eta}, we have for $\mathcal{C}_{\mathcal{G}_N} := \mathcal{C}_{\mathcal{G}_N^{(0)}} + \mathcal{C}_{\mathcal{G}_N^{(2)}} + \mathcal{C}_{\mathcal{G}_N^{(4)}}$
\begin{align}
\vert \mathcal{C}_{\mathcal{G}_N} - E_N \vert \leq C 
\end{align}
for a constant $C>0$. In order to prove the lower bound  \eqref{eq:G-lb}, we collect the results from Lemma \ref{lemma_G0}-\ref{lemma:G4} that lead for to 
\begin{align}
 \mathcal{G}_N -E_N \geq \frac{1}{2}  \mathcal{H}_N  - C_1  \langle \xi_N, \; \cN_+  \xi_N \rangle - C_2   \; . \label{eq:bound-G-eins}
\end{align}
(see also \cite[Proposition 3.2]{BCCS}). Furthermore, Ii follows from \cite[p.250]{BCCS} that there exist $C_1,C_2>0$ such that 
\begin{align}
\mathcal{G}_N - E_N \geq C \mathcal{N}_+ - C_2 
\end{align}
that plugging into \eqref{eq:bound-G-eins} yields the first bound \eqref{eq:G-lb} of Proposition \ref{prop:G} (see also  \cite[Eq. (4.5)]{BCCS} ). 


The second bound \eqref{eq:G-dc} follows immediately from Lemma \ref{lemma_G0}-\ref{lemma:G4}. 
\end{proof}

\section{Proof of main theorems} 
\label{sec:proof-thm}

In this section we conclude the proof of the main results. 

\subsection{Proof of Theorem \ref{thm:main}}
\begin{proof}
We introduce the notation 
\begin{align}
\xi_N := e^{B( \eta)} \mathcal{U}_N \psi_N
\end{align}
for the ground state of the excitation Hamiltonian $\cG_N$ defined in \eqref{def:G}. First we prove that there exits $C,c_0>0$ such that for sufficiently small $\widetilde{\kappa} >0$ we have 
\begin{align}
\langle \psi_N, e^{\widetilde{\kappa} \cN_+} \psi_N \rangle \leq C e^{\widetilde{\kappa} } \langle \xi_N, e^{c_0 \widetilde{\kappa}  \cN_+} \xi_N \rangle \; . \label{eq:claim0}
\end{align}
and thus, that it sufficies to consider the expectation value of $e^{\kappa \cN_+} = e^{c_0\widetilde{\kappa} \cN_+}$ in the excitation vector $\xi_N$ to prove Theorem \ref{thm:main}. For the proof of \eqref{eq:claim0}, we recall that with the definition of \eqref{def:wN} that 
\begin{align}
\langle \psi_N, e^{\widetilde{\kappa}  \cN_+} \psi_N \rangle = \langle \xi_N, e^{\widetilde{\kappa}  \wN} \xi_N \rangle \;. 
\end{align}
For $s \in [0,1]$ and $c_0 >0$ we define the Fock space vector 
\begin{align}
\label{eq:wN-1}
\xi_N(s) = e^{(1-s) \widetilde{\kappa}  c_0 \cN_+ /2} e^{s \widetilde{\kappa}  \wN/2} \xi_N 
\end{align}
that satisfies 
\begin{align}
\| \xi_N (1) \|^2 = \langle \xi_N, e^{\widetilde{\kappa}  \wN} \xi_N \rangle, \quad \text{and} \quad \| \xi_N (0) \|^2 = \langle \xi_N, e^{c_0 \widetilde{\kappa}  \cN_+} \xi_N \rangle \; . 
\end{align}
Therefore, to prove \eqref{eq:claim0}, we need to control the difference of $\| \xi_N (0) \|^2$ and $\| \xi_N (1) \|^2$. For this we compute 
\begin{align}
\label{eq:partial-xi-s-0}
\partial_s \| \xi_N (s) \|^2 = 2  \widetilde{\kappa}  \Re \langle \xi_N (s),  \left( e^{(1-s) c_0 \widetilde{\kappa}  \cN_+/2}  \wN e^{-(1-s)c_0 \widetilde{\kappa}  \cN_+/2} - c_0 \cN_+ \right) \xi_N (s) \rangle \; . 
\end{align} 
It follows from Lemma \ref{lemma:bounds-wN} that for $\widetilde{\kappa}  c_0 \leq 1$ we have 
\begin{align}
\vert \Re  \langle \xi_N (s),   e^{(1-s) c_0 \widetilde{\kappa}  \cN_+/2}  \wN e^{-(1-s)c_0 \kappa \cN_+/2}  \xi_N (s) \rangle  \vert \leq C \| (\mathcal{N}_+ +1 ) \xi_N(s) \|^2 
\end{align}
for a constant $C >0$. Thus for $c_0 >C$ (that exists for $\kappa>0$ sufficiently small) we have from \eqref{eq:partial-xi-s-0} 
\begin{align}
\partial_s \| \xi_N (s) \|^2  \leq 2 \widetilde{\kappa}  \langle \xi_N (s), \left[ \left( C- c_0\right) \cN_+  + C \right] \xi_N (s) \rangle \leq C \widetilde{\kappa}  \| \xi_N (s) \|^2 \label{eq:wN-2}
\end{align}
yielding with Gronwall's inequality the desired estimate \eqref{eq:claim0}. 

We recall that \eqref{eq:claim0} implies that in order to prove Theorem \ref{thm:main}, it suffices to prove that for sufficiently small $\kappa >0$ there exists $C>0$ such that 
\begin{align}
\langle \xi_N, e^{\kappa \cN_+} \xi_N \rangle \leq  e^{C\kappa}  \;  
\end{align}
To this end we show as a preliminary step that there exists $C>0$ such that 
\begin{align}
\label{eq:bound-N-1}
\langle e^{\kappa \cN_+} \xi_N, \; \mathcal{N}_+ e^{\kappa \cN_+} \xi_N \rangle \leq C  \| e^{\kappa \cN_+} \xi_N\|^2. 
\end{align}
We observe that since $\mathcal{N}_+ \leq C \mathcal{H}_N$, instead of \eqref{eq:bound-N-1}, it suffices to show that 
\begin{align}
\langle e^{\kappa \cN_+} \xi_N, \; \mathcal{H}_N \; e^{2\kappa \cN_+} \xi_N \rangle \leq C  \| e^{\kappa \cN_+} \xi_N\|^2 \; 
\end{align}
for a positive constant $C>0$. From \eqref{eq:G-lb} of Proposition \ref{prop:G} it follows that there exists $C_1, C_2 >0$ such that 
\begin{align}
\label{eq:bound-H1} 
\langle e^{\kappa \cN_+} \xi_N, \; \mathcal{H}_N e^{\kappa \cN_+} \xi_N \rangle \leq C_1 \langle e^{\kappa \cN_+} \xi_N, \; \left( \mathcal{G}_N  - E_N \right) e^{\kappa \cN_+} \xi_N \rangle + C_2  \| e^{\kappa \cN_+} \xi_N \| \; . 
\end{align}
We recall that $\xi_N$ is the ground state of $\mathcal{G}_N$, i.e. satisfies $\mathcal{G}_N \xi_N = E_N \xi_N $. Therefore we have 
\begin{align}
2 \langle \xi_N, \; & e^{ \kappa \mathcal{N}_+} (\mathcal{G}_N - E_N ) e^{ \kappa \mathcal{N}_+} \xi_N \rangle \notag \\
& = \langle \xi_N , \; \left[ e^{ \kappa \mathcal{N}_+} , \; \mathcal{G}_N \right] \;  e^{ \kappa \mathcal{N}_+} \xi_N \rangle +  \langle \xi_N,  \;  e^{ \kappa \mathcal{N}_+} \; \left[ \mathcal{G}_N,  \;  e^{ \kappa \mathcal{N}_+} \right] \xi_N \rangle \notag \\
&= \langle \xi_N , \; \left[ e^{ \kappa \mathcal{N}_+} , \; \mathcal{G}_N \right] \;  e^{ \kappa \mathcal{N}_+} \xi_N \rangle - \langle \xi_N,  \;  e^{ \kappa \mathcal{N}_+} \; \left[  e^{ \kappa \mathcal{N}_+} ,  \;  \mathcal{G}_N, \right] \xi_N \rangle \notag \\ 
&= -   \langle \xi_N,  \;  \left[ e^{ \kappa \mathcal{N}_+} , \; \; \left[  e^{ \kappa \mathcal{N}_+} ,  \;  \mathcal{G}_N \right] \right]  \xi_N \rangle\; .  \label{eq:claim11} 
\end{align}
yielding with \eqref{eq:bound-H1}
\begin{align}
\langle e^{\kappa \cN_+} \xi_N, \; \mathcal{H}_N e^{\kappa \cN_+} \xi_N \rangle \leq C_1 \langle \xi_N, \; \left[ e^{\kappa \cN_+} , \; \left[ e^{\kappa \cN_+} , \mathcal{G}_N  \right] \right] e^{\kappa \cN_+} \xi_N\rangle + C_2 \|  e^{\kappa \cN_+} \xi_N \|^2 \; . 
\end{align}
From \eqref{eq:G-dc} of Proposition \ref{prop:G} we furthermore find 
\begin{align}
\langle e^{\kappa \cN_+} \xi_N, \; \mathcal{H}_N e^{\kappa \cN_+} \xi_N \rangle \leq C_1 \kappa^2  \langle e^{\kappa \cN_+} \xi_N, \; \mathcal{H}_N e^{\kappa \cN_+} \xi_N \rangle + C_2 \| e^{\kappa \cN_+} \xi \|^2 \
\end{align}
for sufficiently small $\kappa>0$. Thus 
\begin{align}
( 1- C_1 \kappa^2) \langle e^{\kappa \cN_+} \xi_N, \; \mathcal{H}_N e^{\kappa \cN_+} \xi_N \rangle \leq  C_2 \| e^{\kappa \cN_+} \xi \|^2  \label{eq:bound-H-end}
\end{align}
and we arrive with for sufficiently small $\kappa>0$ at 
\begin{align}
 \langle e^{\kappa \cN_+} \xi_N, \; \mathcal{N}_+ e^{\kappa \cN_+} \xi_N \rangle \leq \langle e^{\kappa \cN_+} \xi_N, \; \mathcal{H}_N e^{\kappa \cN_+} \xi_N \rangle \leq C \| e^{\kappa \cN_+} \xi \|^2 \;  \label{eq:bound-N-final}
\end{align}
where the first estimate follows from the gap of the kinetic energy and $v \geq 0$. 
Next we use \eqref{eq:bound-N-1} to prove Theorem \ref{thm:main}. To this end we define for $s \in [0,1]$ the Fock space vector 
\begin{align}
\label{eq:xi-s-beginning}
\xi_{N} (s) := e^{s \kappa \mathcal{N}_+} \xi_N \; . 
\end{align}
Then we have 
\begin{align}
 \| \xi_N (1) \|^2 =  \| e^{\kappa \cN_+} \xi_N \|^2 \quad \text{and} \quad \| \xi_N (0) \|^2 = \| \xi_N \|^2 =1  \label{eq:xsi-eins}
\end{align}
thus, to control $\| \xi_N (1) \|^2$ for sufficiently small $\kappa$ it thus suffices to control the derivative $\partial_s \| \xi_N (s) \|^2$. We compute 
\begin{align}
\partial_s \| \xi_N (s) \|^2 =  2\kappa \langle \xi_N (s), \mathcal{N}_+  \xi_N (s) \rangle 
\end{align}
and arrive with \eqref{eq:bound-N-final} for sufficiently small $\kappa>0$ at 
\begin{align}
\vert \partial_s \| \xi_N (s) \|^2 \vert \leq    C \kappa   \langle \xi_N (s),   \xi_N (s) \rangle \; . 
\end{align}
With Gronwall's inequality we obtain $\| \xi_N (1) \|^2 \leq e^{C \kappa} \| \xi_N (0) \|^2 =  e^{C \kappa} $. Thus  the desired estimate 
\begin{align}
\label{eq:xi-s-end}
 \langle \xi_N, \; e^{2 \kappa \mathcal{N}_+} \xi_N \rangle \leq  e^{C \kappa} \; . 
\end{align}
follows.

The proof for excited states $\xi_N^{(k)}$ with $k \in \mathbb{N}$ satisfying \eqref{ass:ev} follows similarly, i.e. by estimating 
\begin{align}
\langle e^{\kappa \cN_+} \xi_N^{(k)}, \; \mathcal{H}_N \; e^{2\kappa \cN_+} \xi_N^{(k)} \rangle \leq  C_1 \langle e^{\kappa \cN_+} \xi_N^{(k)}, \; \left( \mathcal{G}_N  - E_N \right) e^{\kappa \cN_+} \xi_N^{(k)} \rangle + C_2  \| e^{\kappa \cN_+} \xi_N^{(k)} \| \; . 
\end{align}
For the first term we use instead of the eigenvalue equation $( \mathcal{G}_N - E_N) \xi_N =0$ for the ground state, that $( \mathcal{G}_N - E^{(k)}_N) \xi_N^{(k)} = 0 $ and $E_N^{(k)} = E_N + C$ for some $C>0$ so that \eqref{eq:claim11} becomes 
 \begin{align}
 2 \langle \xi_N^{(k)}, \; & e^{ \kappa \mathcal{N}_+} (\mathcal{G}_N - E_N ) e^{ \kappa \mathcal{N}_+} \xi_N^{(k)} \rangle \notag \\
&= -   \langle \xi_N^{(k)},  \;  \left[ e^{ \kappa \mathcal{N}_+} , \; \; \left[  e^{ \kappa \mathcal{N}_+} ,  \;  \mathcal{G}_N \right] \right]  \xi_N^{(k)} \rangle + C \| e^{\kappa \cN_+} \xi_N^{(k)} \|^2  \; . 
 \end{align}
that leads to 
\begin{align}
\langle e^{\kappa \cN_+} \xi_N^{(k)}, \; \mathcal{H}_N \; e^{2\kappa \cN_+} \xi_N^{(k)} \rangle \leq  C_3 \vert \langle \xi_N^{(k)},  \;  \left[ e^{ \kappa \mathcal{N}_+} , \; \; \left[  e^{ \kappa \mathcal{N}_+} ,  \;  \mathcal{G}_N \right] \right]  \xi_N^{(k)} \rangle \vert + C_4 \| e^{\kappa \cN_+} \xi_N^{(k)} \|^2 
\end{align}
and we can follows the lines of the proof for the ground state. 
\end{proof}

\subsection{Proof of Corollary \ref{thm:LDE}}
\label{sec:proof-LDE}

In this section, we provide the proofs of Corollary \ref{thm:LDE} and Eq. \ref{eq:LDE-O}. 

\begin{proof}[Proof of Corollary \ref{thm:LDE}] The proof is based on Theorem \ref{thm:main} and ideas from \cite[Appendix A]{BCCS}, where the expectation value  $ \cN_+ $ in the ground state $\psi_N$ is computed explicitly. First, let us explain how to extend the computation of $\langle \psi_N, \big( \cN_+ - \mu \big) \psi_N\rangle$ from \cite[Appendix A]{BCCS} to get $\langle \psi_N, \big(\cN_+ - \mu \big)^2 \psi_N\rangle$. The key input from \cite[Eq. (6.7)]{BCCS} is that, up to a phase factor of $\psi_N$, we have the norm approximation  
\begin{align} \label{eq:psiN-norm}
\| U_N \psi_N - e^{B(\eta)} e^A e^{B(\tau)} \Omega\|^2 \le C N^{-1/4}.  
\end{align}
where $e^{A}$ and $e^{B(\tau)}$ are unitary transformations with the cubic kernel
$$
A = \frac{1}{\sqrt{N}} \sum_{r\in P_H, v\in P_L} \eta_r (\sigma_v b^*_{r+v} b^*_{-r}b^*_{-v}+ \gamma_v b^*_{r+v}b^*_{-r}b_v - {\rm h.c.}) = A_\sigma+ A_\gamma- {\rm h.c.}
$$
and the quadratic kernel
$$
B(\tau)= \frac{1}{2}\sum_{p\in \Lambda_+^*} (b_p^* b_{-p}^* - b_{-p}b_p),\quad \tanh(2\tau_p) = -\frac{G_p}{F_p}. 
$$
(see \cite[Eq. (3.34)]{BCCS} and  \cite[Eq.  (5.9)]{BCCS}, respectively). 

We will prove that the norm approximation \eqref{eq:psiN-norm} still holds true if the cubic transformation $e^{A}$ is removed (a similar idea was used recently in \cite{COS} to study the norm approximation in the dynamical problem). To see this, we use the pointwise estimates 
\begin{align}
\label{eq:bounds-eta,tau}
 |\tau_p|\le C |p|^{-2},\quad |\sigma_p|\le C,\quad |\gamma_p|\le C 
\end{align}
for all $p\in \Lambda^*_+$, as well as the stability estimates 
\begin{align}\label{eq:stability-B-A}
e^{-tB(\tau)} (\cN_++1)^k e^{tB(\tau)} \le C_k  (\cN_++1)^k, \quad e^{-tA} (\cN_++1)^k e^{tA} \le C_k  (\cN_++1)^k.
\end{align}
Here in \eqref{eq:stability-B-A}, the first estimate is similar to the bound for  $e^{B(\eta)}$, while the second estimate was discussed in \cite[Prop. 4.2]{BCCS}. We start by noting that
\begin{align*}
\| b_v e^{B(\tau)}\Omega\|^2 &= \int_0^1 dt \partial_t \| b_v e^{tB(\tau)}\Omega\|^2 = \int_0^1 dt \langle \Omega, e^{-tB(\tau)}[b_v^* b_v,B(\tau)] e^{tB(\tau)}\Omega \rangle \\
&\le \int_0^1 dt \langle \Omega, e^{-tB(\tau)} (\eta_v (\cN_+ +1) )e^{tB(\tau)}\Omega \rangle \le C |v|^{-2}
\end{align*}
for all $v\in \Lambda^*_+$.  Therefore, 
\begin{align*}
|\langle \xi_1, A_\sigma \xi_2 \rangle| &\le \frac{1}{\sqrt{N}} \sum_{r\in P_H, v\in P_L} |\eta_r \sigma_v|  | \langle b_{-v} b_{r+v} b_{-r}\xi_1,\xi_2\rangle|  \\
&\le  \frac{C}{\sqrt{N}} \sum_{|r| \ge N^{1/2}\ge v} |r|^{-2} |v|^{-2} \| b_{r+v} b_{-r}\xi_1 \| \|\xi_2\| \\
& \le  \frac{C}{\sqrt{N}} \left( \sum_{|r| \ge N^{1/2}\ge v} |r|^{-4} |v|^{-4} \|\xi_2\|^2 \right)^{1/2} \left( \sum_{|r| \ge N^{1/2}\ge v}  \| b_r b_{-v} \xi_1\|^2 \right)^{1/2}\\
&\le  \frac{C}{N^{3/4}} \|\xi_2\|^2 \| (\cN+1)^2 \xi_1\| 
\end{align*}
for all vectors $\xi_1,\xi_2\in \cF_+^{\le N}$. By similar estimates for $A_\gamma$, we also find that 
\begin{align}
|\langle \xi_1, A \xi_2 \rangle| \le \frac{C}{N^{3/4}} \| (\cN+1)^2 \xi_1\| \| (\cN+1)^2 \xi_2\|.  
\end{align}
Consequently,
\begin{align*}
1- \langle e^{B(\tau)}  \Omega, e^{A}e^{B(\tau)} \Omega \rangle &= - \int_0^1 dt \partial_t \langle e^{B(\tau)}  \Omega, e^{tA}e^{B(\tau)}  \Omega \rangle = - \int_0^1 dt   \langle e^{B(\tau)}  \Omega, A e^{tA}e^{B(\tau)}  \Omega\rangle \\
&\le  \frac{C}{N^{3/4}} \| (\cN+1)^2 e^{B(\tau)}  \Omega\| \| (\cN+1)^2 e^{tA}e^{B(\tau)}  \Omega \| \le  \frac{C}{N^{3/4}}
\end{align*}
which is equivalent to 
\begin{align} \label{eq:psiN-norm-2}
\| e^{B(\eta)} e^{B(\tau)}  \Omega - e^{B(\eta)} e^{A}e^{B(\tau)} \Omega\|^2  = \| e^{B(\tau)}  \Omega - e^{A}e^{B(\tau)} \Omega\|^2 \le  \frac{C}{N^{3/4}}. 
\end{align}  
Using \eqref{eq:psiN-norm}, \eqref{eq:psiN-norm-2} and the triangle inequality we obtain 
 \begin{align} \label{eq:psiN-norm-3}
\| U_N \psi_N - e^{B(\eta)}e^{B(\tau)} \Omega\|^2 \le  \frac{C}{N^{1/4}}. 
\end{align}  

Next, we use the simplified norm approximation  \eqref{eq:psiN-norm-3} to compute the r.h.s. of \eqref{eq:log-mom} for which we need the first and second moment of the operator $(\cN_+ - \mu\big)$ in expectation of the ground state. We argue similarly as in  \cite[Eq.  (A.1)]{BCCS} to find that
\begin{align}\label{eq:Psi-cNk-trans}
\Big| \langle \psi_N, \big(\cN_+ - \mu\big)^k \psi_N\rangle - \langle  e^{B(\eta)}  e^{B(\tau)}  \Omega, \big(\cN_+ - \mu\big)^k  e^{B(\eta)}  e^{B(\tau)}  \Omega \rangle \Big| \le C_k N^{-1/8}  
\end{align}
for every $k=1,2$. Moreover, 
\begin{align}
\label{eq:Psi-cNk-trans-b}
\Big|  \langle  e^{B(\eta)}  e^{B(\tau)}  \Omega, \big( \cN_+ - \mu \big)^k - \bigg( \sum_{p\in \Lambda_+^*} b_p^* b_p - \mu \bigg)^k  e^{B(\eta)}  e^{B(\tau)}  \Omega \rangle \Big| \le C_k N^{-1/8}  .
\end{align}
for $k =1,2$ and thus it remains to compute 
$$ \langle  e^{B(\eta)}  e^{B(\tau)}  \Omega, \bigg(\sum_{p\in\Lambda_+^*} b_p^* b_p - \mu\bigg)^k  e^{B(\eta)}  e^{B(\tau)}  \Omega \rangle.$$
for $k=1,2$. By using \eqref{eq:action-bogo} and \eqref{eq:estimates-dp*} we can estimate 
$$ e^{-B(\eta)} b_p  e^{B(\eta)} \approx \gamma_p b_p +\sigma_p b_{-p}^*.$$
and 
$$ e^{-B(\tau)}   e^{-B(\eta)}  b_p  e^{B(\eta)} e^{B(\tau)}  \approx  (\gamma_p \cosh \tau_p+\sigma_p \sinh \tau_p)  b_p^* + (\gamma_p \sinh \tau_p + \sigma_p \cosh \tau_p) b_{-p}  .$$
More precisely, using that by the properties of the hypergeometric functions 
\begin{align}
\gamma_p \cosh \tau_p+\sigma_p \sinh \tau_p = \cosh (\eta_p  + \tau_p), \quad 
\gamma_p \sinh \tau_p + \sigma_p \cosh \tau_p = \sinh (\eta_p + \tau_p)  
\end{align}
and $\vert \eta_p + \tau_p - \nu_p \vert \leq C N^{-1}$ from \cite[Section 3]{RS} with $\nu_p$ given by \eqref{def:sigma0}, we have introducing the notations 
\begin{align}
\widetilde{\gamma}_p := \cosh ( \nu_p ), \quad \widetilde{\sigma}_p := \sinh( \nu_p ) 
\end{align}
we find from \eqref{eq:action-bogo}, \eqref{eq:estimates-dp*}, \eqref{eq:bounds-eta,tau} and Theorem \ref{thm:main} 
\begin{align}
\big\vert \langle  & e^{B( \eta)} e^{B( \tau)} \Omega, \bigg( \sum_{p \in \Lambda_+^*} b_p^*b_p -\mu \bigg)^k e^{B( \eta)} e^{B( \tau)} \Omega \rangle \notag \\
& \hspace{2cm} - \langle \Omega, \bigg( \sum_{p \in \Lambda_+^*} \big[ \widetilde{\gamma}_p b_p^* + \widetilde{\sigma}_p b_{-p} \big]  \big[ \widetilde{\gamma}_p b_p + \widetilde{\sigma}_p b_{-p}^* \big] - \mu \bigg)^k \Omega \rangle \big\vert \leq CN^{-1/2} \; . 
\end{align}
It remains to compute the expectation value 
\begin{align}
\label{eq:leftover}
\langle \Omega, \bigg( \sum_{p \in \Lambda_+^*} \big[ \widetilde{\gamma}_p b_p^* + \widetilde{\sigma}_p b_{-p} \big]  \big[ \widetilde{\gamma}_p b_p + \widetilde{\sigma}_p b_{-p}^* \big] - \mu \bigg)^k \; .  \Omega \rangle \; . 
\end{align}
for $k=1,2$. For $k=1$, we find with $b_p \Omega =0$ and the commutation relations \eqref{eq:comm-b} 
\begin{align}
\langle \Omega, \sum_{p \in \Lambda_+^*} \big[ \widetilde{\gamma}_p b_p^* + \widetilde{\sigma}_p b_{-p} \big]  \big[ \widetilde{\gamma}_p b_p + \widetilde{\sigma}_p b_{-p}^* \big]  \Omega \rangle - \mu = \sum_{p \in \Lambda_+^*} \widetilde{\sigma}_p^2 - \mu  = 0 \label{eq:firstmoment}
\end{align}
by \eqref{def:sigma0}. Thus, the term linear in $\lambda$ of the r.h.s. of \eqref{eq:log-mom} vanishes. To compute the quadratic term, we find with \eqref{eq:firstmoment} 
\begin{align}
\label{eq:leftover}
\langle \Omega, & \bigg( \sum_{p \in \Lambda_+^*} \big[ \widetilde{\gamma}_p b_p^* + \widetilde{\sigma}_p b_{-p} \big]  \big[ \widetilde{\gamma}_p b_p + \widetilde{\sigma}_p b_{-p}^* \big] - \mu \bigg)^2 \Omega \rangle \notag \\
&= \langle \Omega, \bigg( \sum_{p \in \Lambda_+^*} \big[ \widetilde{\gamma}_p b_p^* + \widetilde{\sigma}_p b_{-p} \big]  \big[ \widetilde{\gamma}_p b_p + \widetilde{\sigma}_p b_{-p}^* \big]\bigg)^2 \Omega \rangle - \mu^2 \; . 
\end{align}
To compute the remaining expectation value, we observe that expectations of operators in the vacuum vanish, whenever the number of (modified) creation operators does not match the number of (modified) annihilation operators. Thus, the expectation value of \eqref{eq:leftover} in the vacuum reduces with the commutation relations \eqref{eq:comm-b} to 
\begin{align}
\langle \Omega, & \bigg( \sum_{p \in \Lambda_+^*} \big[ \widetilde{\gamma}_p b_p^* + \widetilde{\sigma}_p b_{-p} \big]  \big[ \widetilde{\gamma}_p b_p + \widetilde{\sigma}_p b_{-p}^* \big]- \mu \bigg)^2 \Omega \rangle \notag \\
=& \sum_{p,q \in \Lambda_+^*}    \widetilde{\sigma}_p \widetilde{\gamma}_p \widetilde{\gamma}_q \widetilde{\sigma}_q  \langle \Omega, b_{-p} b_p b_q^* b_{-q}^* \Omega \rangle +\sum_{p,q \in \Lambda_+^*}   \widetilde{\sigma}_p^2  \widetilde{\sigma}_q^2  \langle \Omega, b_{-p} b_{-p}^* b_{-q} b_{-q}^* \Omega \rangle - \mu^2 \notag \\
=& \bigg( 1 - \frac{1}{N}\bigg)  \sum_{p\in \Lambda_+^*}    \widetilde{\sigma}_p^2 \widetilde{\gamma}_p^2  \; . \label{eq:secondmoment}
\end{align} 
Note that   $\sum_{p\in \Lambda_+^*}    \widetilde{\sigma}_p^2 \widetilde{\gamma}_p^2$ is exactly $\sigma^2$ defined in \eqref{def:sigma0}. Thus we have proved that 
$$\lim_{N\to \infty}\mathbb{E}((\cN_+-\mu)^2)= \sigma^2.$$

In conclusion,  by Taylor's theorem and Theorem \ref{thm:main}, there exits $\lambda_0 >0$ such that for all $\lambda < \lambda_0$
\begin{align}\label{eq:log-mom}
\mathbb{E} \big[ e^{\lambda (  \mathcal{N}_+ - \mu) } \big] &= \bigg( 1 + \lambda \langle\psi_N, \big( \mathcal{N}_+ - \mu \big) \psi_N \rangle + \frac{\lambda^2}{2} \langle \psi_N, \big(\mathcal{N}_+ - \mu \big)^2 \psi_N \rangle + \mathcal{O}( \lambda^3 ) \bigg) \nonumber\\
& = 1 +  \frac{\lambda^2}{2} \sigma^2 + \mathcal{O}( \lambda^3 )  + o(1)_{N\to \infty}
\end{align}
for all $\lambda<\lambda_0$. In the last line we have used $ \langle\psi_N, \big( \mathcal{N}_+ - \mu \big) \psi_N \rangle\to 0$ and  $\langle\psi_N, \big( \mathcal{N}_+ - \mu \big)^2 \psi_N \rangle\to \sigma^2$. Thus we have proved \eqref{eq:exp-N+-comp}.


Finally, we observe that by Markov's inequality for any $\lambda >0$
\begin{align}
\mathbb{P} \big[  \mathcal{N}_+ - \mu > x \big] \leq \;  e^{- \lambda x} \mathbb{E} \big[ e^{\lambda (  \mathcal{N}_+ - \mu) } \big], 
\end{align}
leading to 
\begin{align}
\log \mathbb{P} \big[  \mathcal{N}_+ - \mu > x \big] \leq \;  -\lambda x + \log  \mathbb{E} \big[ e^{\lambda (  \mathcal{N}_+ - \mu) } \big]. \label{eq:log-P}
\end{align}
Inserting \eqref{eq:log-mom} in the last estimate and taking the infimum over all $\lambda <\lambda_0$ we obtain \eqref{eq:fund-cor-2}. The bound \eqref{eq:fund-cor-3} follows easily from \eqref{eq:fund-cor-2}.
%
%
%
%
%
%
\end{proof}

\begin{proof}[Proof of Eq. \eqref{eq:LDE-O}] We use similar ideas as for the proof of Corollary \ref{thm:LDE}.To be precise, we find by Markov's ineqquality for all $\lambda>0$, Theorem \ref{thm:main} and Taylor's theorem 
\begin{align}
\log &\mathbb{P} \bigg[ \sum_{i=1}^N O_i - \widetilde{\mu} > x \bigg] \notag \\
\leq& - \lambda x + \log \mathbb{E} \bigg[ e^{\lambda (\sum_{i=1}^N O_i - \widetilde{\mu}_0 )} \bigg] \notag \\
=&\lambda x + \log \bigg( \lambda  \langle \psi_N, \bigg( \sum_{i=1}^N O_i - \widetilde{\mu} \bigg) \psi_N \rangle + \frac{\lambda^2}{2} \langle \psi_N, \bigg( \sum_{i=1}^N O_i - \widetilde{\mu} \bigg)^2 \psi_N \rangle + O( \lambda^3 )\bigg]\label{eq:log-momO}
\end{align}
where we used that $\sum_{i=1}^N O_i = d \Gamma (O)$ for $N$-particle wave functions, the estimate $ \| d \Gamma (O) \psi_N \| \leq \| O \|_{\rm op} \| \mathcal{N}_+ \psi_N \| $ and $[\mathcal{N}_+ , d \Gamma (O)] =0$. We follow the lines of the proof of Theorem \ref{thm:LDE} and thus are left with computing the expectation value 
\begin{align}
\langle \Omega, & \big( \sum_{p,q \in \Lambda_+^*} O_{p,q} \big[ \widetilde{\gamma}_p b_p^* + \widetilde{\sigma}_p b_{-p} \big]  \big[ \widetilde{\gamma}_q b_q + \widetilde{\sigma}_q b_{-q}^* \big]- \widetilde{\mu} \big)^k \Omega \rangle
\end{align}
for $k=1,2$. Similarly as before, we find for $k=1$ 
\begin{align}
\langle \Omega, &  \sum_{p,q \in \Lambda_+^*} O_{p,q} \big[ \widetilde{\gamma}_p b_p^* + \widetilde{\sigma}_p b_{-p} \big]  \big[ \widetilde{\gamma}_q b_q + \widetilde{\sigma}_q b_{-q}^* \big] \Omega \rangle = \sum_{p \in \Lambda_+^*} O_{p,p} \widetilde{\sigma}_p^2 - \widetilde{\mu} = 0 
\end{align}
and 
\begin{align}
& \langle  \Omega,  \bigg( \sum_{p,q \in \Lambda_+^*} O_{p,q} \big[ \widetilde{\gamma}_p b_p^* + \widetilde{\sigma}_p b_{-p} \big]  \big[ \widetilde{\gamma}_q b_q + \widetilde{\sigma}_q b_{-q}^* \big] - \widetilde{\mu} \bigg)^2 \Omega \rangle \notag \\
&= \sum_{p,q,m,n \in \Lambda_+^*} O_{p,q}O_{m,n}  \big[ \widetilde{\sigma}_p \widetilde{\gamma}_q\widetilde{\sigma}_n \widetilde{\gamma}_m \langle \Omega,   b_{-p}  b_q b^*_m b^*_{-n} \Omega \rangle +  \widetilde{\sigma}_p \widetilde{\sigma}_q\widetilde{\sigma}_n \widetilde{\sigma}_m \langle \Omega,   b_{-p}  b_{-q}^* b_{-m} b^*_{-n} \Omega \rangle\big] - \widetilde{\mu}^2 \notag \\
&= \bigg( 1- \frac{1}{N} \bigg) \sum_{p,q \in \Lambda_+^*} \vert O_{p,q} \vert^2 \widetilde{\gamma}_q^2 \widetilde{\sigma}_p^2 + \widetilde{\mu}^2 - \widetilde{\mu}^2 \; . 
\end{align}
Plugging this back into \eqref{eq:log-momO} and optimizing w.r.t. to $0 < \lambda < \lambda_0$, arrive at Eq. \eqref{eq:LDE-O}. 
\end{proof}

\subsection{Proof of Theorem \ref{thm:posT}}

\begin{proof}As a preliminary step, we show that for any positive inverse temperature $\beta = 1/T >0$ the partition function satisfies 
\begin{align}
 c_\beta \leq  e^{\beta E_N} Z(\beta ) :=e^{\beta E_N}\Tr e^{-\beta H_N }  \leq C_\beta \label{eq:Z-bounds}
\end{align}
for positive constants $c_\beta ,C_\beta>0$. We start with the upper bound of \eqref{eq:Z-bounds}. To this end, we write by cyclicity of the trace 
\begin{align}
e^{\beta E_N}Z( \beta) =\Tr e^{-\beta (\mathcal{G}_N - E_N ) }  
\end{align}
with $\mathcal{G}_N $ defined in \eqref{def:G}. By Proposition \ref{prop:G} we find that the partition function is bounded from above by 
\begin{align}
e^{\beta E_N} Z( \beta) \leq e^{ C_1 \beta } \Tr e^{- C_2 \beta\mathcal{H}_N}  \leq e^{ C_1 \beta } \Tr e^{- C_2 \beta\mathcal{K}}  
\end{align}
for positive constants $C_1,C_2 >0$ and for $\mathcal{K}$ given by \eqref{def:HN}. We write the trace in terms of the eigenbasis of $\mathcal{K}$ and find with the exponential laws 
\begin{align}
e^{\beta ( E_N - C_1 )}Z( \beta) \leq \sum_{n_p \in \mathbb{Z}} e^{- C_2  \beta \sum_{p \in \Lambda_+^*} n_p p^2 } = \sum_{n_p \in \mathbb{Z}} \prod_{p \in \Lambda_+^*}\left(  e^{- \beta p^2 } \right)^{n_p} = \prod_{p \in \Lambda_+^*} \frac{1}{1-e^{-C_2 \beta p^2}} \; 
\end{align}
where we concluded by the geometric series in the last step. We proceed with the logarithmic laws 
\begin{align}
\ln e^{\beta E_N} Z( \beta) &\leq \beta C_1  - \sum_{p \in \Lambda_+^*} \ln ( 1- e^{- C_2 \beta p^2} ) \leq \beta C_1 +  C_3 \sum_{p \in \Lambda_+^* } e^{-C_2 \beta p^2} \notag\\ 
&\leq \beta C_1 +  C_3 \sum_{p \in \Lambda_+^* } e^{-C_2 \beta p} = \beta C_1 +  C_3 \frac{1}{ 1- e^{-C_2 \beta} } 
\end{align}
for some positive constant $C_3 >0$ and thus, the upper bound in \eqref{eq:Z-bounds} follows. 

For the lower bound in \eqref{eq:Z-bounds} we remark that it follows from \cite[Prop. 3.2]{BCCS} (with similar arguments as in the proof of Proposition \ref{prop:G}) that 
\begin{align}
\mathcal{G}_N - E_N \leq C_1 \mathcal{H}_N  + C_2 \mathcal{N}_+ \leq  C \mathcal{H}_N
\end{align}
for some constants $C,C_1,C_2 >0$. Moreover, it follows from Sobolev inequality that 
\begin{align}
\mathcal{V}_N \leq C \mathcal{K}^2 
\end{align}
and thus 
\begin{align}
\mathcal{G}_N - E_N \leq C (\mathcal{K}^2 + 1  )  \; . 
\end{align}
Again by cyclicity of the trace, we find in the eigenbasis of $\mathcal{K}$ that 
\begin{align}
e^{\beta E_N - \beta C} Z( \beta) \geq \sum_{n_p \in \mathbb{Z}} e^{- C_2  \beta \sum_{p \in \Lambda_+^*} n_p p^4 } = \sum_{n_p \in \mathbb{Z}} \prod_{p \in \Lambda_+^*}\left(  e^{- \beta p^4 } \right)^{n_p} = \prod_{p \in \Lambda_+^*} \frac{1}{1-e^{-C_2 \beta p^4}} . \; 
\end{align}
We conclude with the logarithmic laws that 
\begin{align}
\ln e^{\beta E_N} Z( \beta) \geq \beta C_1  - \sum_{p \in \Lambda_+^*} \ln ( 1- e^{- C_2 \beta p^2} ) \geq \beta C_1  +   \sum_{p \in \Lambda_+^* } e^{-C_2 \beta p^4}  \geq \beta C_1  + e^{- C_2 \beta} 
\end{align}
and thus the lower bound in \eqref{eq:Z-bounds} follows.

Now, we prove \eqref{eq:N-posT}. Since $\mathcal{U}_N \cN_+ \mathcal{U}_N^* = \cN_+$ we find by cyclicity of the trace and definitions \eqref{def:G}, \eqref{def:wN}  
\begin{align}
e^{\beta E_N}\Tr \left[  e^{-\beta H_N} e^{2\widetilde{\kappa} \cN_+ } \right] =  \Tr \left[  e^{-\beta( \mathcal{G}_N - E_N )} e^{2\widetilde{\kappa} \wN } \right] \; . 
\end{align}
This time, we write the trace in the eigenbasis $\lbrace \xi_j \rbrace_{j \in \mathbb{N}}$ of the excitation Hamiltonian $\mathcal{G}_N$ with corresponding eigenvalues $E_j$ . With these notations we get 
\begin{align}
e^{\beta E_N} \Tr \left[  e^{-\beta H_N} e^{2\widetilde{\kappa} \cN_+ } \right] = \sum_{j \in \mathbb{N}} e^{ - \beta (E_N - E_j )} \langle \xi_j, \; e^{2\widetilde{\kappa} \wN } \xi_j \rangle . 
\end{align}
 With similar arguments as in \eqref{eq:wN-1}-\eqref{eq:wN-2} we find that 
\begin{align}
e^{\beta E_N} \Tr \left[  e^{-\beta H_N} e^{2\widetilde{\kappa} \cN_+ } \right] = \sum_{j \in \mathbb{N}} e^{ - \beta (E_N - E_j ) + C \kappa} \langle \xi_j, \; e^{2 \kappa \cN_+ } \xi_j \rangle \label{eq:T1}
\end{align}
for $\kappa = c_O \widetilde{\kappa}$ and some $c_0, C>0$ and thus it remains to estimate the r.h.s. of \eqref{eq:T1}. 
Similarly as in \eqref{eq:xsi-eins} we define for $s \in [0,1]$
\begin{align}
\xi_j (s) := e^{s \kappa \cN_+} \xi_j 
\end{align}
satisfying $\| \xi_j (1) \| = \langle \xi_j, e^{2 \kappa \cN_+} \xi_j \rangle$ and $\| \xi_j (0) \|^2 = 1$. As in Section \ref{sec:proof-thm} we perform a Gronwall argument and compute 
\begin{align}
\partial_s \| \xi_j (s) \|^2 = \langle \xi_j (s), \; \cN_+ \xi (s) \rangle
\end{align}
Similarly as in \eqref{eq:bound-H1}-\eqref{eq:bound-H-end} we find for sufficiently small $\kappa>0$ with the eigenvalue equation $ ( \mathcal{G}_N - E_N ) \xi_j = ( E_j - E_N ) \xi_j $ that 
\begin{align}
\langle \xi_j (s), \; \cN_+ \xi_j (s) \rangle \leq \langle \xi_j (s), \; \mathcal{H}_N \xi_j (s) \rangle \leq \frac{C}{1- \kappa^2}  (E_j - E_N + 1 ) \| \xi_N (s) \|^2  \; . 
\end{align}
Thus, we arrive with Gronwall's inequality at 
\begin{align}
\langle \xi_j, e^{2 \kappa \cN_+} \xi_j \rangle =  \| \xi_j (1) \|^2 \leq e^{ C ( E_j - E_N + 1)} \| \xi_j (0) \|^2 = e^{ C \kappa ( E_j - E_N + 1)} \; . 
\end{align}
For sufficiently large $\beta >0$ we thus find 
\begin{align}
\langle \xi_j, e^{2 \kappa \cN_+} \xi_j \rangle \leq e^{C \kappa + \beta (E_j - E_N)/2} \; .   \label{eq:T-case2}
\end{align}
Thus, from \eqref{eq:T1} and \eqref{eq:T-case2} we find that 
\begin{align}
\frac{ \Tr \left[  e^{-\beta H_N} e^{2 \widetilde{\kappa} \cN_+ } \right]}{Z( \beta)}  \leq e^{C \kappa} \frac{ e^{\beta E_N/2}Z( \beta /2)  }{ e^{\beta E_N} Z( \beta )} \leq C_\beta e^{C \kappa}
\end{align}
and we conclude with \eqref{eq:Z-bounds}. 
\end{proof}

\subsection*{Data availability} Data sharing not applicable to this article as no datasets were generated or analysed during the current study.

\section*{Declarations}

\subsection*{Conflict of interest} The authors declare that they have no conflict of interest.

\subsection*{Funding} This work was partially funded by the Deutsche Forschungsgemeinschaft via the DFG project Nr. 426365943 (PTN) and by the European Research Council via the ERC CoG RAMBAS  - Project-Nr. 10104424 (PTN and SR).


\begin{thebibliography}{29}


\bibitem{WC-95} M. H. Anderson, J. R. Ensher, M. R. Matthews, C. E. Wieman, and E. A. Cornell. 
Observation of Bose-Einstein condensation in a dilute atomic vapor. {\em Science} 269 
(1995), pp. 198--201.





\bibitem{BasCenSch-21} G. Basti, S. Cenatiempo, and B. Schlein. A new second order upper bound for the
ground state energy of dilute Bose gases. {\em Forum Math. Sigma}  9 (e74) (2021), pp. 1--38.


\bibitem{BCCS_cond} 
C. Boccato, C. Brennecke, S. Cenatiempo and B. Schlein. Complete {B}ose-{E}instein condensation in the {G}ross-{P}itaevskii regime. {\em Commun. Math. Phys.} 359 (2018), pp. 975--1026. 


\bibitem{BDS-23} C. Boccato, A.  Deuchert and D. Stocker. Upper bound for the grand canonical free energy of the Bose gas in the Gross-Pitaevskii limit. SIAM J. Math. Anal. to appear,  arXiv:2305.19173. 

\bibitem{BCCS} C. Boccato, C. Brennecke, S. Cenatiempo and B. Schlein. Bogoliubov Theory in the {G}ross-{P}itaevskii Limit. {\em Acta Mathematica} 222 (2019), pp.  219--335.

\bibitem{BCCS_optimal} 
C. Boccato, C. Brennecke, S. Cenatiempo and B. Schlein. Optimal Rate for Bose-Einstein Condensation in the Gross-Pitaevskii Regime. {\em Commun. Math. Phys.} 376 (2020), pp. 1311--1395. 

\bibitem{BLN} C. Brennecke, J. Lee and P.T. Nam. Second order expansion of Gibbs state reduced densities in the Gross-Pitaevskii regime. {\em SIAM J. Math. Anal.}, in press. (Preprint arXiv:2310.05448).

\bibitem{Bogoliubov-47} N. N. Bogoliubov. On the theory of superfluidity. {\em J. Phys. (USSR)} 11 (1947), pp. 23--32.


\bibitem{Bose-24} S. Bose. Plancks Gesetz und Lichtquantenhypothese. {\em Z. Phys.} 26 (1924), pp. 178--181.

\bibitem{BKS} G. Ben Arous, K. Kirkpatrick, B. Schlein, Central Limit Theorem in Many-Body Quantum Dynamics. {\em Commun. Math. Phys}, 321, pp. 371--417 (2013)

\bibitem{BS} C. Brennecke and B. Schlein. Gross-Pitaevskii Dynamics for Bose-Einstein Condensates. {\em Analysis \& PDE} 12 (2019), pp. 1513-1596.

\bibitem{BSS_optimal_2} 
C. Brennecke, B. Schlein and S. Schraven. Bose-{E}instein Condensation with Optimal Rate for Trapped Bosons in the {G}ross-{P}itaevskii Regime. {\em Math. Phys. Anal. Geom.} 25:12 (2022). 

\bibitem{BSS_bogo}
C. Brennecke, B. Schlein, and S. Schraven. Bogoliubov Theory for Trapped Bosons in the {G}ross-{P}itaevskii Regime. {\em Ann. Henri Poincar\'{e}} 23 (2022), pp. 1583--1658. 


\bibitem{Brooks-2023} M. Brooks. Diagonalizing Bose Gases in the Gross-Pitaevskii Regime and Beyond. {Preprint 2023}.  arXiv:2310.11347. 


\bibitem{BS-2023} M. Brooks and R. Seiringer. The Fr\"ohlich Polaron at Strong Coupling -- Part I: The Quantum Correction to the Classical Energy. {Preprint 2023}. arXiv:2207.03156.

\bibitem{CD-23} M. Caporaletti and A. Deuchert. Upper bound for the grand canonical free energy of the Bose gas in the Gross-Pitaevskii limit for general interaction potentials. Preprint 2023. arXiv:2310.12314. 

\bibitem{COS} C. Caraci, J. Oldenburg, B. Schlein, Quantum Fluctuations of Many-Body Dynamics around the Gross-Pitaevskii Equation, Preprint: arXiv:2308.11687. 

\bibitem{K-95}  K. B. Davis, M. O. Mewes, M. R. Andrews, N. J. van Druten, D. S. Durfee, D. M.
Kurn, and W. Ketterle. Bose-Einstein Condensation in a Gas of Sodium Atoms. {\em Phys.
Rev. Lett.} 75 (1995), pp. 3969--3973.

\bibitem{DS} A. Deuchert and R. Seiringer. 
Gross-Pitaevskii Limit of a Homogeneous Bose Gas at Positive Temperature. 
{\em Arch. Rational Mech. Anal.} 236 (2020), pp. 1217--1271. 



\bibitem{FouSol-20}  S. Fournais and J. P. Solovej. The energy of dilute Bose gases. {\em Ann. of Math.} (2) 192
(2020), pp. 893--976.

\bibitem{FouSol-22} S. Fournais and J. P. Solovej. The energy of dilute Bose gases II: The general case.
{\em Invent. Math.} 232 (2023), 863--994. 



\bibitem{GreSei-13} P. Grech and R. Seiringer. The excitation spectrum for weakly interacting bosons in a trap. {\em Commun.
Math. Phys.} 322 (2013), pp. 559--591. 

\bibitem{HHNST} F. Haberberger, C. Hainzl, P. T. Nam, R. Seiringer, and A. Triay. The free energy of dilute Bose gases at low temperatures. Preprint 2023 (arXiv:2304.02405).

\bibitem{HHSST} F. Haberberger, C. Hainzl, B. Schlein, and A. Triay. Upper Bound for the Free Energy of Dilute Bose Gases at Low Temperature. Preprint 2024 (arXiv:2405.03378). 

\bibitem{HST} C. Hainzl, B. Schlein and A. Triay. Bogoliubov theory in the {G}ross-{P}itaevskii limit: a simplified approach. {\em Forum Math. Sigma} 10 (e10), 2022.

\bibitem{KRS} K.~Kirkpatrick, S.~Rademacher, and B.~Schlein. A large deviation principle for many--body quantum dynamics. {\em Ann. Henri Poincar\'e}  22 (2021), pp. 2595--2618. 

 \bibitem{LeeHuaYan-57} T. D. Lee, K. Huang, and C. N. Yang. Eigenvalues and eigenfunctions of a Bose
system of hard spheres and its low-temperature properties. {\em Phys. Rev.} 106 (6) (1957), pp. 1135--1145.


\bibitem{LS-02} 
E.H. Lieb and R. Seiringer. Proof of Bose-Einstein Condensation for Dilute Trapped Gases. {\em Phys. Rev. Lett.} 88 (2002), p. 170409.


\bibitem{LS} 
E.H. Lieb, and R. Seiringer, Derivation of the {G}ross-{P}itaevskii equation for rotating Bose gases. {\em Commun. Math. Phys.} 264 (2006), pp. 505--537. 




\bibitem{LNSS} 
M. Lewin, P.T. Nam, S. Serfaty and J.P. Solovej. Bogoliubov spectrum of
interacting Bose gases. {\em Comm. Pure Appl. Math.} 68 (2014), pp. 413--471.

\bibitem{M} D. Mitrouskas. {\em Derivation of mean-field equations and next-order corrections for bosons and fermions}. Dissertation, LMU Munich (2017). 

\bibitem{MP-2023} D. Mitrouskas and P. Pickl. Exponential decay of the number of excitations in
the weakly interacting Bose gas. {Preprint:  arXiv:2307.11062 (2023)}. 

\bibitem{Einstein-24} A. Einstein. Quantentheorie des einatomigen idealen Gases.  {\em Sitzungsberichte der
Preu\ss ischen Akademie der Wissenschaften} XXII (1924), pp. 261--267.


\bibitem{Nam-18} P.T. Nam. Binding energy of homogeneous Bose gas. {\em Lett. Math. Phys.} 108 (2018), 141-159.

\bibitem{NamNap-21} P.T. Nam and M. Napi\'{o}rkowski. Two-term expansion of the ground state one-body density matrix of a mean-field Bose gas. {\em Calc. Var. PDE} 60 (2021), Art. 99, 1-30.

\bibitem{NNRT} 
P.T. Nam, M. Napi\'{o}rkowski, J. Ricaud, and A. Triay. Optimal rate of condensation for trapped bosons in the Gross-Pitaevskii regime. {\em Analysis \& PDE} 15 (2022), pp.  1585--1616. 


\bibitem{NRS} 
P.T. Nam, N. Rougerie and  R. Seiringer.  Ground states of large bosonic systems: The {G}ross-{P}itaevskii limit revisited. {\em Analysis \& PDE}  9 (2016), pp. 459--485.


\bibitem{NT} P.T. Nam and  A. Triay. {B}ogoliubov excitation spectrum of trapped Bose gases in the {G}ross-{P}itaevskii regime. {\em J. Math. Pures Appl.}, to appear (arXiv:2106.11949). 

\bibitem{Onsager} L. Onsager, Electrostatic Interaction of Molecules, {\em J. Phys. Chem.}, 43, 2, pp. 189 - 196, (1939) 

\bibitem{PRV} L. Portinale, S. Rademacher, D. Virosztek, Limit theorems for empirical measures of interacting quantum systems in Wasserstein space, Preprint: arXiv:2312.0054. 

\bibitem{Rsing} S. Rademacher, Central limit theorem for Bose gases interacting through singular potentials. \textit{Lett. Math. Phys.} 110, 2143-2174 (2020).

\bibitem{R} S.~Rademacher. Large deviations for the ground state of weakly interacting {B}ose gases. Preprint 2023 (arXiv:2301.00430). 

\bibitem{RS} S. Rademacher, B. Schlein. Central limit theorem for Bose-Einstein condensates. J. of Math. Phys., 60(7):071902 (2019). 

\bibitem{RSe}
S.~Rademacher and R.~Seiringer. Large deviation estimates for weakly interacting bosons. {\em J. Stat. Phys.} 188  (9) (2022). 

\bibitem{Seiringer-11} R. Seiringer. The Excitation Spectrum for Weakly Interacting Bosons. {\em Commun. Math. Phys.} 306 (2011), 565--578.  

\bibitem{YauYin-08} H.-T. Yau and J. Yin. The second order upper bound for the ground energy of a Bose
gas. {\em J. Stat. Phys.} 136 (2009), pp. 453--503.

\end{thebibliography}
\end{document}